\documentclass[lettersize,journal,twoside]{IEEEtran}
\usepackage{cite}
\usepackage{amssymb}
\usepackage{amsfonts}
\usepackage{amsmath}
\usepackage{amsthm}
\usepackage{mathrsfs}
\usepackage{bm}
\usepackage{verbatim}
\usepackage{autobreak} 
\usepackage{graphicx} 
\usepackage{float} 
\usepackage{subfigure} 
\usepackage{tabularx}
\usepackage{algpseudocode} 
\usepackage{autobreak}
\usepackage{fancyhdr}

\usepackage[linesnumbered, ruled]{algorithm2e}
\usepackage{url}
\makeatletter
\renewcommand{\@thesubfigure}{\hskip\subfiglabelskip}
\makeatother
\makeatletter
\renewcommand{\maketag@@@}[1]{\hbox{\m@th\normalsize\normalfont#1}}%
\makeatother
\allowdisplaybreaks[4]
\usepackage[colorlinks,
            linkcolor=black,
            anchorcolor=black,
            citecolor=blue]{hyperref}
\usepackage{hyperref}

\newtheorem{remark}{Corollary}
\newtheorem{proposition}{Proposition}
\begin{document}

\title{ RIS-Aided MIMO Systems with Hardware Impairments: Robust Beamforming Design and Analysis}
\author{
{Jintao Wang, Shiqi Gong, Qingqing Wu, Shaodan Ma}
\thanks{
  Manuscript received August 3, 2022; revised December 23, 2022; accepted February 12, 2023. This work was supported   in part by the Science and Technology Development Fund, Macau SAR under Grants 0087/2022/AFJ and SKL-IOTSC(UM)-2021-2023, in part by the National Natural Science Foundation of China under Grant 62101614 and 62261160650, and in part by the Research Committee of University of Macau under Grant MYRG2020-00095-FST. \emph{(Corresponding authors: Shiqi Gong and Shaodan Ma.)}
}
\thanks{ J. Wang and S.  Ma are  with the State Key Laboratory of Internet of Things for Smart City and the Department of Electrical and Computer Engineering, University of Macau, Macao SAR, China (e-mails: wang.jintao@connect.um.edu.mo; shaodanma@um.edu.mo). }
\thanks{ S. Gong is with the School of Cyberspace Science  and Technology, Beijing Institute of Technology, Beijing 100081, China (e-mail: gsqyx@163.com).}
\thanks{
  Q. Wu is with the Department of Electronic Engineering, Shanghai Jiao Tong University, 200240, China (e-mail: qingqingwu@sjtu.edu.cn)
} 
}

\markboth{IEEE Transactions on Wireless Communications}{Wang \MakeLowercase{\textit{et al.}}: RIS-Aided MIMO Systems with Hardware Impairments: Robust Beamforming Design and Analysis}

\maketitle

\IEEEpubid{\begin{minipage}[t]{\textwidth}\ \\[12pt] \centering
  \copyright \ 2023 IEEE. Personal use of this material is permitted. Permission from IEEE must be obtained for all other uses, in any current or future media, including reprinting/republishing this material for advertising or promotional purposes, creating new collective works, for resale or redistribution to servers or lists, or reuse of any copyrighted component of this work in other works.
\end{minipage}}

\vspace{-12pt}
\begin{abstract}
Reconfigurable intelligent surface (RIS) has been anticipated to be a novel cost-effective technology to improve  the performance of future wireless systems.
In this paper, we investigate {a} practical RIS-aided multiple-input-multiple-output (MIMO)  system in the presence of  transceiver hardware impairments,  RIS phase noise  and imperfect channel state information (CSI).  Joint design of the MIMO  transceiver and  RIS reflection matrix to minimize the total  average  mean-square-error (MSE) of all data streams is particularly considered. 
This joint design problem is non-convex and challenging to solve due to the newly considered practical imperfections.
To  tackle the issue,  
we {first} analyze the total average  MSE  by incorporating the impacts of the  above system imperfections. 
Then,  in order to handle the tightly coupled optimization variables and non-convex NP-hard constraints, an efficient iterative algorithm based on alternating optimization (AO) framework is proposed with guaranteed convergence, where each subproblem admits {a closed-form optimal} solution by leveraging the majorization-minimization (MM) technique. 
{Moreover, via exploiting the  special structure of the  unit-modulus constraints, we  propose a modified Riemannian gradient ascent (RGA) algorithm  for  the discrete  RIS phase shift optimization.}
Furthermore, the optimality of the proposed algorithm is  validated  under  line-of-sight (LoS) channel conditions, and the irreducible MSE floor effect induced by imperfections of both  hardware and CSI  is also revealed in the high {signal-to-noise ratio (SNR)} regime. Numerical results  show the superior MSE  performance  of our proposed algorithm over the adopted benchmark schemes, and demonstrate that increasing {the number of RIS elements} is not always beneficial under the above {system imperfections}.

\end{abstract}

\section{Introduction}

  Recently, reconfigurable intelligent surface (RIS) has attracted much attention in wireless communications due to its ability of reshaping  the wireless propagation environment dynamically. Specifically, the  RIS  is composed of a large number of passive and low-cost reflecting elements, which can adjust the phase of the incident signal towards the target direction with the aid of  a {\color{black}smart controller}  and thus  create the favorable  communication channels.
  Since the RIS does not require {\color{black} costly  radio frequency} (RF) chains, it can be  regarded as a cost-effective solution for improving the spectrum and energy efficiency of future  wireless networks. Inspired by  the high  flexibility of the RIS deployment, its  integration  with  the cutting-edge techniques, such as unmanned aerial vehicle (UAV) communications\cite{UAV2021}, mmWave communications\cite{mmWave2021}, secure communications\cite{secure2021}, wireless power transfer\cite{wpt2021} and so on, has also triggered {\color{black}an upsurge research interest} \cite{wu2021intelligent}.

  Considering the above  significant benefits of the RIS, there have been plenty of works focusing on the  RIS-aided communication systems\cite{wu2019intelligent,transpowermini13,Multicell_2020_Pan,rate4,rateZhangJun,MSE6zhaoxin,MSE_gong,MSE9kaizhe,enereff7,enereff12,gong2020beamforming,Hua2021}. These works can be classified by different design objectives, e.g., the transmit power minimization\cite{wu2019intelligent,transpowermini13}, the rate  maximization\cite{Multicell_2020_Pan,rate4,rateZhangJun}, the mean-square-error (MSE) minimization\cite{MSE6zhaoxin,MSE_gong,MSE9kaizhe}, the energy efficiency maximization\cite{enereff7,enereff12}, the minimum signal-to-interference-plus-noise-ratio (SINR) maximization\cite{gong2020beamforming}. In addition to the above considered narrowband scenarios, the sum-rate maximization of the RIS-aided orthonormal frequency division multiplexing (OFDM) system over the  frequency-selective channels has also been studied \cite{OFDM1,yang2020intelligent,OFDM3}.
  {However, all the above works rely on the use of ideal hardware.}

 Nevertheless, 
 the transmitter and receiver usually  suffer from  non-negligible hardware impairments (HWIs)   in  practice, such as  amplifier nonlinearities, analog-to-digital converters (ADCs) nonlinearities, digital-to-analog converters (DACs) nonlinearities,  {\color{black} in-phase (I) and quadrature (Q) imbalance} and   oscillator phase noise, etc \cite{RF_imperfections_book_2008,JunJuan_2019}. 
 Even taking the compensation algorithms, the residual hardware impairments \cite{residual_transmit_RF_impairments_2010}  still cause the mismatch between the intended signal and the actual radiated signal,  {thereby leading to the degradation of system performance, such as the ergodic capacity, achievable rate, MSE \cite{how_much_HI_affect_2020,khel2021effects,saeidi2021weighted}. Specifically, the authors of \cite{how_much_HI_affect_2020}  analyzed the detrimental effects of transceiver hardware impairments on the ergodic capacity, the outage probability, and the spectral efficiency of the RIS-aided single-input-single-output (SISO) system. Moreover, the authors in \cite{khel2021effects} analyzed the effects of transceiver hardware impairments on the ergodic capacity, showing that transceiver hardware impairments imposes a finite limit on the ergodic capacity, which is unrelated to the number of RIS elements and BS antennas.}
 To alleviate this issue, there have sprung up some  works  aiming  at RIS-aided communication systems with transceiver hardware impairments
\cite{zhou2021secure,hong_shen_beamforming_2021,Spectral_and_Energy_Efficiency_2020}.
{ For example, to ensure the maximum secrecy rate, the authors in \cite{zhou2021secure} studied the robust transmission design for a RIS-aided secure communication system in the presence of transceiver hardware impairments.}
In \cite{hong_shen_beamforming_2021},   the robust transmit precoder and RIS reflection matrix were jointly optimized to maximize the received signal-to-noise-ratio (SNR) of the RIS-aided multi-input-single-output (MISO) system via the generalized Rayleigh quotient and majorization-minimization (MM) algorithm. 
 
 In addition,  taking into account practical finite-resolution phase shifts, the hardware impairments at the  RIS are usually modeled as  RIS phase noise.  
Currently, there are two popular distributions used for modeling  the RIS phase noise, e.g., the uniform and the Von Mises distributions\cite{phaseerror_explain}.
Accordingly,  the joint impacts of  RIS phase noise and transceiver hardware impairments on the RIS-assisted communications {\color{black}were} investigated in \cite{MSEphasenoise2022,liu2020energy,phasenoise2,khel2021effects,chu2022ris,Dai_Jianxin_MU_MISO_2021}.
For example, the authors in \cite{MSEphasenoise2022} modeled the random RIS phase noise following a zero-mean  Von Mises distribution and studied the MSE minimization problem by jointly optimizing the transceiver and RIS reflection matrix, where the closed-form continuous phase shifts are derived using MM technique. 
{ Also, the inevitable phase errors at the RIS affected the diversity order \cite{khel2021effects}, hence it degraded the signal-to-noise-plus-distortion ratio which leads to a significant reduction in the ergodic capacity.
On the other hand, 
the work in \cite{chu2022ris} considered the uniformly distributed RIS phase noise and unveiled the impact of phase shift errors and transmission hardware impairments on the system sum throughput for the RIS-aided wireless-powered Internet of Things (IoT) network, in which wireless energy transfer and information receptions are studied. 
With the same uniformly distributed RIS phase noise,} the authors in \cite{Dai_Jianxin_MU_MISO_2021} aimed to maximize the achievable rate by designing the phase shifts using the condition of statistical channel state information (CSI), where the iterative genetic algorithm (GA) is applied.


%
It is clear that  the aforementioned  works  are all conducted under  the assumption of perfect CSI. In general, perfect CSI is  hard to obtain due  to   channel estimation errors,  channel feedback delays and quantization errors. 
 It is  therefore meaningful  to investigate the joint optimization of the MIMO transceiver and RIS reflection matrix  by incorporating  the  effects of both system hardware impairments and imperfect CSI. 
 To our best knowledge,
  there {\color{black}were} only limited studies considering the RIS-aided wireless networks with both non-negligible  hardware impairments   and imperfect CSI\cite{HIandCSIerror2022,liu2020beamforming,papazafeiropoulos2021intelligent}. 
  {\color{black} Specifically, the authors in \cite{HIandCSIerror2022} studied the robust beamforming design for the transmit power  minimization where the joint impacts of transceiver hardware impairments and statistical CSI errors is considered. 
  However, the impact of RIS phase noise is neglected in this work. Meanwhile,  the works\cite{liu2020beamforming,papazafeiropoulos2021intelligent} investigated the channel estimation methods accounting for both transceiver hardware impairments and RIS phase noise. Armed with the estimated channels, these two works further studied the joint beamforming design for the channel capacity and the achievable sum spectral efficiency, respectively.
  Note that only the single-antenna nodes or users are considered in all these above works. 
  To the best of our knowledge, there has been no literature  investigating the  comprehensive  impacts  of the transceiver hardware impairments,  RIS phase noise  and imperfect CSI  on the RIS-aided MIMO system, which thus motivates this {\color{black}work}.
  }

  In this paper, we   consider {\color{black}an}  RIS-aided  point-to-point MIMO system  in the presence of  transceiver hardware impairments,  RIS phase distortion  and imperfect CSI. It is known that different  performance metrics of  RIS-aided MIMO communications such as  the sum rate%
  , the  total transmit power  and  the energy efficiency 
  have been well optimized. These performance metrics  essentially  represent different trade-offs among the MSEs of multiple data streams\cite{Xing_obj_2021}. As such,  we aim to minimize the total MSE of the considered system by  jointly designing the  MIMO  transceiver and the RIS reflection matrix subject to the transmit power constraint and the discrete unit-modulus constraints at the  RIS. 
  Unfortunately, since the  integration  of the above three types of system imperfections renders the optimization problem  much complicated, most  existing numerical algorithms  adopted by the RIS-relevant works cannot be straightforwardly applied. The main contributions of our work  {\color{black}are summarized} as follows.
 
    \begin{itemize}
      \item This is the first paper to investigate the joint impacts of transceiver hardware impairments, RIS phase noise and imperfect CSI on the RIS-aided point-to-point MIMO systems. We formulate the joint transceiver and RIS reflection matrix design problem as the total average MSE minimization problem, where the analytical expression is derived in the presence of statistical  CSI errors and RIS phase noise. 
      { Different from the existing research focusing on the single-antenna users, our work considers the general MIMO system setup, thereby leading to an intractable matrix-valued optimization problem.  Note that the existing optimization methods are inapplicable to our problem.}
  
      \item { We propose an MM-based algorithm framework to solve the intractable matrix-valued optimization problem with guaranteed convergence.} Specifically, for the transmit precoder design with the intractable objective function, a locally tight lower bound is found using the MM technique. Then, based on Karush-Kuhn-Tucker (KKT) conditions, the optimal transmit precoder is derived in closed-form in each iteration. Similarly, we also apply the MM technique to solve the non-convex discrete phase constraints at the RIS, namely the two-tier MM-based algorithm. In addition, by exploiting the unit-modulus discrete phase constraints, we propose a novel modified Riemannian gradient ascent (RGA) algorithm to obtain the sub-optimal solution of the RIS reflection matrix. 
  
      \item 
      { We reveal that an irreducible MSE floor exists in the high-SNR regime.  Specifically, we derive an explicit expression of the MSE floor and analyze the impacts of the transceiver hardware impairments, the RIS phase noise, and the imperfect CSI on the average MSE in the special case of RIS-aided MISO system, individually. The simulation results demonstrate the impacts of different system imperfections on the system's performance, providing some practical insights for implementing the RIS-aided MIMO system.}
    \end{itemize}

  The remainder of this paper is organized as follows. Section II introduces the system model and problem formulation. The joint robust  MIMO transceiver  and RIS reflection matrix design is  presented in Section III. Section IV discusses the optimality of the proposed MM-based AO algorithm and analyzes the MSE performance  under some special cases.
  Numerical results are shown  in Section V. Finally, Section VI concludes this paper. 

  {\bf Notations:}
  The notation $\mathbb{E}$ represents the expectation on the random variables. 
  ${\mathbb{C}^{ M\! \times \!N}}$ denotes the $ M \times N$ complex space. 
  $\bf{A}^{\star}$, ${\bf{A}}^T$, ${\bf{A}}^H$, ${\bf{A}}^{-1}$ and $\rm Tr({\bf{A}})$ represent the conjugate, transpose, Hermitian, inverse and trace of matrix ${\bf{A}}$, respectively. 
 ${\bf{I}}_{ d}$ denotes a $ d \times d$ identity matrix, and  
  $[{\bf{a}}]_{ i}$  denotes the $ i$-th element of vector ${\bf{a}}$.
   $[{\bf{A}}]_{ ij}$  represents the $(i,j)$-th element  of matrix ${\bf{A}}$.  The notations ${\rm vec}\left( {\bf{A}} \right)$,
$\otimes$  and  $\odot$ denote the matrix vectorization,  Kronecker product and  Hadamard product, respectively.
  ${\rm{ diag}}({\bf{a}})$ indicates a square diagonal matrix whose diagonal elements 
  consist  of
  a vector ${\bf{a}}$.
   ${\rm{\widetilde{diag} }}({\bf{A}})$ represents a square diagonal matrix whose diagonal elements are the same as those  of the matrix ${\bf{A}}$.
  $\Re\{\}$ returns the real part of the complex input, and  
  $| a|$ represents the modulus of the complex input $ a$. The words
   “with respect to” and “ circularly symmetric complex Gaussian” are abbreviated as “w.r.t.” and “CSCG”, respectively. ${\cal{U}}[ a,b] $ denotes the uniform distribution over the interval $ [a,b]$.
   { The main symbols used in this paper are listed in Table I.}
   \vspace{-1mm}
\begin{table*}[!t]  
  \caption{List of Symbols}
  \vspace{-7pt}
  \centering
  \begin{tabular}{|c|c|c|c|}
  \hline
  Symbol                  & Description  & Symbol     & Description  \\ 
  \hline
  $ N_T/N_R $         & Number of BS/user antennas  
  &
  $M/d$               & Number of RIS elements/data streams   \\ 
  \hline
  $ {\bf{x}} $             & Transmit data symbols &
  $ {\bf{s}}/{\bf{y}} $  & Transmitted/Received signal at BS/user\\ 
  \hline 
  $ {\bf{W}}/{\bf{C}} $ & Transmit precoder/Linear equalizer &
  $b$ & Number of quantization bits \\
  \hline
  $ {{\bm{\kappa}}_T}/{{\bm{\kappa}}_R} $  & Transmitter/Receiver distortion noise &
  $ {\beta_T}/{\beta_R} $ & Normalized transmitter's/receiver's distortion level 
    \\
  \hline
  $ {{\bf{H}}_{\rm cas}} $ & Cascaded BS-RIS-user channel &
  $ {{\bf{G}}_{m}} $ & Compound channel associated \\
  & & & with $m$-th RIS element  \\
  \hline
  $ {{\bf{H}}_d}/{{\bf{H}}_I}/{{\bf{H}}_r} $ & BS-user/BS-RIS/RIS-user channel &
  ${{\bf{\bar H}}_d}/{{\bf{\bar G}}_m}$  & Estimated direct/compound channel \\
  \hline
  $ \Delta {\bf{ H}}_d/\Delta {\bf{ G}}_m$ & CSI errors  &
   $ {\sigma_d^2}/{\sigma_m^2}$ & Channel estimation inaccuracy \\
   \hline
   $ {\bm{\Theta}}/{\bm{\theta}} $ & RIS reflection matrix/vector  &
   ${\phi_m}/{\Delta \phi_m}$ & RIS phase shift/distortion \\
   \hline
   $P$ & Transmit power &
   $\sigma^2$ & Noise variance \\
   \hline
    $\epsilon_b$ &Average phase distortion level 
    & ${\bf{\bar H}}_{\bm{\theta}}$ & Available cascaded BS-RIS-user channel \\
   \hline
   ${\bf{T}}^{\rm SI}_{{\bf{W}}}$ & Joint impacts of system imperfections   &   ${\bf{\ H}}_{\rm cat}$ & Concatenated channel \\
   \hline
   ${\bf{T}}^{\rm cas}_{{\bf{W}},\bm{\theta}}$/${\bf{T}}^{\rm com}_{{\bf{W}}}$ & Received signal covariance matrix &  $\nu_I$/$\nu_r$ & Complex channel gain of the LoS path \\
   & associated with ${\bf{\bar H}}_{\bm{\theta}}$/${\bf{\bar G}}_{m}$ & & for the BS-RIS/RIS-user channel    \\ 
   \hline                    
${\bf a}_{TX}$/${\bf a}_{RX}$ & Array response vector at BS/user & ${\bf a}_{RA}$/${\bf a}_{RD}$ & Array response vector of AoA/AoD at RIS  \\ 
\hline          
  \end{tabular}
  \end{table*}

\section{System Model and Problem Formulation}  

\subsection{System Model}
    As shown in  Fig.~\ref{Fig1}, 
    a narrowband RIS-aided MIMO communication system  is considered, where an RIS  equipped with $M$ reflecting elements is deployed to assist in downlink communications from a base station (BS) equipped with $ N_T$ antennas to a user equipped with $N_R$ antennas. 
    Denoted by $ {\bf{x}} \in  {\mathbb{C}^{d \times 1}}$ with $ d \leq {\rm{min}}(N_T,N_R)$  and $\mathbb{E} \left[ {{\bf{x}}{{\bf{x}}^H}} \right] \!=\! {{\bf{I}}_d}$ the {\color{black}transmit} data symbols, the transmit signal at the  BS  considering the realistic hardware impairments
   is then  expressed as \footnote{Most existing works related to hardware impairments aware transceiver designs have mathematically  modeled  the effects of transmitter  hardware distortion,  such as nonlinearities, IQ-imbalance and phase noise, as shown in  (\ref{transmit_signal}) \cite{hong_shen_beamforming_2021}. } 
    \begin{equation}\label{transmit_signal}
        {\bf{s}}={{\bf{W}}}{{\bf{x}}}+ {{\bm{\kappa}}_T}, \vspace{-2mm}
    \end{equation}
    where ${\bf{W}}\in  {\mathbb{C}^{N_T \times d}}$ represents the transmit precoder and ${{\bm{\kappa}}_T} \in \mathbb{C}^{N_T \times 1}$ denotes the transmitter distortion noise, which models the aggregate residual hardware impairments after calibration or pre-distortion\cite{RF_imperfections_book_2008} at the BS and is independent of the data symbols {\bf{x}}.  The elements of ${{\bm{\kappa}}_T} $ are usually assumed to be  CSCG distributed  variables  with zero mean and variance being proportional to the transmit power of each antenna\cite{energy_efficiency_HI_2020_Oluwatayo},
    namely, $ {{\bm{\kappa}}_T} \!\sim\!  {\cal C}{\cal N}( {{\bf{0}}, {\beta _T^2}{\rm{ \widetilde{diag}}}( { {{{\bf{W}}}{{\bf{W}}}^H} } )} )$, where $ {\beta _T}\in[0,1]$ characterizes the normalized distortion level whose value is much less than 1\cite{Capacity_Emil_Bjornson_2013}.  The BS-user channel, the BS-RIS channel and the RIS-user channel are denoted as $ {\bf{H}}_d \in \mathbb{C}^{N_R \times N_T}$, $ {\bf{H}}_I \in \mathbb{C}^{ N_T \times M} $ and $ {\bf{H}}_r \in \mathbb{C}^{N_R \times M}$, respectively. The reflection coefficient of the $m$-th RIS reflecting element  is  given by   ${{\theta}}_{m}=a_m e^{j\phi _m},m\in\!\mathcal{M}\!=\!\{1,2,\cdots,M\}$, 
    where $a_m$ and $\phi _m$ denote the  reflection  amplitude and   phase shift, respectively, and the maximum reflection  amplitude $a_m=1$ is assumed for simplicity. As such,  the cascaded BS-RIS-user channel   can be  expressed   as 
    $ {\bf{H}}_{\rm cas}={\bf{H}}_d+{\bf{H}}_r {\bm{\Theta}} {\bf{H}}_I^{H}$, where 
    $\bm{\Theta} \! =\! \rm diag \left( [ \theta_1, \cdots, \theta_M] \right)$ stands for the RIS reflection matrix.
    {
     Specifically, the direct link between  the BS and the user may be unavailable due to short wavelengths at high frequencies and severe blockage from buildings and trees\cite{WuXianda_millimeter} in the realistic scenarios. As a remedy, the RIS  can be deployed in the high altitudes to intelligently reflect the incident signal to the target user for  assisting in wireless  communications.
     In such scenarios, the RIS-related channels are assumed to have only LoS components.\footnote{ The simplified RIS-related MIMO channel model is adopted to assist in analyzing the optimality of our proposed MM-based AO algorithm in Section III.} 
    }

    \begin{figure}[t] \vspace{-5mm}
      \centering
      \includegraphics[width=0.45\textwidth]{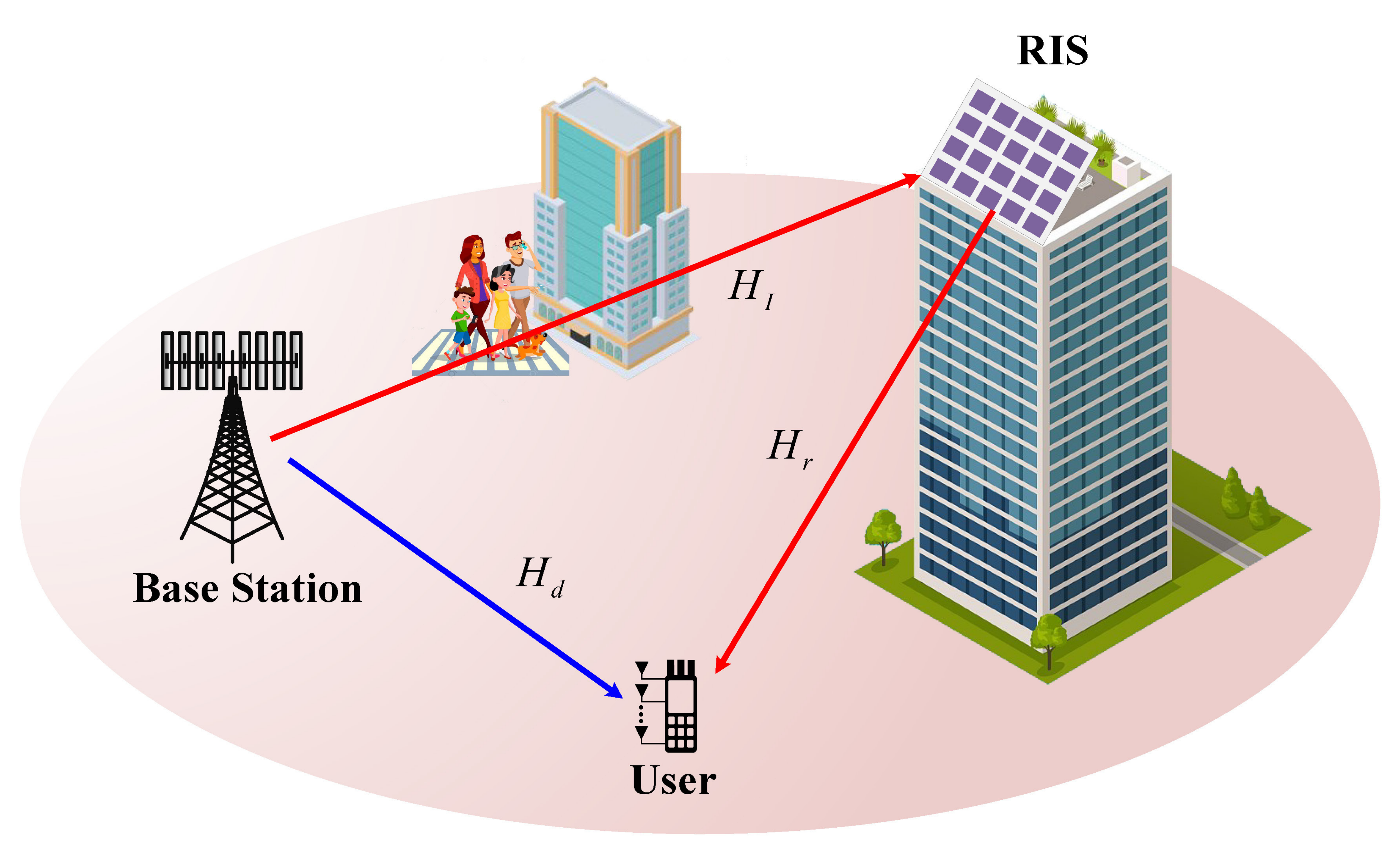}\label{system model}
      \caption{An RIS-aided point-to-point  MIMO communication system.}\label{Fig1}
      \vspace{-5mm}
    \end{figure}

    Due to the passive reflection property of the RIS, the cascaded BS-RIS-user channel estimation   is generally more cost-effective than the separate channel estimation, in which the BS-RIS channel ${\bf{H}}_I$ and the  RIS-user channel ${\bf{H}}_r$  are estimated independently. 
    In terms of our work, we  {\color{black}assume that} channel reciprocity holds, based on which  the  cascaded BS-RIS-user channel estimation is performed 
    at the BS via the uplink pilot transmission. 
    Specifically, we first  divide  the whole  uplink training phase into $M+1$ subphases.
    The ``ON/OFF" mode is adopted by  each RIS reflecting  element by setting $ a_m\!=\!1/0,~\forall m\!\in\! \mathcal{M}$. 
    In the first subphase, the direct channel $ {\bf{H}}_d$ is estimated by setting all the reflecting elements to ``OFF''.
    Then, the compound  channel $ {\bf{G}}_m={\bf{h}}_{r,m}{\bf{h}}_{I,m}^{H}$ associated with the $m$-th RIS reflecting element is obtained by turning the $m$-th RIS reflecting element ``ON" while {\color{black}keeping} the others ``OFF",  where ${\bf{h}}_{r,m}$ and ${\bf{h}}_{I,m}$ denote the $m$-th columns of  ${\bf{H}}_{r}$ and ${\bf{H}}_{I}$, respectively.
    In this context, the cascaded BS-RIS-user channel ${\bf{H}}_{\rm cas}$ is reexpressed as 
    \begin{equation}\label{channel model 1}
      \setlength{\abovedisplayskip}{3pt}
      \setlength{\belowdisplayskip}{3pt}
        {\bf{H}}_{\rm cas}= {\bf{H}}_d+{\sum\nolimits_{i = 1}^M{\theta_m}{\bf{G}}_m}.  
    \end{equation}
    Then, the received signal $\rm \bf{y}$ at the user can be written as 
        \begin{equation}\label{received signal}
          \setlength{\abovedisplayskip}{3pt}
          \setlength{\belowdisplayskip}{3pt}
          {{\bf{y}}}= \underbrace{  {\bf{H}}_{\rm cas} \left({{\bf{W}}}{{\bf{x}}}+ {{\bm{\kappa}}_T} \right)+ {\rm {\bf{n}}}}_{\rm {\bf{\tilde y}}} + {{\bm{\kappa}}_R},            
        \end{equation}
         where $\tilde {\bf{y}}\in {\mathbb{C}^{N_R \times 1}}$ denotes the undistorted received signal and $ {\bf{n}} \in {\mathbb{C}^{N_R \times 1}}$ represents the additive white Gaussian noise (AWGN) drawn from  
         $\mathcal{C} \mathcal{N} \left( {\bf{0}}, \sigma^2 {\bf{I}}\right)$. Analogous to the definition of ${{\bm{\kappa}}_T}$, 
         ${{\bm{\kappa}}_R} \sim  {\cal C}{\cal N}  ( {\bf{0}},{  \beta _R^2}  {\rm{ \widetilde{diag}}} \left( {\mathbb{E} \left[ {{\bf{\tilde y}}{\bf{\tilde y}}^H} \right] } \right) )$ denotes the CSCG distributed receiver distortion noise, 
          where $ { \beta _R}$ indicates the distortion level at the user and 
          \begin{subequations}
            \begin{alignat}{2}      
              {  \beta _R^2}  {\rm diag} \left({\mathbb{E}\rm  \left[ {{\bf{\tilde y}}{\bf{\tilde y}}^H} \right] }\right) 
              & \mathop  = \limits^{({a_1})}  {  \beta _R^2}  {\rm{ \widetilde{diag}}} \big(  {\bf{H}}_{\rm cas} {\bf{W}}{\bf{W}}^H {\bf{H}}_{\rm cas}^H +\sigma ^2{{\bf{I}}_{{N_R}}}  \nonumber \\
              & \qquad  +{\beta _T^2} {\bf{H}}_{\rm cas} {\rm{ diag}}\left( { {{{\bf{W}}}{{\bf{W}}}^H} } \right){\bf{H}}_{\rm cas}^H \big)
              \nonumber \\
              & \mathop \approx \limits^{({a_2})} {  \beta _R^2}  {\rm{ \widetilde{diag}}} \left(  {\bf{H}}_{\rm cas} {\bf{W}}{\bf{W}}^H {\bf{H}}_{\rm cas}^H +\sigma ^2{{\bf{I}}_{{N_R}}}\right). \tag{4}
            \end{alignat} 
          \end{subequations}
             The equality $({a_1})$ holds based on \eqref{received signal} and $({a_2})$ holds since the term $\rm { \beta _R^2} {\beta _T^2}$ is sufficiently small.

             {{\textit{Remark:}} The training overhead and estimation complexity for this ``ON/OFF'' channel estimation scheme could be considerably reduced by leveraging the RIS elements grouping method \cite{youIntelligentReflectingSurface2020,yang2020intelligent}. Specifically, since the RIS elements are usually densely deployed, the channels of the adjacent RIS reflection elements are generally correlated. By grouping the adjacent RIS elements into a block and applying the same reflection pattern to each block element, the number of estimated channels could be decreased quadratically. Such grouping method is suitable for practical implementation.  
             Inspired by that, the compound channels of each group are estimated sequentially with the pilot transmission by turning on the RIS elements of the corresponding group, significantly reducing the training overhead.}
  \subsection{Realistic Channel Modeling }
  It follows from \eqref{channel model 1} that the cascaded BS-RIS-user channel consists of three parts: the direct channel $ {\bf{H}}_d$, the compound channels $ {\bf{G}}_m$'s, and the RIS reflection coefficients $ \theta_{m}$'s. In realistic communication systems, due to the presence of  the non-negligible hardware impairments, channel estimation errors and feedback delays, the accurate cascaded BS-RIS-user channel  ${\bf{H}}_{\rm cas}$ is   difficult to obtain. That is to say,  the  perfect CSI  knowledge of the  channels $ {\bf{H}}_d$ and $ {\bf{G}}_m $'s at the BS is hard to realize. Motivated by this fact, in this paper,  we consider the imperfect CSI model composed of the channel estimate and statistical CSI errors{\cite{wang2014outage,zhang2020robust,secure2021}}. Accordingly, the actual  channels  $\{{\bf{H}}_d,{\bf{G}}_m\}$ can be  modeled  as  
 \begin{equation} \label{channel model 2}
   {\bf{H}}_{d} =  {\bf{\bar H}}_{d} + \Delta {\bf{H}}_{d} ,~~ \ {\bf{G}}_{m} = {\bf{\bar G}}_m + \Delta {\bf{ G}}_m, 
\end{equation} 
 where $ {\bf{\bar H}}_{d}$ and $ {\bf{\bar G}}_m$ denote the estimated direct channel and the estimated compound  channel associated with the $m$-th RIS reflecting element, respectively.  $\Delta {\bf{ H}}_{d}$ and $\Delta {\bf{ G}}_m$ are the corresponding CSI errors, respectively. 
 ${\rm vec}( \Delta {\bf{ H}}_{d})$ $({\rm  vec} ( \Delta {\bf{ G}}_m ))$ is assumed to follow the CSCG distribution with zero mean and covariance matrix $\sigma_{d}^2 {\bf{I}} $ ($\sigma_{m}^2 \bf{I}$),  i.e.
  $ {\rm vec}( \Delta {\bf{ H}}_{d} ) \sim \mathcal{C} \mathcal{N} ( {\bf{0}},\sigma_{d}^2 {\bf{I}} )$ and ${\rm  vec} ( \Delta {\bf{ G}}_m ) \sim \mathcal{C} \mathcal{N} ( {\bf{0}}, \sigma_{m}^2 \bf{I} ) $, 
  where $\sigma_{d}^2$ and $    \sigma_{m}^2$ indicate the  estimation inaccuracy of ${\bf{ H}}_{\rm d} $ and ${\bf{ G}}_m $, respectively.

     Additionally, since
      the RIS employs low-resolution, low-cost phase shifts for   practical implementation,  the phase shift $\phi_m$ at the $m$-th RIS reflecting element  usually takes the discrete value belonging to the set $\mathcal{F}\triangleq \{-\pi,{\frac{2\pi} {2^{\rm b}}}-\pi,...,{\frac {{2\pi}({2^{\rm b}}-1)}  {2^{\rm b}}}-\pi \} $, where $b$ represents the number of  quantization bits. Nevertheless,
     due to the presence of    quantization errors induced by the discrete RIS phase shifts, the RIS phase distortion  can be modeled as $\bigtriangleup \phi_m \sim {\cal{U}}[{\frac {-\pi}{2^{\rm b}}},{\frac {\pi} {2^{\rm b}}}] $. Consequently, the actual phase over the $m$-th RIS reflecting element is  given by \vspace{-1.5mm}
 \begin{equation} \label{RIS_distortion}
 \phi_m^{\rm act}= {\phi}_m+\bigtriangleup \phi_m, 
\end{equation} 
 where  $\rm \bigtriangleup \phi_m$ denotes the phase distortion of the $m$-th RIS reflecting  element. Accordingly, the actual  RIS reflection coefficient $\theta_m^{\rm act}$ is modeled as  $\theta_m^{\rm act}= {\theta}_m e^{j \Delta\phi _m}$. 

\subsection{Problem Formulation}
At the user side, the linear equalizer ${\bf{C}} \!\in\! {\mathbb{C}^{d \times N_R}}$ is applied to obtain the estimated data symbol vector $ {\bf{\hat{x} }}={\bf{C}} {\bf{y}}$. Then, the resultant MSE matrix is  divided into the following two parts  as defined in \eqref{MSE} at the top of the next page,
\begin{figure*}[t] 
        \vspace{-2mm}
        \begin{align}\label{MSE}
           {\bf{MSE}}= \mathbb{E} \left[ \left({\bf{\hat{x} }}-{\bf{x}}\right) \left({\bf{\hat{x} }}-{\bf{x}}\right)^H  \right] 
          =& \underbrace{ { {\bf{I}}_d} - {\bf{C}}{\bf{H}}_{\rm cas}{\bf{W}} - {{\bf{W}}^H}{{\bf{H}}_{\rm cas}^H}{{\bf{C}}^H}  + {\bf{C}}\left( {{\bf{H}}_{\rm cas}{\bf{W}}{{\bf{W}}^H}{{\bf{H}}_{\rm cas}^H} + {\sigma ^2}{\bf{I}}} \right){{\bf{C}}^H} }_{{\bf{E  }}^{\rm ideal}} \nonumber \\
          &+\underbrace{{\bf{C}} \big( 
            {\beta _T^2} {\bf{H}}_{\rm cas}{\rm{ \widetilde{diag}}}\left( { {{{\bf{W}}}{{\bf{W}}}^H} } \right)  {{\bf{H}}_{\rm cas}^H} +
            {  \beta _R^2}  {\rm{ \widetilde{diag}}} \left(  {\bf{H}}_{\rm cas} {\bf{W}}{\bf{W}}^H {\bf{H}}_{\rm cas}^H \right) +
             \sigma ^2 { \beta _R^2} {{\bf{I}}_{{N_R}}}
          \big)   {{\bf{C}}^H}}_{{\bf{E  }}^{\rm add}},
        \end{align}
        \vspace{-2mm}
        \hrulefill
  \end{figure*} 
      where ${\bf{E  }}^{\rm ideal}$ denotes the MSE matrix  of the ideal RIS-aided  MIMO system,  and ${\bf{E  }}^{\rm add}$ is an  additional MSE matrix introduced by  transceiver hardware impairments. 
    By recalling the  statistical CSI  errors in \eqref{channel model 2} and  the RIS phase distortion in \eqref{RIS_distortion}, 
      the total average MSE of the RIS-aided MIMO system is  derived as
      \begin{align} \label{obj1}
        f_{\rm MSE}({\bf{C}},{\bf{W}},{\bm{\theta}})= & \mathbb{E}_{{\bf{H}},{\bm{\theta}}}  \left[ \rm  Tr \left( {\bf{MSE}} \right) \right] \nonumber \\ 
         = & {\rm Tr} \big( {{\bf{I}}_d} \!-\! {\bf{C{\bar H_{\bm{\theta}}}W}} \!-\! {{\bf{W}}^H}{{{\bf{\bar H}_{\bm{\theta}}}}^H}{{\bf{C}}^H} \nonumber \\
        & 
        + {\bf{C}} {(  {\bf{T}}^{\rm cas}_{{\bf{W}},\bm{\theta}} \!+\! \epsilon_b {\bf{T}}^{\rm com}_{{\bf{W}}} +{\bf{T}}^{\rm SI}_{{\bf{W}}} ) 
      }{{\bf{C}}^H} \big) ,
      \end{align}
      \begin{figure*}[t] \vspace{-2mm}
        \begin{subequations}\label{TW}
          \setlength{\abovedisplayskip}{3pt}
          \setlength{\belowdisplayskip}{3pt}
        \begin{align}
            &{\epsilon_b ={1-{{{4^b}} \over {{\pi }^2}}  {{\sin }^2}  ( {{\pi  \over {{2^b}}}} )}},  \label{9a} \\
            &{\bf{\bar H}}_{\bm{\theta}} = \ {{{\bf{\bar H}}}_d} + {{{2^b}} \over \pi }\sin ( {{\pi  \over {{2^b}}}} )\sum\nolimits_{m = 1}^M {{\theta}_m{{{\bf{\bar G}}}_m}},  \label{9b} \\
              &{\bf{T}}^{\rm cas}_{{\bf{W}},\bm{\theta}} =
               {{\bf{{\bar H_{\bm{\theta}}}W}}{{\bf{W}}^H}{{{\bf{\bar H_{\bm{\theta}}}}}^H} }  +
            {\beta _T^2}{\bf{\bar H_{\bm{\theta}}}} {\rm{ \widetilde{diag}}} ({\bf{W}}{\bf{W}}^H) {\bf{\bar H_{\bm{\theta}}}}^H +
              {\beta _{R}^2}{\rm{ \widetilde{diag}}}({\bf{\bar H_{\bm{\theta}}}} {\bf{W}}{\bf{W}}^H {\bf{\bar H_{\bm{\theta}}}} ),   \label{9c} \\
              &{\bf{T}}^{\rm com}_{{\bf{W}}} =  
              {\sum\nolimits_{m = 1}^M  \left( {{{{\bf{\bar G}}}_m}{\bf{W}}{{\bf{W}}^H}{{{\bf{\bar G}}}_m}^H}  + {\beta_T^2} {{{{\bf{\bar G}}}_m}{\rm{ \widetilde{diag}}}( {{\bf{W}}{{\bf{W}}^H}} ){{{\bf{\bar G}}}_m}^H}  + {\beta _R^2} {{\rm{ \widetilde{diag}}}( {{{{\bf{\bar G}}}_m}{\bf{W}}{{\bf{W}}^H}{{{\bf{\bar G}}}_m}^H} ) } \right)  },    \label{9d}  \\
                &{\bf{T}}^{\rm SI}_{{\bf{W}}} =
               \left(1 \! + \! {\beta_{R}^2}\right)\sigma^2{{\bf{I}}_{{N_R}}}+\left( {1 + {\beta _T^2} + {\beta_R^2}} \right){\rm Tr}\left( {{\bf{W}}{{\bf{W}}^H}} \right)\big( {{\sigma}_d^2 + \sum\nolimits_{i = 1}^M {{\sigma}_m^2} } \big){{\bf{I}}_{{N_R}}}. \label{9e} 
        \end{align}        \vspace{-3pt}
        \vspace{-3pt}
        \end{subequations}
            \vspace{-2mm}
              \hrulefill
        \end{figure*}
      \noindent where the involved  auxiliary variables are defined in \eqref{TW}  at the top of the next page. 
      {\color{black}
      The term \eqref{9a} represents the average phase distortion level at the RIS,  
     and  \eqref{9b}  denotes the available  cascaded BS-RIS-user channel  based on  the estimated channels $\bar{\bf H}_d$ and $\bar{\bf G}_m$'s and the  discrete (quantized)  RIS  phase shifts $\theta_m$'s.} The equations
      \eqref{9c}  and   \eqref{9d}  represent the received signal covariance matrices  associated with the available cascaded BS-RIS-user channel ${\bf{\bar H}}_{\bm{\theta}}$  and the estimated compound  channels ${{\bf{\bar G}}}_m$'s, respectively. Finally, the {\color{black}term in} \eqref{9e}  models  the  joint impacts of  system imperfections, such as transceiver hardware  distortion and  CSI  errors, on the   total average MSE.

In this paper, we aim to  minimize the  total  average MSE $f_{\rm MSE}({\bf{C}},{\bf{W}},{\bm{\theta}}) $   by jointly designing the active transceiver and passive RIS beamforming under both the transmit power constraint and the  discrete  RIS phase shift constraints, which is formulated as
 \begin{subequations}\label{Problem P1}
 \begin{align}
 (\text{P1}): \ \min_{{\bf{C}},{\bf{W}},{\bm{\theta}}} \quad  \!\!&  f_{\rm MSE}({\bf{C}},{\bf{W}},{\bm{\theta}})  \\
\!\!\mbox{s.t.} \quad 
      &{ {(1+{\beta _T^2})}{\rm  Tr}\left( {\bf{W}}{\bf{W}}^H \right) \leq P, \label{P1C1} } \\
 & \phi_m \in \mathcal{F},~\forall m \in \mathcal{M}, \label{P1C2} 
 \end{align} 
\end{subequations}
where  $P$ denotes the maximum transmit power. In general, problem (P1) is jointly non-convex  w.r.t.  $\{ {\bf{C}},{\bf{W}},{\bm{\theta}} \}$, since  the optimization variables  are tightly coupled in the objective function and the discrete unit-modulus  phase constraints are  also non-convex.  It is worth noting that most  existing joint active and passive beamforming designs in the RIS-aided point-to-point  MIMO system with perfect CSI are not straightforwardly  applicable  on account of the existence of ${\bf{T}}^{\rm cas}_{{\bf{W}},\bm{\theta}}$ and ${\bf{T}}^{\rm com}_{{\bf{W}}}$, both  of  which are  non-convex  matrix-valued functions of $\{{\bf{W}},{\bm{\theta}}\}$. In the following section,  we propose an efficient AO algorithm by leveraging the philosophy of MM technique to 
 overcome the intractability of problem (P1).

\section{MM-based  Alternating Optimization Algorithm}
  In this section, with the aid of MM technique,  we propose an  alternating optimization algorithm  to solve the non-convex  problem (P1). 
    Firstly, for any  given $\{ {\bf{W}}, {\bm{\theta}} \}$, we  find that problem (P1) is convex on   the unconstrained  linear equalizer $\bf{C}$.  
    {Inspired by this fact, the optimal ${\bf{C}}^{\star}$ can be readily derived by setting the first derivative of the objective function w.r.t $\bf{C}$ to 0, which is known as the well-known Wiener filter, i.e.,} 
\begin{equation}
 {\bf{C^\star}}={{\bf{W}}^H}{{{\bf{\bar H}_{\bm{\theta}}}}^H}{(  {\bf{T}}^{\rm cas}_{{\bf{W}},\bm{\theta}} \!+\! \epsilon_b {\bf{T}}^{\rm com}_{{\bf{W}}} +{\bf{T}}^{\rm SI}_{{\bf{W}}} )}^{-1}.
\end{equation}         
Substituting $\bf{C^\star}$ into the objective function in \eqref{Problem P1}, we can rewrite problem (P1)  as
\begin{subequations}
 \begin{alignat}{2}\label{P2}
 (\text{P2}):\  \max_{{\bf{W}},{\bm{\theta}}} \quad  & g_{\rm MSE}({\bf{W}},{\bm{\theta}}) \!=\!{\rm Tr}\big(  {{\bf W}^H}{\bf{\bar H}}_{\bm{\theta}}^H \big(  {\bf{T}}^{\rm cas}_{{\bf{W}},\bm{\theta}} \!+\! \epsilon_b {\bf{T}}^{\rm com}_{{\bf{W}}} \nonumber \\
 & \qquad \qquad \qquad \quad +{\bf{T}}^{\rm SI}_{{\bf{W}}} \big)^{-1}{\bf{\bar H}}_{\bm{\theta}}{{\bf W}}    \big),~~~\\
  \mbox{s.t.} \quad 
 & \eqref{P1C1},~ \eqref{P1C2},
\end{alignat}
\end{subequations}
  Unfortunately,   problem (P2) is  
    still jointly non-convex w.r.t. ${\bf W}$ and ${\bm \theta}$  owing to the complicated objective function
    with strong variable coupling, which thus motivates us to  develop an efficient AO algorithm  to solve it.
 In the AO procedure,   the MM technique can be applied to   find the  optimal closed-form solutions of  both  sub-problems in terms of  the transmit precoder ${\bf W}$ and the passive RIS reflection vector ${\bm \theta}$,    as elaborated below.
 
  \subsection{Optimization of the Transmit Precoder}
    In this subsection, we resort to   optimize the transmit precoder   $\bf{W}$ while keeping the passive RIS beamforming $\bm{\theta}$ fixed.  
    By defining $\rm {\bf{X}}_{{\bf W}}={\bf{{\bar H_{\bm{\theta}}}W}}$ and $
    \rm {\bf{Y}}_{{\bf W}}= {\bf{T}}^{\rm cas}_{{\bf{W}},\bm{\theta}} \!+\! \epsilon_b {\bf{T}}^{\rm com}_{{\bf{W}}} +{\bf{T}}^{\rm SI}_{{\bf{W}}}$,  
    the objective function of  problem (P2) turns into $g_{\rm sub1}({{\bf W}})=\rm Tr\big({\bf{X}}_{{\bf W}}{\bf{Y}}_{{\bf W}}^{-1}{\bf{X}}_{{\bf W}}^H\big)$,
    which is jointly convex w.r.t. ${\bf{X}}_{{\bf W}} $ and ${\bf{Y}}_{{\bf W}}$. Furthermore, 
    armed with  the MM technique, a surrogate function (locally tight lower bound) of $g_{\rm sub1}{({\bf W})}$ at a given point $\left\{{\bf{X}}_{{\bf W}_{t}},{\bf{Y}}_{{\bf W}_{t}} \right\}$ can be constructed  as
  \begin{subequations}
  \begin{align}\label{MMsub}
    &{g_{\rm sub1}^{\rm Low}({\bf{W}};{\bf W}_{t})} \nonumber \\
    &\!\overset{(b_1)} {=} \!\!{2\Re\left\{ {\rm Tr} ( {\bf{X}}_{{\bf W}_{t}}^H {\bf{Y}}_{{\bf W}_{t}}^{-1} {\bf{X}}_{{\bf W}} ) \right\}\!\!-\!\!{\rm Tr} ( {\bf{Y}}_{{\bf W}_{t}}^{-1} {\bf{X}}_{{\bf W}_{t}}{\bf{X}}_{{\bf W}_{t}}^{H} {\bf{Y}}_{{\bf W}_{t}}^{-1} {\bf{Y}}_{{\bf W}}) }\nonumber\\
    &\!\overset{(b_2)}{=}\! 2\Re \left\{ {\rm Tr} \left( {\bf{X}}_{{\bf W}_{t}}^H {\bf{Y}}_{{\bf W}_{t}}^{-1} {\bf{{\bar H}_{\bm{\theta}}W}}   \right) \right\}\!-\!{\rm Tr} \big( {\bf{Z}}_{{\bf W}_{t}}   {\bf{W }}{\bf{W}}^{H}\big) 
    , \tag{13}
  \end{align} 
\end{subequations}

\begin{figure*}[!t] 
  \begin{subequations}\label{MM_eq1}
  \begin{align}
    & {\bf{Z}}_{{\bf W}_{t}} =  {\bf{Z}}_{{\bf W}_{t}}^{\rm cas} +\epsilon_b {\bf{Z}}_{{\bf W}_{t}}^{\rm com} + {\bf{Z}}_{{\bf W}_{t}}^{\rm SI},  \tag{16a} \\
    & {\bf{Z}}_{{\bf W}_{t}}^{\rm cas} ={{{\bf{\bar H_{\bm{\theta}}}}}^H} {\bf{\hat W}}_{t}^H {\bf{\hat W}}_{t} {\bf{{\bar H_{\bm{\theta}}}}}+
     \beta_T^2 {\rm{ \widetilde{diag}}} ( {{{\bf{\bar H_{\bm{\theta}}}}}^H} {\bf{\hat W}}_{t}^H {\bf{\hat W}}_{t} {\bf{{\bar H_{\bm{\theta}}}}} )+
      \beta_R^2 {{{\bf{\bar H_{\bm{\theta}}}}}^H} {\rm{ \widetilde{diag}}} ({\bf{\hat W}}_{t}^H {\bf{\hat W}}_{t} ){\bf{{\bar H_{\bm{\theta}}}}}, \tag{16b} \\                
     & {\bf{Z}}_{{\bf W}_{t}}^{\rm SI} = \left( {1 + {\beta _T^2} + {\beta _R^2}} \right)( {{\sigma _d}^2 + \sum\nolimits_{i = 1}^M {{\sigma _m}^2} } ) {\rm Tr} ( {\bf{\hat W}}_{t}^H {\bf{\hat W}}_{t} ) {{\bf{I}}_{{N_T}}}, \tag{16c} \\  
     & {\bf{Z}}_{{\bf W}_{t}}^{\rm com}=  \sum\nolimits_{m = 1}^M (  {{{{\bf{\bar G}}}_m}^H {\bf{\hat W}}_{t}^H {\bf{\hat W}}_{t} {{{\bf{\bar G}}}_m}}  \!+\! 
     {\beta _T^2} {{{{\bf{\bar G}}}_m}^H{\rm{ \widetilde{diag}}}( {\bf{\hat W}}_{t}^H {\bf{\hat W}}_{t} ){{{\bf{\bar G}}}_m}} \!+\! 
     {\beta _R^2} {{\rm{ \widetilde{diag}}}( {{{{\bf{\bar G}}}_m}^H {\bf{\hat W}}_{t}^H {\bf{\hat W}}_{t} {{{\bf{\bar G}}}_m}} ) } ). \tag{16d}
    \end{align}
    \vspace{-3pt}
  \end{subequations}
      \vspace{-2mm}
        \hrulefill
  \end{figure*}

    \noindent where ${\bf W}_{t}$ denotes  the current iterate in the $t$-th  AO iteration. 
    The involved auxiliary  matrix ${\bf{Z}}_{{\bf W}_{t}} $ are defined in (16) at the top of the next page by setting ${\bf{\hat W}}_{t}= {\bf{X}}_{{\bf W}_{t}}^H {\bf{Y}}_{{\bf W}_{t}}^{-1}$.  
    The  equality $(b_1)$  is due to the first-order Taylor expansion,
     and  the equality $(b_2)$ holds based on  the definitions of  ${\bf{X}}_{{\bf W}}$ and ${\bf{Y}}_{{\bf W}}$ together with the matrix identity $\rm Tr ( {\bf{A}} {\rm{ \widetilde{diag}}}{(\bf{B})} )=Tr (  {\rm{ \widetilde{diag}}}{(\bf{A})}{\bf{B}} )$.
    Then, according to the MM principle, the subproblem w.r.t. the transmit precoder ${\bf{W}}$  can be  expressed as 
    \begin{subequations}
      \begin{alignat}{2}
        (\text{P3}): \ \ \max_{{\bf{W}}} \quad &{g_{\rm sub1}^{\rm Low}({\bf{W}};{\bf W}_{t})},~~
        \mbox{s.t.} \quad 
        &\eqref{P1C1}. \tag{14}
      \end{alignat}
    \end{subequations}

  \noindent We notice that problem (P3) is convex in $\bf{W}$ and thus the KKT  conditions can be applied to find the global optimal solution to  this subproblem. Specifically, the Lagrangian function of problem (P3) is derived as 
    \begin{subequations}
      \setlength{\abovedisplayskip}{3pt}
      \setlength{\belowdisplayskip}{3pt}
    \begin{align}
     L( {{\bf{W}},\lambda } ) = &2\Re\left\{ {\rm Tr} \left( {\bf{X}}_{{\bf W}_{t}}^H {\bf{Y}}_{{\bf W}_{t}}^{-1} {\bf{{\bar H}_{\bm{\theta}}W}}  \right)\right\}\! \nonumber \\
     &-\!{\rm Tr} \left(\big( {\bf{Z}}_{{\bf{W}}_{t}}\!+\!\lambda {\bf{I}}_{N_T} \big) {\bf{W }}{\bf{W}}^H  \right)\!+\!  \frac{\lambda P}{1+{\beta _T^2}}, \tag{15}
    \end{align}
    \end{subequations}

    \noindent where $\lambda$ denotes the Lagrangian multiplier associated with the total power  constraint.
    By setting the first derivative of  the Lagrangian function $L\left( {{\bf{W}},\lambda } \right)$ w.r.t ${\bf{W}} $ to 0, the  optimal  $\bf{W}^{\star}$  can be derived in the following closed  form 
        \begin{equation}\label{update_W}
          {\bf{W}}^{\star}=\left( {\bf{Z}}_{{\bf{W}}_{t}}+\lambda {\bf{I}}_{N_T}\right)^{-1}{{{\bf{\bar H_{\bm{\theta}}}}}^H}  {\bf{Y}}_{{\bf W}_{t}}^{-1} {\bf{X}}_{{\bf W}_{t}},
        \end{equation}
where the optimal $ \lambda$  should be  chosen to satisfy the complementary slackness  condition, i.e., $\lambda((1+{\beta _T^2}){\rm Tr}( {\bf{W}}{\bf{W}}^H )\!-\!P)\!=\!0$.  Note that if ${\rm Tr}( {\bf{W}}{\bf{W}}^H ) \leq \frac{P}{1+{\beta _T^2}} $, {\color{black}it follows that the optimal  $ \lambda^{\star }\!=\!0$}. Otherwise, we have
\begin{equation}\label{lambda}
  \sum\limits_{i = 1}^{N_T} \frac{[{{\bf{U}}_Z {{{\bf{\bar H}}}_{\bm{\theta}}^H} {\bf{Y}}_{{\bf W}_{t}}^{-1} {\bf{X}}_{{\bf W}_{t}}  {\bf{X}}_{{\bf W}_{t}}^H {\bf{Y}}_{{\bf W}_{t}}^{-1} {\bf{\bar H}}_{\bm{\theta}} {{\bf{U}}^H_Z} }]_{i,i}}
  {\left( [{\bm \Lambda_Z}]_{i,i} + \lambda  \right)^2} \!=\!\frac{P}{1+{\beta _T^2}},
\end{equation}
where the  unitary matrix ${{\bf{U}}_Z}$  and the diagonal matrix ${{\bf \Lambda}}_Z$ come from the eigenvalue decomposition (EVD) of  ${\bf{Z}}_{{\bf{W}}_t}$, i.e. $ {\bf{Z}}_{{\bf{W}}_t}={{\bf{U}}_Z} {{\bf \Lambda}_Z} {{\bf{U}}^H_Z}$.
It is easily inferred that the left-hand side of \eqref{lambda} is a monotonically non-increasing function w.r.t. $\lambda$. Hence,  the optimal $\lambda^{\star}$ can be uniquely determined by the one dimensional search method, e.g., bisection method \cite{Multicell_2020_Pan}. 

\subsection{Optimization of the Passive RIS }
    In this subsection, from the perspective of low complexity,  
    we propose two different  algorithms for  the RIS reflection vector design   given the transmit precoder ${\bf W}$.  
    The first algorithm still adopts the MM technique to find the closed-form optimal ${\bm{\theta}}$,  similarly to the transmit precoder optimization. The second algorithm 
   modifies the traditional  RGA algorithm  so that it is suitable for the discrete RIS phase shift design.   Both the two algorithms are  guaranteed to converge to a finite objective value.
  
   \subsubsection{Two-tier MM-based Algorithm}
   Firstly, by rewriting  ${\bf{\bar H}}_{ \bm{\theta} }$ in \eqref{9b} as $
    {\bf{\bar H}}_{ \bm{\theta} }\!=\!{\bf{\ H}}_{\rm cat}({ \bm{\tilde \theta}} \otimes{\bf{I}}_{N_T})$, 
    where $\rm { \bm{\tilde \theta}}\!=\!\left[ 1 \  {{{\bm{\theta }}^H}} \right]^H$ and
    $ {\bf{\ H}}_{\rm cat} \in \mathbb{C}^{ N_R \times (1+M)N_T}$  is the concatenated channel defined as
    $ {\bf{\ H}}_{\rm cat}\!=\! 
    [
    \begin{matrix}
      {{{{\bf{\bar H}}}_d}} & {{{{2^b}} \over \pi }\sin ({{\pi  \over {{2^b}}}}){{{\bf{\bar G}}}_1}} &  \cdots  & {{{{2^b}} \over \pi }\sin ({{\pi  \over {{2^b}}}}){{{\bf{\bar G}}}_{\rm{M}}}}  \cr 
    \end{matrix} 
    ],$ 
  and substituting  it  into problem (P2), we  can reexpress $g_{\rm MSE} ({\bf{W}},{\bm{\theta}})$ in  an explicit  form w.r.t. $\tilde{\bm \theta}$  as  
 \begin{align} \label{thetaf}
   g_{\rm sub2}
({\bm{\tilde \theta}})\!=& \!{\rm Tr} \big(    {\bf{W}}^H  ( { \bm{\tilde \theta}} \!\otimes\! {\bf{I}}_{N_T} )^H {\bf{\ H}}_{\rm cat}^H   
 \big(
  {\bf{\tilde T}}^{\rm cas}_{{\bm{\tilde \theta}}}   \nonumber \\
  & +
\epsilon_b {\bf{T}}^{\rm com}_{ {\bf{W}}} +{\bf{T}}^{\rm SI}_{{\bf{W}}} 
 \big)^{-1}
 \rm {\bf{\ H}}_{\rm cat} ( { \bm{\tilde \theta}} \!\otimes\! {\bf{I}}_{N_T} ) {\bf{W}}   \big),
  \end{align} 
  where 
${\bf{\tilde T}}^{\rm cas}_{{\bm{\tilde \theta}}} $ is obtained  by replacing $ {\bf{\bar H}}_{\bm{\theta}}$ involved in $ {{\bf{T}^{\rm cas}_{{\bf{W}},\bm{\theta}}}} $ with $ {\bf{\ H}}_{\rm cat}( { \bm{\tilde \theta}} \!\otimes\! {\bf{I}}_{N_T})$. 
It is readily inferred that ${g_{\rm sub2}}({\bm{\tilde \theta}})$ is non-convex w.r.t $ {\bm{\tilde \theta}}$. 
Analogous to the optimization of $\bf{W}$, we can still apply  the MM technique to find a tractable  surrogate function of ${g_{\rm sub2}}({\bm{\tilde \theta}})$. Specifically, let us define   ${\bf{M}}_{ \bm{\tilde \theta}}\!=\!{\bf{H}}_{\rm cat} ( { \bm{\tilde \theta}} \otimes {\bf{I}}_{N_T} ) {\bf{W}}$ and 
    ${\bf{N}}_{\bm{\tilde \theta}}\!= \!{\bf{\tilde T}}^{\rm cas}_{{\bm{\tilde \theta}}}\!+\!
    \epsilon_b {\bf{T}}^{\rm com}_{ {\bf{W}} } \!+\!{\bf{T}}^{\rm SI}_{{\bf{W}}} $. Then, 
    a surrogate function  ${g_{\rm sub2}^{\rm Low}({\bm{\tilde \theta}};{\bm{\tilde \theta}}_{t})}$ at the given point 
    $ {\bm{\tilde \theta}}_{t}$
  can be  constructed as 
  \begin{subequations} \label{thetaMMf}
  \begin{align}
    &{g_{\rm sub2}^{\rm Low} ({\bm{\tilde \theta}}; {\bm{\tilde \theta}}_{t})} \nonumber \\
    &\mathop  = \limits^{({c_1})} 2\Re\left\{ \rm tr ( {\bf{M}}_{\bm{\tilde \theta}_{t}} ^H {\bf{N}}_{\bm{\tilde \theta}_{t}}^{-1} {\bf{M}}_{\bm{\tilde \theta}} ) \right\}\!-\!{\rm tr ( {\bf{N}}_{ \bm{\tilde \theta}_{t}}^{-1} {\bf{M}}_{\bm{\tilde \theta}_{t}} {\bf{M}}_{\bm{\tilde \theta}_{t}}^{H} {\bf{N}}_{ \bm{\tilde \theta}_{t}}^{-1} {\bf{N}}_{ \bm{\tilde \theta}} )} \label{c1} \\
     &\mathop  = \limits^{({c_2})}   2\Re\left\{ \rm {\bm{\xi}^H} 
      \rm vec    ( { \bm{\tilde \theta}} \!\otimes\! {\bf{I}}_{N_T} )
      \right\} \!-\!
       \rm vec^H ( { \bm{\tilde \theta}} \!\otimes\! {\bf{I}}_{N_T})
{\bf{K}}
 \rm vec ( { \bm{\tilde \theta}} \!\otimes\! {\bf{I}}_{N_T}) \nonumber \\
 & \qquad \!+\! {\rm const} \label{c2} \\   
 & \mathop  = \limits^{({c_3})}
 2\Re \left\{ \rm  {\bm{\bar \xi }}^H  {\bm {\tilde \theta}}  \right\}
 \!-\!{\bm {\tilde \theta}}^H  
 {\bf{\bar K}}
 {\bm {\tilde \theta}}
 + {\rm const},  \label{c3}
\end{align}
\end{subequations}
 where ${\rm const}$ denotes the constant term  irrelevant to ${\bm {\tilde \theta}}$ and  the auxiliary matrices are defined in \eqref{auxiliary matrix} at the top of the next page.
  \begin{figure*}[t] 
    \begin{subequations} \label{auxiliary matrix}
      \setlength{\abovedisplayskip}{3pt}
      \setlength{\belowdisplayskip}{3pt}
      \begin{align}
        & {\bm{\xi }} ={\rm vec} (  {\bf{H}}_{\rm cat}^H {\bf{N}}_{ \bm{\tilde \theta}^{(t)}}^{-1}  {\bf{M}}_{ \bm{\tilde \theta}^{(t)}} {\bf{W}}^H ), \\   
         &{\bf{K}} = ({\bf{W}}{\bf{W}}^H+\beta_T^2 {\rm{ \widetilde{diag}}}({\bf{W}}{\bf{W}}^H))^T \otimes ({ {\bf{H}}_{\rm cat}^H {\bf{N}}_{ \bm{\tilde \theta}_{t}}^{-1} {\bf{M}}_{\bm{\tilde \theta}_{t}} {\bf{M}}_{\bm{\tilde \theta}_{t}}^{H} {\bf{N}}_{ \bm{\tilde \theta}_{t}}^{-1} {\bf{H}}_{\rm cat}}) \nonumber \\
         & \qquad  + ({\bf{W}}{\bf{W}}^H)^T \otimes \beta_T^2 {\bf{H}}_{\rm cat}^H {\rm{ \widetilde{diag}}} ({\bf{N}}_{ \bm{\tilde \theta}_{t}}^{-1} {\bf{M}}_{\bm{\tilde \theta}_{t}} {\bf{M}}_{\bm{\tilde \theta}_{t}}^{H} {\bf{N}}_{ \bm{\tilde \theta}_{t}}^{-1}) {\bf{H}}_{\rm cat},  \\
         &
         \bm{\bar \xi }=\sum\nolimits_{\alpha  = 1}^{N_T} {{\bm{ \xi }}_{i \in \mathcal{I}}(\alpha)},  
         \quad  
         {\bf{\bar K}}=\sum\nolimits_{\alpha = 1}^{N_T} \sum\nolimits_{\beta = 1}^{N_T}{\bf{ K}}_{i \in \mathcal{I}, j \in \mathcal{I}}(\alpha,\beta). 
        \end{align}
      \end{subequations}    
      \vspace{-2mm}
      \hrulefill
      \end{figure*}
    The equation \eqref{thetaMMf} describes the construction of an effective surrogate function of $g_{\rm sub2}({\bm {\tilde \theta}})$ and the process of transformation from the implicit function to the explicit one w.r.t ${\bm {\tilde \theta}}$.
    The equality $(c_1)$ holds similar to $ (b_1)$, and $ (c_2)$ is due to the matrix identity
     $ \rm Tr( {\bf{A}}{\bf{B}}{\bf{C}}{\bf{B}}^H)=vec^H({\bf{B}})({\bf{C}}^T\otimes {\bf{A}}) vec({\bf{B}})$.
     {\color{black} Due to the sparsity of the vector $\rm vec ( { \bm{\tilde \theta}} \otimes {\bf{I}}_{N_T} )$, we define the  index set  $\mathcal{I} $ which includes  the indexes of all non-zero elements in $\rm vec ( { \bm{\tilde \theta}} \otimes {\bf{I}}_{N_T} )$. Then, ${\bm{ \xi }}_{i \in \mathcal{I}}$  denotes the auxiliary vector composed of the elements of ${\bm{ \xi }}$ indexed by the index set $\mathcal{I}$. ${\bm{ \xi }}_{i \in \mathcal{I}}(\alpha)$ represents the $\alpha$-th sub-vector of ${\bm{ \xi }}_{i \in \mathcal{I}}$ and  ${\bf{ K}}_{i \in \mathcal{I}, j \in \mathcal{I}}(\alpha,\beta)$ is similarly defined.} 
     {\color{black}With} simple  manipulations, it is easily inferred that $(c_3)$ holds. 
     
     Recalling  $\rm { \bm{\tilde \theta}}=\left[ 1, \  {{{\bm{\theta }}^H}} \right]^H$, we next intend  to equivalently   transform 
    ${g_{\rm sub2}^{\rm Low}({\bm{\tilde \theta}};{\bm{\tilde \theta}}_{t})}$ into a function  w.r.t.  ${\bm{\theta}}$. 
    Specifically, we rewrite ${\bm{\bar \xi}}$ as  ${\bm{\bar \xi }}=[{\bar \xi_1} ; {\bm{\bar \xi}}_2]$ and  divide ${\bf{\bar K}}$ into a $2 \times 2$ block matrix, i.e., ${\bf{\bar K}}=[{\rm \bar k_{1}} \ {\bf{\bar k_{2}}}^H ; {\bf{\bar k_{2}}} \ {\bf{\bar K_3}}]$,  thus yielding
          \begin{equation} \label{f_4_theta}
              {g_{\rm sub2}^{\rm Low}({\bm{\theta}};{\bm{\theta}}_{t})}=
              2\Re\left\{ {\bm{\theta}}^H ({\bm{\bar \xi}}_2\!-\!{\bf{\bar k_{2}}})\right\}\!-\!{\bm{\theta}}^H {\bf{\bar K_3}} {\bm{\theta}} + \rm const2, 
          \end{equation} 
    where $\rm const2$ represents the constant term independent of ${\bm{\theta}}$. 
    It is clear that  ${g_{\rm sub2}^{\rm Low}({\bm{\theta}};{\bm{\theta}}_{t})} $ is a   quadratic
    function w.r.t $ {\bm{\theta}}$. Nevertheless,   the optimal closed-form  $\bm{\theta}$ for  minimizing  ${g_{\rm sub2}^{\rm Low}({\bm{\theta}};{\bm{\theta}}_{t})}$ under the discrete RIS phase constraints is still difficult to obtain. 
    To tackle such difficulty, 
    we  employ the MM strategy  to  perform  the  majorization on    ${g_{\rm sub2}^{\rm Low}({\bm{\theta}};{\bm{\theta}}_{t})}$.
    Since the gradient of  ${g_{\rm sub2}^{\rm Low}({\bm{\theta}};{\bm{\theta}}_{t})} $ is easily found to be  Lipschitz continuous  with  constant { ${\bf{L}}_{g_{\rm sub2}^{\rm Low}}\!=\!\lambda_{\rm max}({\bf{\bar K_3}}) {\bf{I}}$ },   a locally tight lower  bound of ${g_{\rm sub2}^{\rm Low}({\bm{\theta}};{\bm{\theta}}_{t})} $ at the point ${\bm{\theta}}_{t}$  is given by \cite[Lemma~12]{MM}
        {
        \begin{subequations}
          \begin{align}\label{thetasubsub}
            {g_{\rm sub2}^{\rm Low}({\bm{\theta}}; {\bm{\theta}}_{t})} 
            &\geq {{\tilde g}_{\rm sub2}^{\rm Low}({\bm{\theta}};{\bm{\theta}}_{t})}  \!=\!2\Re \left\{ {{{\bm{\theta }}^H}} {\bf{b}}\right\}\!\nonumber \\
            &
            \quad -2\lambda_{\rm max}({\bf{\bar K_3}})M +{\bm{\theta}}_{{t}}^H  {\bf{\bar K_3}}   {\bm{\theta}}_{t} \!+\!{\rm const2}, \tag{23}
          \end{align}
        \end{subequations} }
         \noindent where ${\bf{b}} \!=\! \lambda_{\rm max}({\bf{\bar K_3}}) {{\bm{\theta }}_{t}} \!-\! {{{\bf{\bar K}}}_3}{{\bm{\theta }}_{t}} \!+{\bm{\bar \xi}}_2 \!-\! {{{\bf{\bar k}}}_2} $.
         After these  two majorization operations,  an alternative  lower bound optimization problem  w.r.t.  the  RIS reflection phase ${\bm{ \phi}}$ can be  formulated as 
         \begin{subequations}
         \begin{alignat}{2}\label{P}
           (\text{P4}):\  \max_{{\bm{ \phi}} } \quad \! &
           2   \sum\nolimits_{m = 1}^M { \left| {{{\bf b}_m}} \right|{cos{\left( {{{\bm \phi} _m} \!-\! \angle {{\bf b}_m}} \right)}} } , \nonumber \\ 
           \mbox{s.t.} \ \ \  
           & \eqref{P1C2}. \tag{24}
         \end{alignat}
        \end{subequations}
         It is readily inferred that  ${\phi}_m$ in problem (P4) can be decoupled from the objective function and constraints, which means that the variables can be updated  in parallel. 
         Hence, the optimal RIS beamforming vector to problem  (P4)   is given by 
          \begin{equation}\label{update_theta_MM}
              {\bm{\theta}}^{\star}=e^{j \ {\rm arg} \min_{\bm{\phi} \in \mathcal{F}} |\angle { \bf b}-{\bm{ \phi}} |  }.    
          \end{equation}

        Notice that by  iteratively solving the approximate problem (P4), the objective value of ${g_{\rm sub2}^{\rm Low}({\bm{\theta}})}$ keeps non-decreasing and converges {\color{black} to a finite value}. 

 \subsubsection{Modified  RGA Algorithm}
      Traditionally, the  RGA algorithm has been widely  applied   to  optimize  the continuous    RIS reflection coefficients over the complex circle manifold (CCM)  space. When extending it into the case of  RIS discrete phase shifts, an intuitive solution  is {\color{black}to quantize} the  obtained continuous {\color{black}solution} to find its nearest value. Unfortunately, such {\color{black} a direct quantization} usually leads to significant performance loss, especially {\color{black} for the case with} low-resolution RIS  phase shifts. To alleviate this issue,  discrete RIS  phase shifts are directly considered in each iteration of the modified RGA method,  as elaborated below. 

      \textit{(a) Riemannian gradient:}
      We first consider relaxing the discrete RIS phase shifts to the continuous ones. Since the search space of the continuous phase shifts is the product of $M$ complex circles, denoted as $  \mathcal{S}^{M}\!=\!\{{\bm{ \theta}} \!\in\! \mathbb{C} ^{M}\!: \!|\theta_m|\!=\!1,~ m=1,2,...,M \} $, the maximization of  ${g_{\rm sub2}^{\rm Low}({\bm{\theta}}; {\bm{\theta}}_{t})} $ can be solved by  the classical  CCM algorithm. 
      Specifically, the Riemannian gradient of $f({\bf x})$ at the point ${\bf x}_k$ is essentially  the projection of the Euclidean gradient onto the tangent space of the CCM at the current point  $ {\bf{x}}_k$, and is defined as: 
      $ P_{\tau_{{\bf{x}}_k \mathcal{S}^{M+1} }}({\bm{\eta}}_k)\!=\!{\bm{\eta}}_k\!-\!
      \Re\{ {\bm{\eta}}_k \odot {\bf{x}}_k^*   \}\odot {\bf{x}}_k $, 
      where $ P_{\tau_{{\bf{x}}_k \mathcal{S}^{M+1} }}({\bm{\eta}}_k) $ denotes the projection operator and  ${\bm{\eta}}_k$ represents the Euclidean gradient of $f({\bf x})$ \cite{CCM2019}.
      Based on this, the Riemannian gradient of ${g_{\rm sub2}^{\rm Low}({\bm{\theta}}; {\bm{\theta}}_{t})} $ at the  iteration point ${\bm{\theta}}_{n}$ can be expressed as
      $
        \nabla_{\mathcal{S}^{M+1} }{g_{\rm sub2}^{\rm Low}({\bm{\theta}}_{n})}\!=\!\nabla g_{\rm sub2}^{\rm Low}({\bm {\theta}}_{n})\!-\!
      \Re\left\{ \nabla g_{\rm sub2}^{\rm Low}({\bm {\theta}}_{n}) \odot {\bm {\theta}}_{n}^*   \right\}\odot {\bm {\theta}}_{n} , 
      $
      where $  \nabla {g_{\rm sub2}^{\rm Low}({\bm{\theta}}_{n})}\!=\!2\left({\bm{\bar \xi}}_2\!-\!{\bf{\bar k_{2}}}\!-\!{\bf{\bar K_3}} {\bm {\theta}}_{n}
       \right) $ and $n$ is the iteration index  of the modified RGA algorithm.

      \textit{(b) Update Rule:} 
      We then update ${\bm {\theta}}_{n}$ along the direction of Riemannian gradient with the step size $\rho_n$. Specifically, the update rule of ${\bm {\theta}}_{n}$ is given by ${\bm {\theta}}_{n+1}^{'}={\bm {\theta}}_{n}+\rho_n \nabla_{\mathcal{S}^{M+1} }{g_{\rm sub2}^{\rm Low}({\bm{\theta}}_{n})}$. In particular, ${\bm {\theta}}_{n}$ belongs to  the discrete feasible  set, whereas  ${\bm {\theta}}_{n+1}^{'}$ is generally a non-feasible point. 
      With regard to this fact, we consider  a retraction operator in this step to find the feasible RIS phase shifts. This retraction operator obtains the new point ${\bm {\theta}}_{n+1}$ by mapping the updated point ${\bm {\theta}}_{n+1}^{'}$ to the nearest discrete point in the feasible set, i.e.
      \begin{equation}\label{update_theta_CCM}
        {\bm { \theta}}_{n+1}=e^{j\ {\rm arg} \min_{{\bm \phi} \in \mathcal{F}} |\angle {\bm {\theta}}_{n+1}^{'}-{\bm \phi} | }.
      \end{equation}
      Moreover, in order to preserve the monotonicity of the modified RGA algorithm, we adopt the backtracking line search to determine the step size $\rho_n$.
      The detail of this modified RGA algorithm is presented in Algorithm~\ref{Algorithm1}.
      
        \begin{algorithm}
          \caption{Modified  RGA Algorithm}\label{Algorithm1}
         \SetKwInOut{Input}{Input}
         \SetKwInOut{Output}{Output}
         \Input{{the optimal transmit precoder ${\bf{W}}_{t+1}$ at the $t+1$ iteration}.}
         \Output{the optimal phase shifts ${\bm{\theta}}_{t+1}$.}
         Initialization: $n=0$, $\beta \in (0,1)$, ${\bm{\theta}_{n}}={\bm{\theta}_{t}}$

          \Repeat{$|{g_{\rm sub2}^{\rm Low}({\bm{\theta}}_{n}; {\bm{\theta}}_{t})}-{g_{\rm sub2}^{\rm Low}({\bm{\theta}}_{n-1};{\bm{\theta}}_{t})}|<\epsilon$}
          {
            {Update ${\bm{\theta}_{n+1}}$ by \eqref{update_theta_CCM};} \\
            {\While{${g_{\rm sub2}^{\rm Low}({\bm{\theta}}_{n+1}; {\bm{\theta}}_{t})} < {g_{\rm sub2}^{\rm Low}({\bm{\theta}}_{n}; {\bm{\theta}}_{t})}$}{
              {$\rho_n=\beta \rho_n$;} \\
              {update ${\bm{\theta}_{n+1}}$ by \eqref{update_theta_CCM}}.
            }} 
            {update $n=n+1$;}
          }
          Return ${\bm{\theta}}_{t+1}={\bm{\theta}}_{n} $.
      \end{algorithm} 

      The whole MM-based AO  algorithm   is summarized  in Algorithm~\ref{Algorithm2}. 
       For the initialization, we conduct the  joint design of the MIMO  transceiver  and  RIS  reflection matrix  for the ideal  RIS-aided MIMO system with  $ \sigma_d^2\!=\!\sigma_m^2\!=\!0$ and $\beta_T\!=\!\beta_R\!=\!0$, and choose the {\color{black}obtained} solution as the initial point. 
      
        \begin{algorithm}
        \normalsize
          \caption{Proposed  MM-based AO Algorithm }\label{Algorithm2}
         \SetKwInOut{Input}{Input}
         \SetKwInOut{Output}{Output}
         \Input{System parameters  $d,N_T,N_R,M$, the threshold $\epsilon$, etc.}
         \Output{The transmit beamforming ${\bf{W}}^{\star}$  and RIS phase shifts ${\bm{\theta}}^{\star}$.}
         Set the obtained solution of the ideal RIS-aided MIMO system as the initial point $ ({\bf{W}}_{0},{\bm{\theta}}_{0})$.

          \Repeat{$|g_{\rm MSE}({\bf{W}}_{t+1},{\bm{\theta}}_{t+1})\!-\! g_{\rm MSE}({\bf{W}}_{t},{\bm{\theta}}_{t})|<\epsilon$}{
            For fixed ${\bm{\theta}}_{t}$,  update ${\bf{W}}_{t+1}$ according to  \eqref{update_W}, \\
            For fixed ${\bf{W}}_{t+1}$, compute  ${\bm{\theta}}_{t+1}$ according to \eqref{update_theta_MM} or Algorithm~\ref{Algorithm1},
          }
          Return ${\bf{W}}^{\star}={\bf{W}}_{t+1} $ and ${\bm{\theta}}^{\star}={\bm{\theta}}_{t+1} $.        
      \end{algorithm}

\section{Optimality and Performance Analysis}
        In this section,  in order to explore the optimality of the proposed MM-based AO  algorithm, we focus on  the joint MIMO transceiver and RIS reflection matrix design in the special scenario of {\color{black} the RIS-aided MIMO system with only LoS BS-RIS-user cascaded channels.
        } Moreover,  we  reveal  the  irreducible MSE floor effect  induced  by the hardware distortions  and CSI errors in  the RIS-aided MISO system  in the high SNR-regime. Also,  the convergence behavior and  the complexity  of the proposed MM-based AO algorithm are analyzed.

  \subsection{Optimality Analysis }
  { In this subsection, we assume that the RIS-related channels only have LoS components to  obtain insight into the system behavior.} To be specific,  
        suppose that the uniform linear arrays (ULAs) are  deployed at both the BS and the  user, and a uniform planar array (UPA) is employed at the RIS. {\color{black}Then,}
        the BS-RIS channel ${\bf H}_{I}$ and the RIS-user channel ${\bf H}_{r}$ can be respectively expressed as the well-known Saleh-Valenzuela (SV) channel model  
          \begin{align}\label{SVchannel}
            &{\bf{H}}_{I}=\nu_I {\bf a}_{TX}({\psi}_{TX}) {\bf a}_{R}^H({\psi}_{A},{\vartheta}_{A}),  \nonumber \\ 
            &{\bf{H}}_{r}=\nu_r {\bf a}_{RX}({\psi}_{RX}) {\bf a}_{R}^H({\psi}_{D},{\vartheta}_{D}), 
          \end{align}
        where $\nu_I$/$\nu_r$ denotes the complex channel gain of the LoS path for the BS-RIS/RIS-user channel. ${\psi}_{TX}$  denotes the angle of departure (AoD) associated with the BS, while ${\psi}_{RX}$  denotes the angle of arrival (AoA) associated with the user. ${\psi}_{i}$ and ${\vartheta}_{i}$ with $i\in\{A,D\}$ are the azimuth and elevation angles of arrival/departure associated with the RIS, respectively.    
        ${\bf a}_{TX}({\psi}_{TX})$/${\bf a}_{RX}({\psi}_{RX})$ and ${\bf a}_{R}({\psi}_{i},{\vartheta}_{i})$ denote  the array response vectors at the BS/user and the RIS, respectively.
        For an $N_T$-antenna ULA and an $M$-element UPA, the array response vectors at the BS and RIS are respectively given by 

        \begin{small}
          \begin{align}
          & {\bf a}_{R}({\psi}_{i},{\vartheta}_{i})  
          \!\!=\!\!\frac{1}{\sqrt{M}}\big[ 1, \cdots,e^{j {\frac{2 \pi D}{\gamma}} [{{\rm sin}({\psi_{i}}) {\rm sin}({\vartheta_{i}})+{\rm sin}({\vartheta}_{i})}]  }, \cdots, \nonumber \\ 
           & \qquad \qquad \qquad \qquad \quad  e^{j {\frac{2 \pi D}{\gamma}} [{  (M_x-1) {\rm sin}({\psi_{i}}) {\rm sin}({\vartheta_{i}})+(M_y-1){\rm sin}({\vartheta}_{i})}]  }\big]^T, \nonumber\\
          &{\bf a}_{TX}({\psi}_{TX})\!\!=\!\!\frac{1}{\sqrt{N}}\left[1,e^{j {\frac{2 \pi D}{\gamma}} {{\rm sin}({\psi_{TX}})}},\cdots , e^{j {\frac{2 \pi D}{\gamma}} (N_t-1){{\rm sin}({\psi_{TX}})}} \right]^T,  
          \end{align} 
        \end{small}

        \noindent where $D$ denotes the antenna spacing and  $\gamma$ is the signal wavelength. $M_x$ and $M_y$ represent the number of elements in the row and column of the UPA array, respectively. The array response vector ${\bf{a}}_{RX}$ at the user is similarly defined as ${\bf{a}}_{TX}$ at the BS. For ease of notation, we simply denote the steering vector  ${\bf a}_{R}({\psi}_{A},{\vartheta}_{A})$ and ${\bf a}_{R}({\psi}_{D},{\vartheta}_{D})$  as ${\bf a}_{RA}$ and ${\bf a}_{RD}$ in the following sections, respectively.
        Based on the above channel model, the estimated compound channel ${\bf{\bar G}}_m$ associated with the $m$-th RIS element can be written as ${\bf{\bar G}}_m= \nu_m {\bf a}_{RX} {\bf a}_{TX}^H$ with $\nu_m=\nu_r\nu_I^* [{\bf a}_{RD}]_m^* [{\bf a}_{RA}]_m$.  Based on \eqref{SVchannel}, the original optimization problem (P2) can be further reduced to 
        \begin{subequations}
        \begin{alignat}{2}\label{P5} 
          (\text{P5}):\  \max_{{\bf{w}},{\bm{\theta}}} \quad  
          & g_{\rm MSE}({\bf{w}},{\bm{\theta}}) , \nonumber \\
          \mbox{s.t.} \quad 
           &(1+\beta_T^2)||{\bf{w}}||^2\leq P,~ \eqref{P1C2}, \tag{29}
         \end{alignat} 
        \end{subequations}
         where $g_{\rm MSE}({\bf{w}},{\bm{\theta}})$ and the intermediate variables are defined in \eqref{g5} at the top of the next page.

         \begin{figure*}[t] 
          \begin{subequations} \label{g5}
            \begin{align}
              & g_{\rm MSE}({\bf{w}},{\bm{\theta}})  =\!  {\bf a}_{RX}^H \left( \frac{c_1\!+\!\sum\nolimits_{m = 1}^M  |\nu_m|^2}{c_1}  \left( {\bf R}_A \!+ \!\frac{{\bf{w}}^H {\rm{ \widetilde{diag}}} \left({\bf a}_{TX}{\bf a}_{TX}^H \right) {\bf{w}}} {{{\bf w}^H} {\bf a}_{TX} {\bf a}_{TX}^H{\bf w}} {\bf R}_B \right) \!+ \!\frac{c_2\!+\!c_3{{\bf{w}}^H{{\bf{w}}}}}{c_1{{\bf w}^H} {\bf a}_{TX} {\bf a}_{TX}^H{\bf w}} {{\bf{I}}_{{N_R}}} \right)^{-1} {\bf a}_{RX}, \tag{30a} \\
              & {\bf R}_A={\bf a}_{RX}{\bf a}_{RX}^H+{\beta _{R}^2}{\rm{ \widetilde{diag}}} \left( {\bf a}_{RX}{\bf a}_{RX}^H \right), 
                          \qquad \qquad \qquad \qquad \quad
                          {\bf R}_B={\beta _T^2} {\bf a}_{RX}{\bf a}_{RX}^H, \tag{30b}  \\  
          
              & c_1\!=\!(1\!-\!\epsilon_b) |\nu_r|^2|\nu_I|^{2}|{\bf a}_{RD}^H {\bm{\Theta}}  {\bf a}_{RA}|^2, \qquad 
                          
              c_2\!=\! (1 \! + \! {\beta_{R}^2})\sigma^2 \label{c_1}, \qquad
                          
              c_3\!=\!\left( {1 \!+\! {\beta _T^2} \!+\! {\beta _R^2}} \right) ( {{\sigma _d}^2 \!+\! \sum\limits_{i = 1}^M {{\sigma _m}^2} } ).
              \tag{30c}
          
            \end{align}
          \end{subequations}
              \hrulefill
              \vspace{-3mm}
          \end{figure*}

It is worth noting that problem (P5) is still hard to solve due to the tightly coupled variables and  non-convex constraints. Thanks to the  proposed MM-based AO algorithm, the globally optimal solution of problem (P5) admits a closed form.
Specifically, for the transmit precoder design with fixed RIS phase shifts, the function $g_{\rm MSE}({\bf{w}},{\bm{\theta}})$ can be equivalently  transformed into the following generalized Rayleigh quotient problem:
        \begin{alignat}{2}\label{P6}
          (\text{P6}):\  \max_{{\bf{w}}} \quad  & g_{\rm sub1}({\bf w})=\frac{{{\bf w}^H} {\bf a}_{TX} {\bf a}_{TX}^H{\bf w}} {{{\bf w}^H}{\bf R} {{\bf w}}}, \nonumber \\
           \mbox{s.t.} \quad 
          & (1+\beta_T^2)||{\bf{w}}||^2\leq P ,
         \end{alignat}
         where 
         $
         {\bf R}=(c_1 \! + \! \sum\nolimits_{m = 1}^M  |\nu_m|^2)(N_R \!+ \!{\beta _{R}^2}){\bf a}_{TX} {\bf a}_{TX}^H \! +\!(c_3\! +\! \frac{c_2}{P}\!+\!(c_1\!+\!\sum\nolimits_{m = 1}^M  |\nu_m|^2){\beta _{T}^2} N_R ){{\bf{I}}_{{N_R}}}.
         $  
         {\color{black}
         Since the generalized Rayleigh quotient problem has the tractable property, 
         the function $g_{\rm sub1}({\bf w})$ itself can be regarded as an  MM-based surrogate function.
         }
         It is readily inferred that the  optimal closed-form solution to the above generalized Rayleigh quotient can be  derived as 
         \begin{equation}
         {\bf w}^{\star}={\sqrt{\frac{P}{1+\beta_T^2} }} \frac{{{\bf a}_{TX}}}{||{{\bf a}_{TX}}||}.
         \end{equation}
         Substituting the optimal transmit precoder ${\bf w}^{\star}$ into  $g_{\rm MSE}({\bf{w}},{\bm{\theta}})$, the objective function w.r.t. the RIS reflection beamforming  can be  simplified as  $g_{\rm sub2}({\bm \theta})$ as defined in \eqref{g_sub2} at the top of the next page.
         \begin{figure*}[!t]
         \begin{align}\label{g_sub2}
          \begin{autobreak}
          g_{\rm sub2}({\bm \theta})= \frac{c_1 N_T N_R}{c_3+\frac{c_2}{P}+(c_1 \! + \! \sum\nolimits_{m = 1}^M  |\nu_m|^2){\beta _{T}^2} N_R +(c_1 \! + \! \sum\nolimits_{m = 1}^M  |\nu_m|^2)(N_R+{\beta _{R}^2})N_T}.
        \end{autobreak}
         \end{align} 
         \hrulefill
         \vspace{-3mm}
        \end{figure*}
         Note that the objective function $g_{\rm sub2}({\bm \theta})$ is monotonically increasing w.r.t. $c_1$  which is a function of ${\bm{\Theta}}$ as defined in \eqref{c_1}. Hence, the optimization of the discrete RIS phase shifts can be  formulated as  
         \begin{subequations} 
         \begin{alignat}{2}\label{Px}
          (\text{P7}):\ \max_{{\bm{\theta}}} \quad  & |{\bf a}_{RD}^H {\bm{\Theta}}  {\bf a}_{RA}| ,~~~
            \mbox{s.t.} \quad 
           & \eqref{P1C2}.    \tag{34}
          \end{alignat}
        \end{subequations}
           Similarly, the objective function of problem (P7) itself can also be regarded as an MM-based surrogate function. 
          Then, the optimal phase shifts ${\bm{\theta}}^{\star}$  is given by $ {\bm{\theta}}^{\star}=e^{j ~{\rm arg} \min_{{\bm \phi} \in \mathcal{F}} |\angle{{\bf a}_{RD}}-\angle{{\bf a}_{RA}}-{\bm \phi}|}$.
          Based on the derived global optimal solutions of ${\bf W}$ and ${\bm \theta}$, the optimality of the  proposed  MM-based AO algorithm can  be verified in the special scenario of the RIS-aided MIMO system with only LoS BS-RIS-user cascaded channels.
          
          %
          \subsection{MSE Floor Analysis} 
          {\color{black} In this subsection, we mainly  consider the RIS-aided MISO communication system, where the optimization problem is more tractable and favorable to analyze the MSE performance.}
          In this context, the estimated direct channel ${\bf{\bar H}}_d$ and compound channels  ${\bf{\bar G}}_m$ are reduced to ${\bf{\bar h}}_d$ and  ${\bf{\bar g}}_m$, respectively. 
          Then, the original  MSE minimization problem  (P1) reduces to 
          \begin{subequations} 
          \begin{alignat}{2}\label{MISO_opti}
            (\text{P8}): \min_{{\bf{w}},{\bm{\theta}}} \quad  & f_{\rm MSE}({{\bf{w}},{\bm{\theta}}})= \nonumber \\ 
            & 1\!-\!  
             \frac{{\bf w}^H  {\bf {\bar h}}{\bf {\bar h}}^H {\bf w}}   { c_2\!\!+\!\!c_3{\bf w}^H {\bf w}\!\!+\!\! {\bf {\bar h}}^H {\bf {R}}_{\rm w} {\bf {\bar h}} \!\!+\!\!\epsilon_b \sum\nolimits_{m \!=\! 1}^{M} {{{\bf{\bar g}}}_m}^H {\bf {R}}_{\rm w} {{{\bf{\bar g}}}_m}},\nonumber \\
             \mbox{s.t.} \quad 
            & (1+\beta_T^2)||{\bf{w}}||^2\leq P,~ \eqref{P1C2}, \tag{35}
           \end{alignat}
          \end{subequations} 
          where ${\bf {\bar h}}$ denotes the available cascaded BS-RIS-user channel in the MISO case,  and ${\bf {R}}_{\rm w}$ represents the signal covariance matrix, defined as 
              $ {\bf {\bar h}}={\bf {\bar h}}_d+{{{2^b}} \over \pi }\sin \left( {{\pi  \over {{2^b}}}} \right) \sum\nolimits_{m = 1}^M {{\theta}_m{{{\bf{\bar g}}}_m}}$ and
              ${\bf {R}}_{\rm w} =(1+\beta_R^2) {\bf w}{\bf w}^H+ \beta_T^2 {\rm{ \widetilde{diag}}} {({\bf w}{\bf w}^H)}$, respectively.      
          
          Note that the term of ${\bf {\bar h}}^H {\bf {R}}_{\rm w} {\bf {\bar h}}$ in $f_{\rm MSE}({{\bf{w}},{\bm{\theta}}})$ causes the main intractability for solving the optimization problem due to the tightly coupled variables of ${\bf w}$ and ${\bm \theta}$. 
          Since $\beta_T\geq 0$ always holds, $f_{\rm MSE}({{\bf{w}},{\bm{\theta}}})$ can be lower bounded as
          \begin{align}\label{f_lowerbou}
            &f_{\rm MSE}({{\bf{w}},{\bm{\theta}}}) \!\geq\! f_{\rm lower}({{\bf{w}},{\bm{\theta}}}) \nonumber \\ 
            &\!\!=\!\!1\!-\!\frac{{\bf w}^H  {\bf {\bar h}}{\bf {\bar h}}^H {\bf w}}   { c_2\!\!+\!\!c_3{\bf w}^H {\bf w}\!\!+\!\! (1\!\!+\!\!\beta_R^2){\bf {\bar h}}^H {\bf w}{\bf w}^H {\bf {\bar h}} \!\!+\!\!\epsilon_b \sum\nolimits_{m = 1}^{M} {{{\bf{\bar g}}}_m}^H {{\bf R}_{\rm w}} {{{\bf{\bar g}}}_m}},
          \end{align}
          where the equality holds when $\beta_T\!=\!0$.
          Then, the lower bound of the total average MSE for the RIS-aided MISO system is revealed in the following proposition.

          \begin{proposition}\label{proposition1}
            Considering the hardware impairments and CSI errors, the average MSE, i.e. $f_{\rm MSE}({{\bf{w}},{\bm{\theta}}})$, of the RIS-aided MISO system can be lower bounded by 
            \begin{align}
              &f_{\rm MSE}({{\bf{w}},{\bm{\theta}}}) > f_{\rm lower}^{\rm opt} \! \nonumber \\ 
              & =\!1\!- 
              \!\frac{(M\!+\!1)\lambda_{\rm max}({\bf {\tilde H}}_{\rm cat}^H ( {\bf{Q}}\!+\! (1\!+\!\beta_T^2) \frac{\sigma^2}{P} {\bf{I}}_{N_T} )^{-1} {\bf {\tilde H}}_{\rm cat})} {1\!\!+\!\!(1\!\!+\!\!\beta_R^2)(M\!\!+\!\!1)\lambda_{\rm max}({\bf {\tilde H}}_{\rm cat}^H ( {\bf{Q}}\!\!+\!\! (1\!\!+\!\!\beta_T^2) \frac{\sigma^2}{P} {\bf{I}}_{N_T} )^{-1} {\bf {\tilde H}}_{\rm cat})},
            \end{align}
           \noindent where 
            $ {\bf Q}=( {1 \!\!+\! \!{\beta _T^2} \!\!+\!\! {\beta _R^2}} ) ( {\sigma _d}^2 \!\!+\!\! \sum\nolimits_{i = 1}^M {{\sigma _m}^2} ) {\bf I}_{N_T} \!\!+\! \epsilon_b \sum\nolimits_{m = 1}^{M} \big( \big(1 $ $ + \beta_R^2 \big) {{{\bf{\bar g}}}_m} {{{\bf{\bar g}}}_m}^H \!+\!\beta_T^2 {\rm{\widetilde{diag}}} \big({{{\bf{\bar g}}}_m} {{{\bf{\bar g}}}_m}^H\big) \big),
              $ and ${\bf{\tilde H}}_{\rm cat}\!=\! \big[
                {{{{\bf{\bar h}}}_d}} \ {{{{2^b}} \over \pi }\!\!\sin( {{\pi  \over {{2^b}}}} ) {{{\bf{\bar g}}}_1}} $ $ \  \cdots  \ {{{{2^b}} \over \pi }\sin ( {{\pi  \over {{2^b}}}} ){{{\bf{\bar g}}}_{\rm{M}}}} 
                 \big] $.

            In the high-SNR regime, the lower bound is asymptotically equal to
            \begin{align}
              & \lim_{{\frac{P}{\sigma_2}} \to \infty}  f_{\rm lower}^{\rm opt}\!=\! f_{\rm floor}\! \nonumber \\
              &=\!1\!-\!\frac{(M+1)\lambda_{\rm max}({\bf {\tilde H}}_{\rm cat}^H  {\bf{Q}} ^{-1} {\bf {\tilde H}}_{\rm cat})}{1+(1+\beta_R^2)(M+1)\lambda_{\rm max}({\bf {\tilde H}}_{\rm cat}^H  {\bf{Q}}^{-1} {\bf {\tilde H}}_{\rm cat})}.
            \end{align}        
          \end{proposition}
          \begin{proof}
            See Appendix~\ref{appendix_A} for detailed proof.
          \end{proof}

            Proposition~\ref{proposition1} implies that in the high-SNR regime, there exists  an irreducible MSE floor, i.e., $f_{\rm floor}$,  in the RIS-assisted MISO  system  owing  to the existence of   hardware distortions  and  CSI estimation errors. 
            Specially, for the ideal case without hardware distortions and CSI errors, the term $\lambda_{\rm max}({\bf {\tilde H}}_{\rm cat}^H ( {\bf{Q}}\!+\! (1\!+\!\beta_T^2) \frac{\sigma^2}{P} {\bf{I}}_{N_T} )^{-1} {\bf {\tilde H}}_{\rm cat})$ in Proposition~\ref{proposition1} reduces to $\lambda_{\rm max}(\frac{P}{\sigma^2}{\bf {\tilde H}}_{\rm cat}^H  {\bf {\tilde H}}_{\rm cat})$. Accordingly, $\lambda_{\rm max}$ asymptotically approaches infinity in the high-SNR regime, namely, $\lim_{{\frac{P}{\sigma_2}} \to \infty} \lambda_{\rm max} \to \infty$. Then, we can conclude that the lower bound $ f_{\rm lower}^{\rm opt}$ in Proposition~\ref{proposition1} asymptotically approaches zero for sufficient large SNR.
            To better illustrate the respective impacts of hardware impairments, RIS phase noise, and imperfect CSI on the MSE performance, we have the following corollaries.
            \begin{remark} \label{Corollary1}
              Considering the ideal RIS implementation and perfect CSI knowledge, the average MSE of the RIS-aided MISO system with transceiver hardware impairments in the high-SNR regime is asymptotically equal to 
              \begin{equation}
                \lim_{{\frac{P}{\sigma_2}} \to \infty}  f_{\rm MSE}^{\rm opt}({{\bf{w}},{\bm{\theta}}})=1-\frac{N_T}{\beta_T^2+(1+\beta_R^2)N_T}.
              \end{equation}
            \end{remark}
            
            \begin{remark}\label{Corollary2}
              Assuming that the perfect hardware is implemented both in transceiver and RIS, for any given RIS beamforming vector ${\bm{\theta}}_t$, the average MSE of the RIS-aided MISO system with channel estimation errors in the high-SNR regime is asymptotically equal to 
              \begin{equation}
                \lim_{{\frac{P}{\sigma_2}} \to \infty}  f_{\rm MSE}^{\rm opt}({{\bf{w}},{\bm{\theta}}_{t}})=1 \!-\! \frac{||{\bf{\bar h}}_d+{\bf{\bar G}}{{\bm{\theta}}_{t}} ||^2_2}{\sigma_d^2\!+\! \sum\nolimits_{i = 1}^M {{\sigma _m}^2} \!+\! ||{\bf{\bar h}}_d \!+\! {\bf{\bar G}}{{\bm{\theta}}_{t}} ||^2_2 },
              \end{equation}
              where ${\bf{{\bar G}}}= \big[
                \begin{matrix}
                 {{{{\bf{\bar g}}}_1}} & {{{{\bf{\bar g}}}_2}} & \cdots  & {{{{\bf{\bar g}}}_{\rm{M}}}}  \cr 
                \end{matrix} \big] $. 
            \end{remark}
            Corollary \ref{Corollary1} and Corollary \ref{Corollary2} show that the MSE floor is a monotonically increasing function concerning transceiver hardware impairments and channel estimation errors, respectively. On the other hand,
            it also indicates that the receiver hardware distortion level affects the MSE performance more than that of the transmitter. 
            Similarly, the impact of RIS phase noise on the average MSE is investigated in Corollary \ref{Corollary3}.
            \begin{remark}\label{Corollary3}
              Assuming the ideal transceiver hardware and perfect CSI knowledge, for any given RIS beamforming vector ${\bm{\theta}}_t$, the average MSE of the RIS-aided MISO system with RIS phase noise in the high-SNR region is asymptotically equal to
              \begin{align}
                & \lim_{{\frac{P}{\sigma_2}} \to \infty}  f_{\rm MSE}^{\rm opt}({{\bf{w}},{\bm{\theta}}_{t}}) = \nonumber \\
                & \frac{(1\!-\!\omega_b^2) }{(1\!\!-\!\!\omega_b^2)\!\!+\!\!\omega_b^2 {\bm{\theta}}_{t}^H {\bf{\bar G}}^{\dagger} {\bf{\bar G}} {\bm{\theta}}_{t}  \!\!+\!\! 2\omega_b \Re \{   {\bm{\theta}}_{t}^H {\bf{\bar G}}^{\dagger} {\bf{\bar h}}_d \} \!\!+\!\!  {\bf{\bar h}}_d^H ({\bf{\bar G}}{\bf{\bar G}}^H)^{-1} {\bf{\bar h}}_d },
              \end{align}
              where $\omega_b=\frac{2^b}{\pi}{\rm sin}(\frac{\pi}{2^b})$ and ${\bf{\bar G}}^{\dagger}={\bf{\bar G}}^H({\bf{\bar G}}{\bf{\bar G}}^H)^{-1}$.
            \end{remark} 
            From Corollary \ref{Corollary3}, we can conclude that the average MSE decreases monotonically with the number of quantization bits increasing. Especially, the result of MSE floor effect can be extended to the general case of RIS-aided MIMO system.

          \vspace{-12pt}
          \subsection{Convergence and Complexity Analysis:}
          \vspace{-3pt}
          In this subsection, we aim to analyze  the convergence  property and computational  complexity  of the proposed two-tier MM-based AO (AO-MM) algorithm and modified RGA-based AO (AO-RGA) algorithm, as illustrated in Proposition~\ref{proposition2}. 
          \begin{proposition}\label{proposition2}
           Suppose the solution sequence generated by the proposed algorithm (AO-MM or AO-RGA) is $\{{\bf W}_{t}, {\bm{\theta}}_{t}\}_{t=0,1,\cdots,\infty}$. Then, the objective value $g_{\rm MSE}({\bf{W}},{\bm{\theta}})$ keeps monotonically non-decreasing and finally converges to a {\color{black}finite} value. 
          \end{proposition}  
          \begin{proof}
           See Appendix~\ref{appendix_B} for detailed proof.
          \end{proof}

          {\color{black} In the sequel, we  mainly  analyze the complexity of the proposed AO algorithm for solving the sum-MSE minimization problem (P1). Clearly, this AO algorithm consists of an alternation between the transmit precoder ${\bf{W}}$   and  the RIS reflection vector  ${\bm{\theta}}$. Firstly, in terms of the optimization of the transmit precoder ${\bf{W}}$, the complexity of  computing ${\bf{W}}^{\star}$  lies in  the matrix inversion and  multiplication involved  in \eqref{TW},  which is $O ( N_T^3\!+\!N_R^3\!+\!MN_TN_R^2\!+\!MN_T^2N_R)$. Secondly, by recalling Section III.B, we propose two different algorithms, namely the two-tier MM-based algorithm and  the modified RGA algorithm,   for the optimization of the RIS reflection vector  ${\bm{\theta}}$.  On the one hand,  the complexity of the two-tier MM-based algorithm is given by  $O( (1\!+\!M)^2(N_T^4\!+\!N_R^2N_T^2)\!+\!T_{\rm MM}{(2M^2\!+\!2^b)})$, where $ T_{ \rm MM}$ denotes the number of iterations in the two-tier MM-based algorithm. Therefore, the MM-based  AO algorithm has a  total complexity of $O( T_{\rm iter1}((1\!+\!M)^2(N_T^4\!+\!N_R^2N_T^2)\!+\!T_{\rm MM}{(2M^2\!+\!2^b)}))$, where $ T_{\rm iter1}$ denotes the total number of iterations. On the other hand, the modified RGA algorithm has a computational complexity of $O( (1\!+\!M)^2(N_T^4\!+\!N_R^2N_T^2)\!+\!T_{\rm RGA}(M^2\!+\!T_{\rm \rho}2^b))$, where $ T_{\rm RGA}$ and $ T_{\rm \rho}$ represent the total iteration number of the RGA updates and the number of iterations for searching the step size $\rho_n$, respectively.
          Accordingly, the total complexity  of the proposed AO  algorithm  becomes $O( T_{\rm iter2}((1\!+\!M)^2(N_T^4\!+\!N_R^2N_T^2)\!+\!T_{\rm RGA}(M^2\!+\!T_{\rm \rho}2^b)))$, where $ T_{\rm iter2}$ denotes  the total number of iterations. 
          }

\section{Simulation}
In this section,  numerical simulations are {\color{black}presented}  to demonstrate the superior performance of the proposed  MM-based AO   algorithm for  minimizing  the average total MSE of the RIS-aided MIMO system  under the hardware impairments and imperfect CSI. Unless otherwise stated, the basic system parameters are elaborated as follows: $ N_T\!$ = 8, $N_R\!$ = 8, $d\!$ = 8, $M\!$ = 64, $\sigma^2\!$ = -100 dBm and $P\!$ = 20 dBm. 
{In this paper, we consider a  three-dimensional coordinate system, where the BS is located at (0m,0m,5m),   the RIS is deployed at (0m,85m,10m), and the user locates at  (5m,120m,1.5m). For the large-scale fading, we choose the log-normal shadowing model given by 
  ${{\rm PL}(d_l)}$ [\rm dB] = ${\rm PL_0}\!+\!10 \alpha_l{\rm log} (d_l)+{X_{\sigma}}$, where $\rm PL_0$   denotes the path loss at the reference distance 1 m, $\alpha_l$  and   $d_l$ represent the  path loss exponent and the propagation distance, respectively \cite{cho2010mimo}. $X_{\sigma}$ denotes the random shadowing effect subject to the Gaussian distribution with zero mean and a standard derivation of $\sigma$.    
  In this simulation, $\rm PL_0$ is chosen as 30  dB and the $\sigma$ is set as 3 dB.
  Due to extensive obstacles and scatterers, the path-loss exponent between the BS and the users is given by $\alpha_{BU}$ = 3.75. By deploying the location of the RIS, the RIS-related link has a higher probability of experiencing nearly free-space path loss. Then, we set the path-loss exponents of the BS-RIS link and of the RIS-user link as $\alpha_{BR}\!=\! \alpha_{RU}$ = 2.2.}
As for the small-scale fading,  the direct channel $ {\bf{H}}_d$ is assumed to be Rayleigh fading due to the extensive scatters, while  the compound channels $\rm {\bf{G}}_m$'s obey  Rician  distribution, where  the  Rician factor is set as 0.75.
{ Furthermore,  the  variance of  CSI  errors  is assumed to be $ \sigma_d^2$ = $\sigma_m^2$ = 0.01 \cite{zhang2020robust,wang2014outage}, and the hardware distortion level is set as $ \beta_T$ = $\beta_R$ = 0.08 \cite{zhangMIMOCapacityResidual2014,bjornsonMassiveMIMOSystems2014,holma2011lte}.} 
The number of quantization bits  for the passive  RIS is  $b$ = 2. The average normalized MSE (ANMSE) is adopted as the performance metric, i.e. ${\rm ANMSE}={f_{\rm MSE}({\bf C}\!,\!{\bf W}\!,\!{\bm{\theta}})}/{d}$. All  simulation results  are obtained by averaging over 500 independent channel realizations.  

\begin{figure}[t] 
  \centering
  \includegraphics[width=0.45\textwidth]{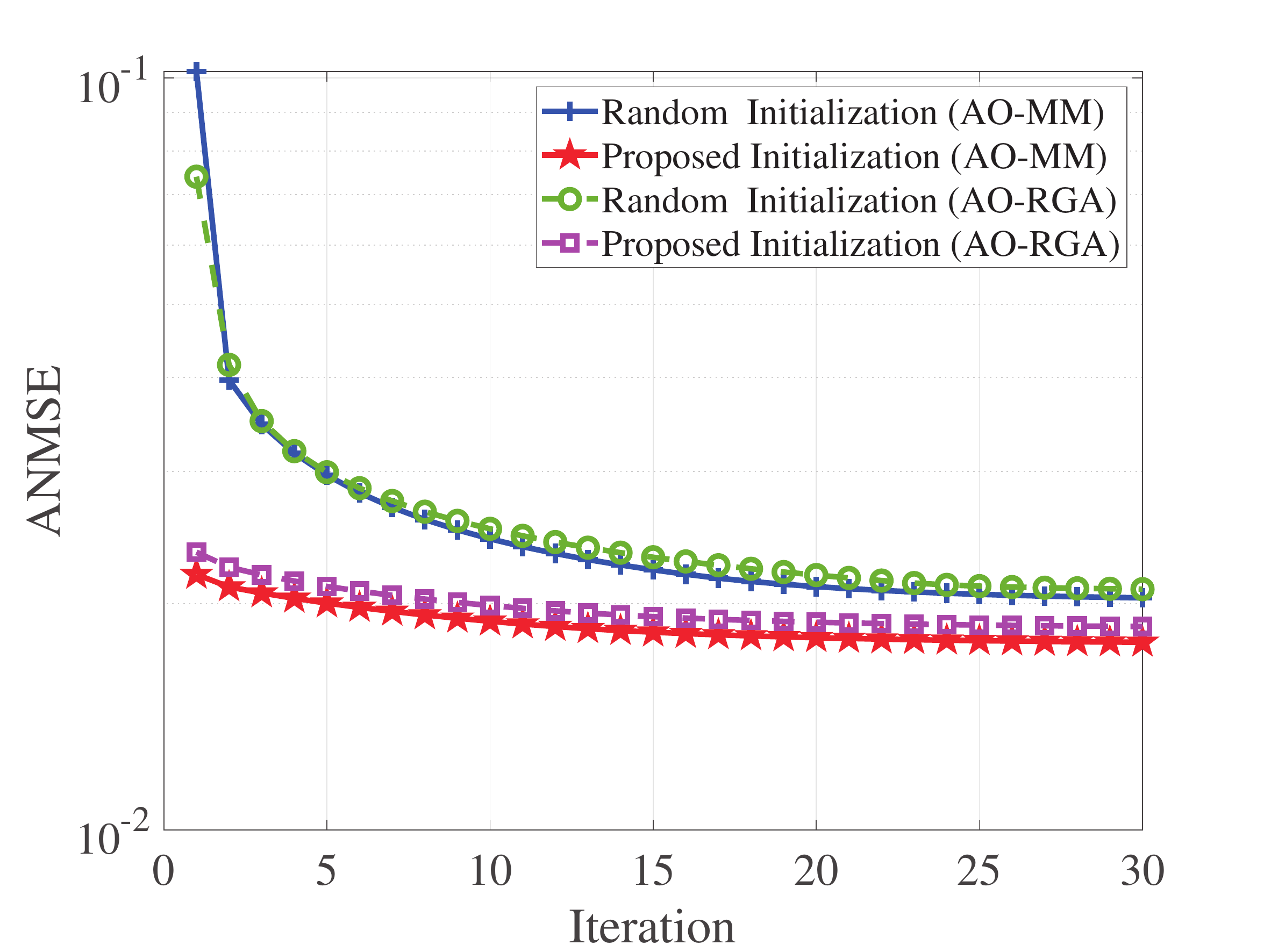}\label{iteration}  
  \caption{Convergence behavior of the proposed AO-MM and AO-RGA 
  algorithms.}\label{Fig2}
  \vspace{-2mm}
\end{figure}

In addition, for a comprehensive  performance comparison, the following benchmark schemes are introduced in the simulation: 
      {\color{black}
      {\bf{1) Perfect hardware:}} We assume perfect hardware implementation in this strategy, i.e. $\beta_T\!=\!\beta_R\!=\!0$. Moreover, the continuous RIS phase shifts are assumed. Naturally, this system setup characterizes an upper bound on the performance of our considered practical system with hardware impairments.
      {\bf{2) Perfect CSI:}} The perfect CSI is assumed to be available in this scheme, namely $\sigma_d^2\!=\!\sigma_m^2\!=\!0$, while only the hardware distortions are considered.
      }
      {\bf{3) Random phase shift:}} The discrete phase shift at each RIS element is randomly chosen and kept fixed in the optimization procedure. 
      {\bf{4) Identity phase shift:}} We adopt the identity phase shifting strategy for the RIS optimization.
      {\bf{5) Nonrobust sheme:}} We adopt the transmit precoder and RIS phase shift design strategy in \cite{MSE6zhaoxin}, which neglects the impacts of hardware impairments and channel estimation errors. Due to its continuous phase shift design, we relax the optimal solution to the nearest discrete phase shifts.

To begin with, Fig.~\ref{Fig2} shows the convergence behavior of the proposed MM-based AO based algorithm under different initializations, where both the two proposed schemes  for optimizing  the RIS reflection matrix, {\color{black} as mentioned in Section III.B}, are considered.   It can be seen from Fig.~\ref{Fig2} that both the proposed AO-MM and AO-RGA algorithms under the ideal-based initialization proposed in Algorithm~\ref{Algorithm2} achieve a lower ANMSE  than that under the random initialization, which means that the choice of initial points significantly affects the ANMSE performance. This is {\color{black}expected} since a good initial point {\color{black}is more likely to lead to a superior solution}.
In addition, we find that the proposed AO-MM and AO-RGA algorithms  can  achieve almost the same converged value under each considered initialization.  Moreover,  these two AO algorithms both converge monotonically within $20$ iterations, which demonstrates their fast convergence speed.

\begin{figure}[t]  
  \centering
  \subfigbottomskip=0pt
  \subfigcapskip=0pt
  \subfigure[(a) Transmit power]
  {
  \includegraphics[width=0.45\textwidth]{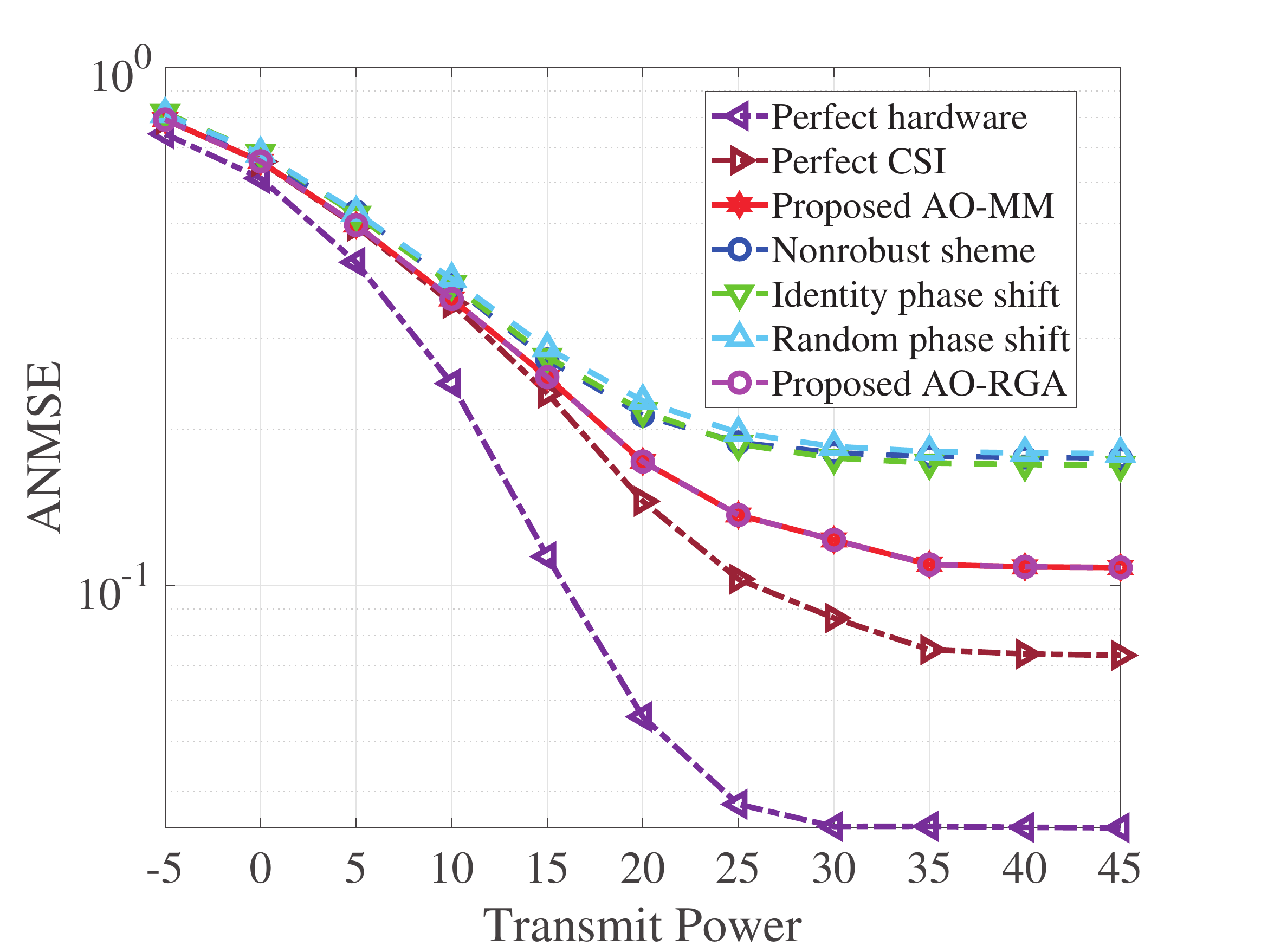}
  }
  \subfigure[(b) Hardware distoriton]
  {
  \includegraphics[width=0.45\textwidth]{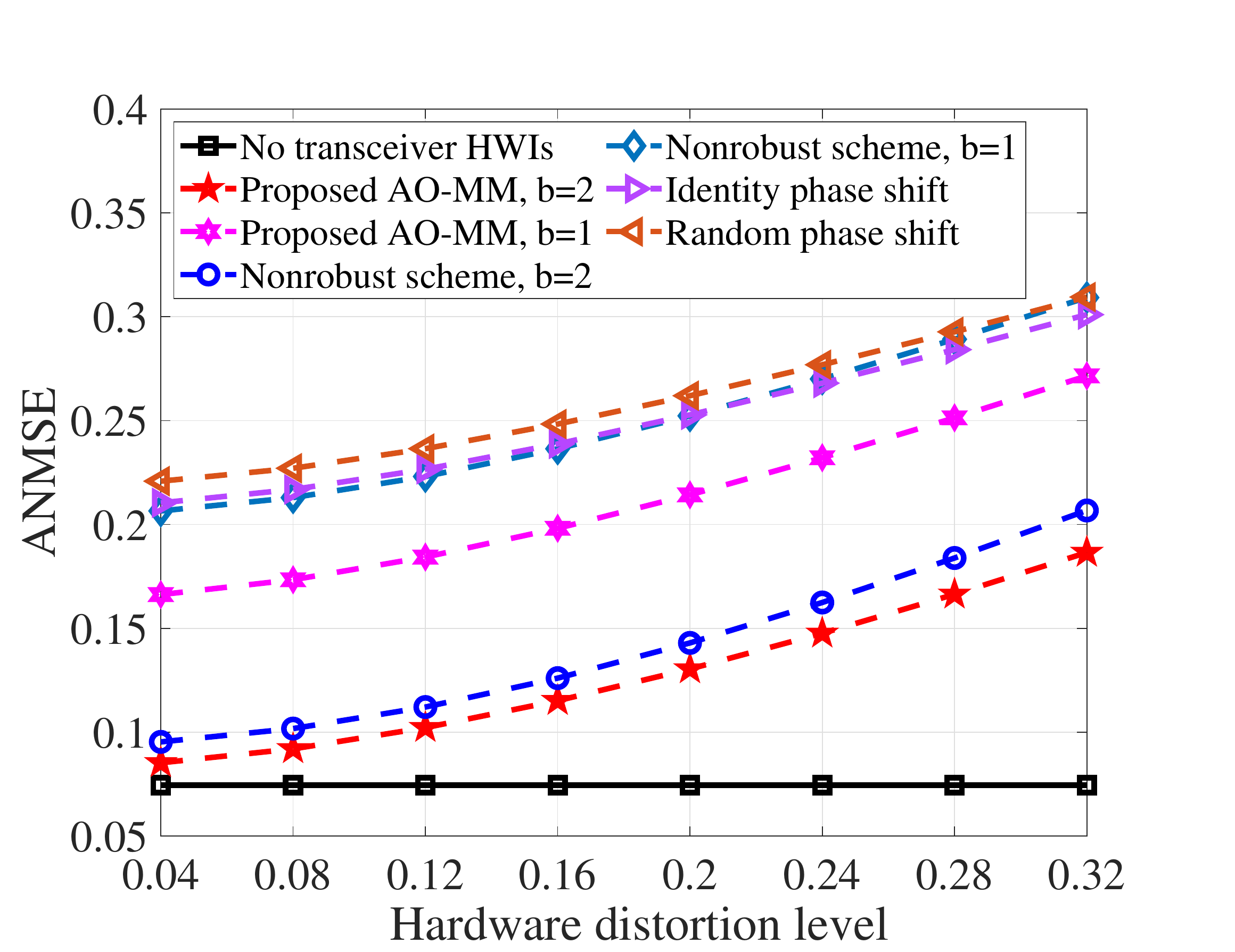}
  }
  \subfigure[(c) CSI error]
  {
  \includegraphics[width=0.45\textwidth]{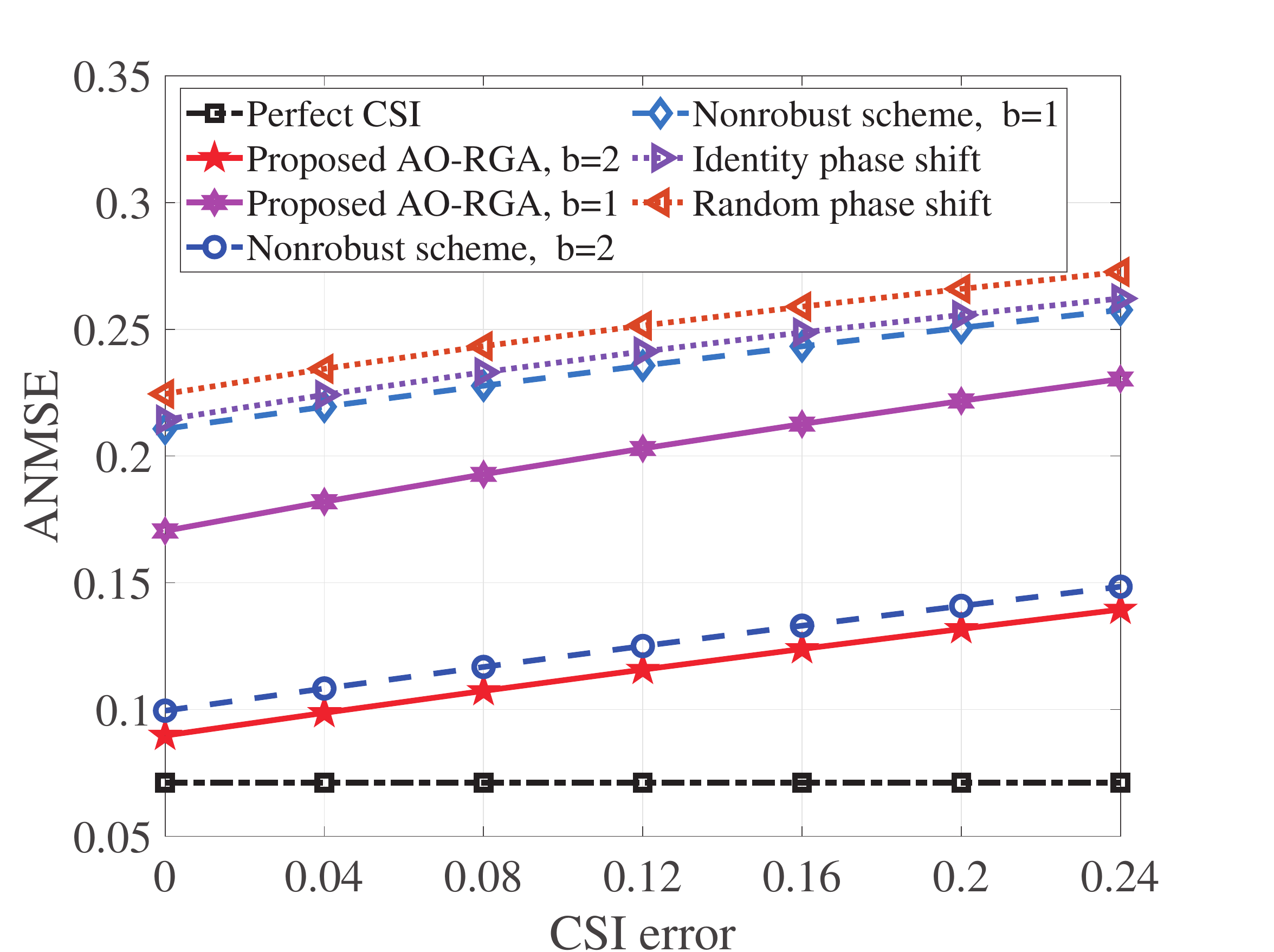}
  }
  \caption{ ANMSE performance comparison for different algorithms versus (a) the transmit power (b) the hardware distortions (c) the CSI errors.} \label{fig3}
  \end{figure}

        { Next, Fig.~\ref{fig3} shows the ANMSE performance of all studied algorithms versus transmit power, hardware distortion, and CSI error levels. }
        First, Fig.~\ref{fig3} (a) reveals the performance comparison between the proposed algorithms and the benchmark schemes with $b=1$. The perfect hardware and perfect CSI schemes serve as the lower bounds of the proposed MM-based robust designs. It is clear that the perfect hardware scheme performs much better than the perfect  CSI scheme, which suggests that the perfect hardware implementation of the transceiver and RIS is much more important than the accurate channel information in the system design. It follows from Fig.~\ref{Fig2} that the proposed AO-RGA and AO-MM algorithms are able to achieve almost the same ANMSE  performance, both of which outperform the nonrobust scheme in \cite{MSE6zhaoxin}.  This is because the nonrobust scheme directly aims at the optimization of the transceiver and RIS reflection matrix under the assumption of perfect CSI and ideal hardware, regardless of the practical non-negligible hardware impairments and  CSI errors, which thus leads to an inevitable ANMSE performance loss.  We also find that the schemes with random phase shifts and identity phase shifts both perform much worse than the proposed AO-RGA (AO-MM)  algorithm with {\color{black}$1$ bit} phase shift optimization, thereby demonstrating the necessity of optimizing the RIS reflection matrix and the superiority of the proposed AO algorithms.

        {
        Fig.~\ref{fig3} (b) illustrates the ANMSE performance versus hardware distortion level $\beta_R$ for different algorithm comparisons, where $\beta_T\!=\!0.08$.
        {The scheme of `No transceiver HWIs' considers the perfect transceiver hardware for the RIS-aided MIMO system, which serves as a benchmark.}
        Firstly, it is readily observed that the ANMSE performance decreases with the increasing of hardware distortion level $\beta_R$.  In addition, as $\beta_R$ increases, the performance gap between the proposed MM-based algorithm and the nonrobust algorithm increases, which shows the importance of taking the hardware impairments into consideration. Clearly, both the random phase shift and identity phase shift schemes perform much worse than the proposed algorithm, which highlights the necessity of optimizing the RIS reflection matrix.

        In Fig.~\ref{fig3} (c), we also plot the ANMSE performance for algorithms comparison under different CSI errors of the compound channel, i.e. $\sigma_m^2$, where $\sigma_d^2\!=\!0.01$. The `perfect CSI' scheme, where only the hardware impairments are considered in the RIS-aided system, acts as a benchmark for other algorithms and keeps constant as the CSI error increases. 
        As the CSI error increases, the system performance of the proposed algorithm and other schemes decreases gradually.
        Naturally, the nonrobust scheme performs worse than the proposed AO-MM or AO-RGA algorithm, since the impact of the CSI errors is ignored in the system design. It is readily observed that the proposed RIS reflection matrix design can bring huge performance gains compared with the random and identity phase shift schemes, which further demonstrates the effectiveness of the proposed RIS design strategy.}
      
\begin{figure}[t] 
  \centering
  \subfigbottomskip=0pt
  \subfigcapskip=0pt
  \subfigure[(a) Number of quantization bits]{
  \includegraphics[width=0.45\textwidth]{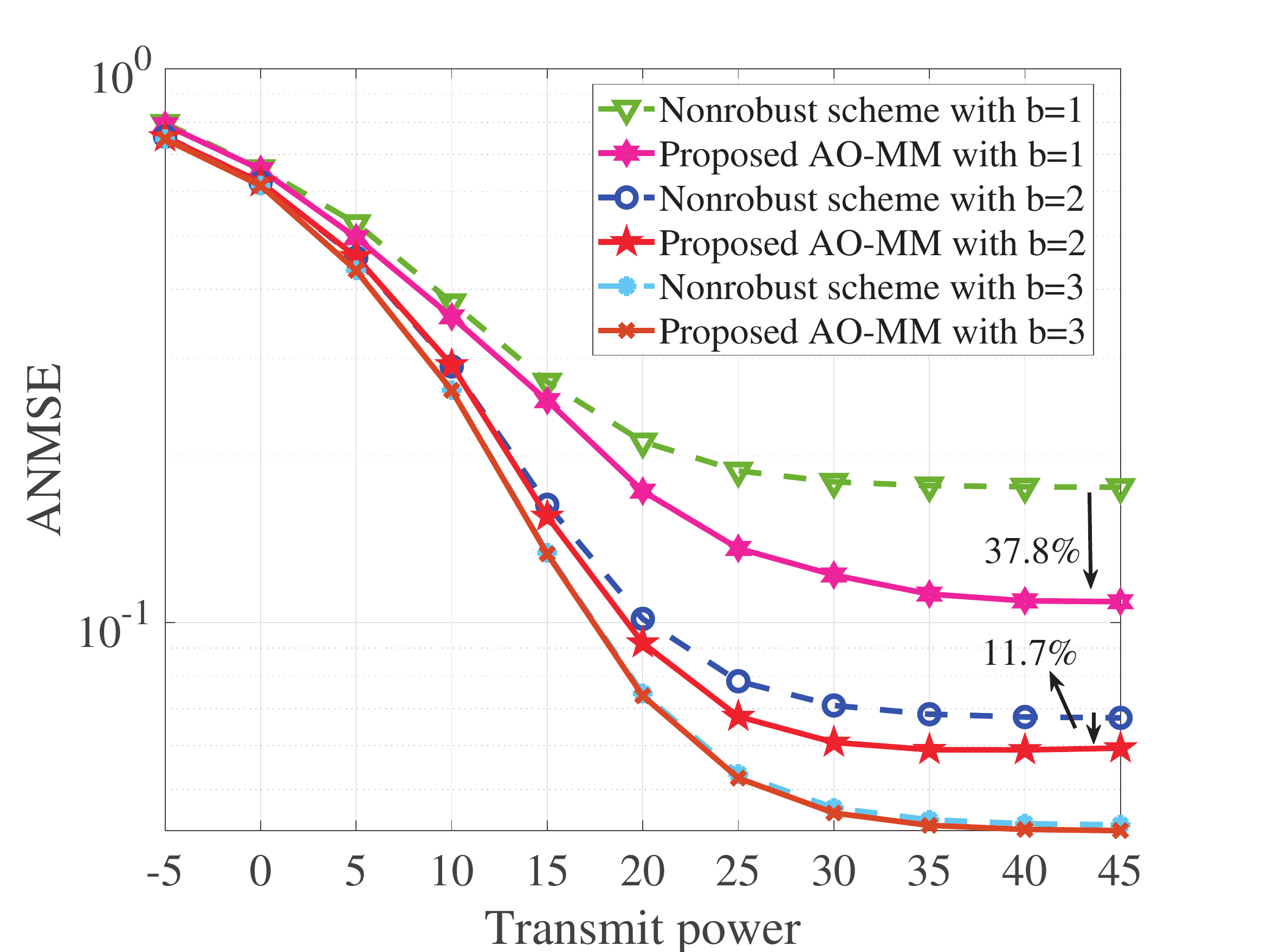}
  }
  \subfigure[(b) Hardware distoriton]{
  \includegraphics[width=0.45\textwidth]{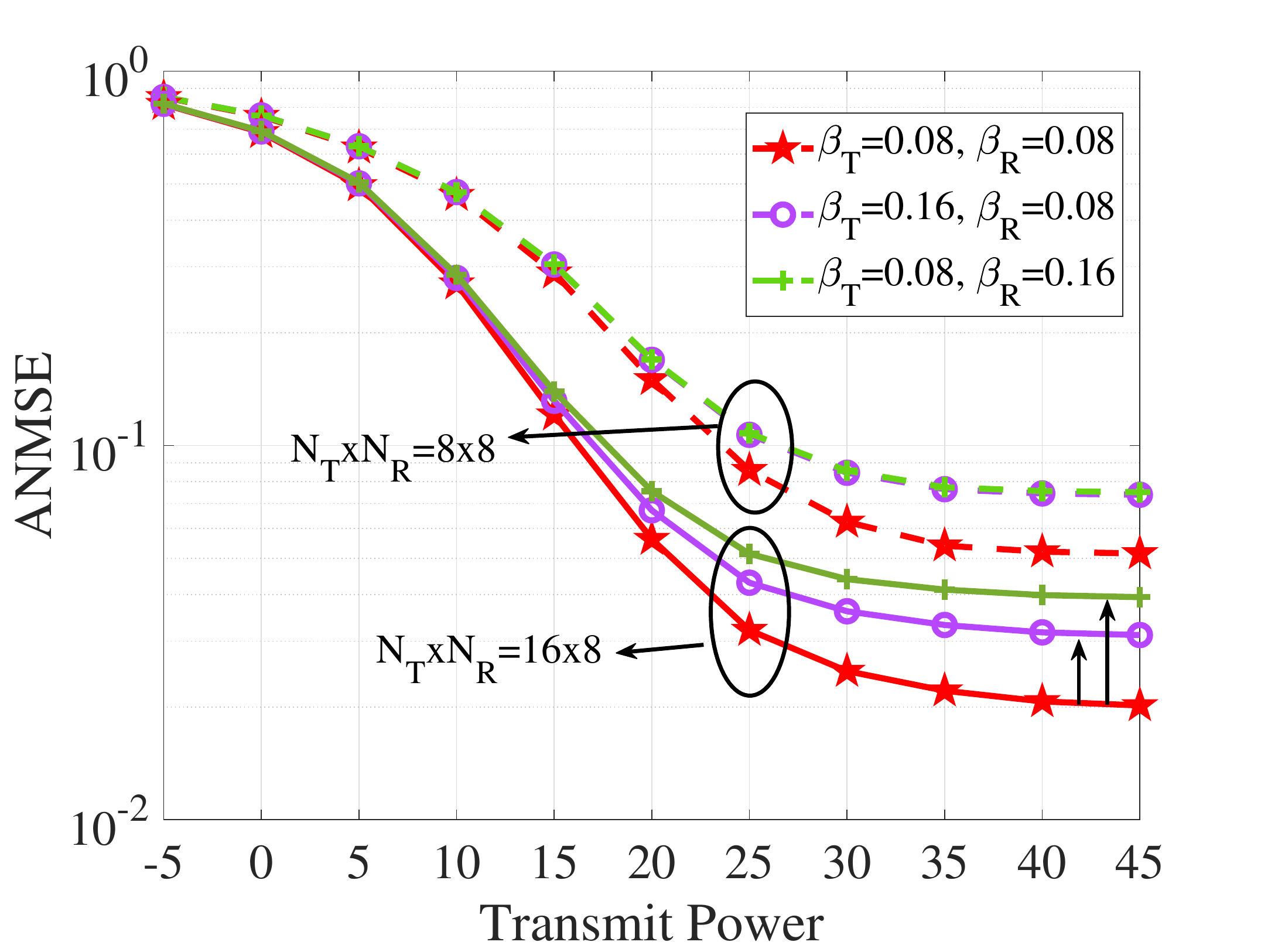}
  }
  \subfigure[(c) CSI error]{
  \includegraphics[width=0.45\textwidth]{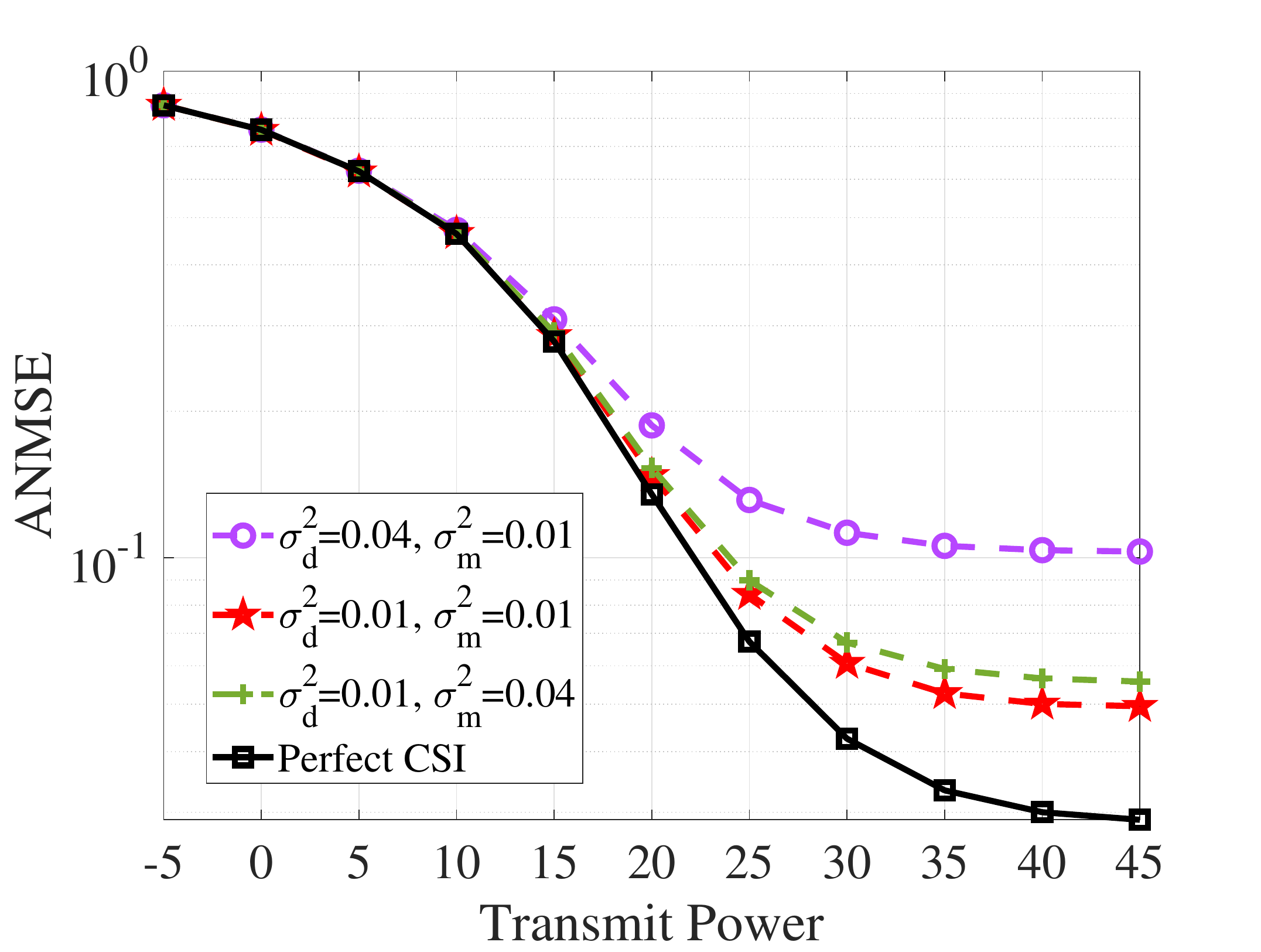}
  }
  \caption{ ANMSE performance comparison under different system imperfections: (a) the number of quantization bits (b) the hardware distortion levels (c) the CSI errors.} \label{differ_para}
  \vspace{-10pt}
  \end{figure}

        {
        Then, in Fig.~\ref{differ_para}, we study the impacts of system imperfections, i.e., the number of quantization bits, the transceiver hardware impairments, and the channel estimations errors, on the system ANMSE performance, respectively.
        To begin with, we compare the proposed AO-MM algorithm with the nonrobust scheme under different number of quantization bits in Fig.~\ref{differ_para} (a).} Obviously, the performance gap between the proposed AO-MM algorithm and the nonrobust scheme decreases from $37.8\%$ to $11.7\%$ as $b$ increases from 1 to 2, since the RIS phase shifts have more feasible values. This reminds us that using the proposed phase shift design strategy, low-resolution RIS phase shifts can achieve acceptable ANMSE performance in practical hardware implementation. Moreover, we readily find that an irreducible MSE floor exists in the high-SNR regime for each $b$, which is consistent with the theoretical results analyzed in Section IV.B. Obviously, increasing the number of quantization bits will decrease the value of the MSE floor, since the resolution of the RIS  phase shifts increases.

       {
        Next, Fig.~\ref{differ_para} (b)  reveals the effect of transceiver hardware distortion on the system ANMSE performance under different Tx-Rx antenna sizes,  represented by the dotted line (8x8 MIMO system) and solid line (16x8 MIMO system), respectively. It clearly shows that the ANMSE performance increases correspondingly with the increasing of transmit antenna size. Besides, the MSE floor is also influenced by the level of transceiver hardware impairments. On the other hand, we can readily conclude that when the number of transmit antennas is large than that of receiver antennas, the receiver's distortion will lead to a more severe impact on the system performance than the transmitter's distortion. Such influence will be further verified and quantized in the RIS-aided MISO system in Fig.~\ref{Fig5}.
        Finally, the impacts of channel estimation errors are illustrated in Fig.~\ref{differ_para} (c). From Fig.~\ref{differ_para} (c), it is readily inferred that the channel estimation inaccuracy clearly affects the system's performance and contributes to the performance floor in the high-SNR regime. Moreover, due to the double path loss effect of the RIS-related compound channels, the channel estimation inaccuracy of the direct channel induces a more severe influence on the system performance than the compound channels.}

        \begin{figure}[t] 
        \centering
            {\includegraphics[width=0.45\textwidth]{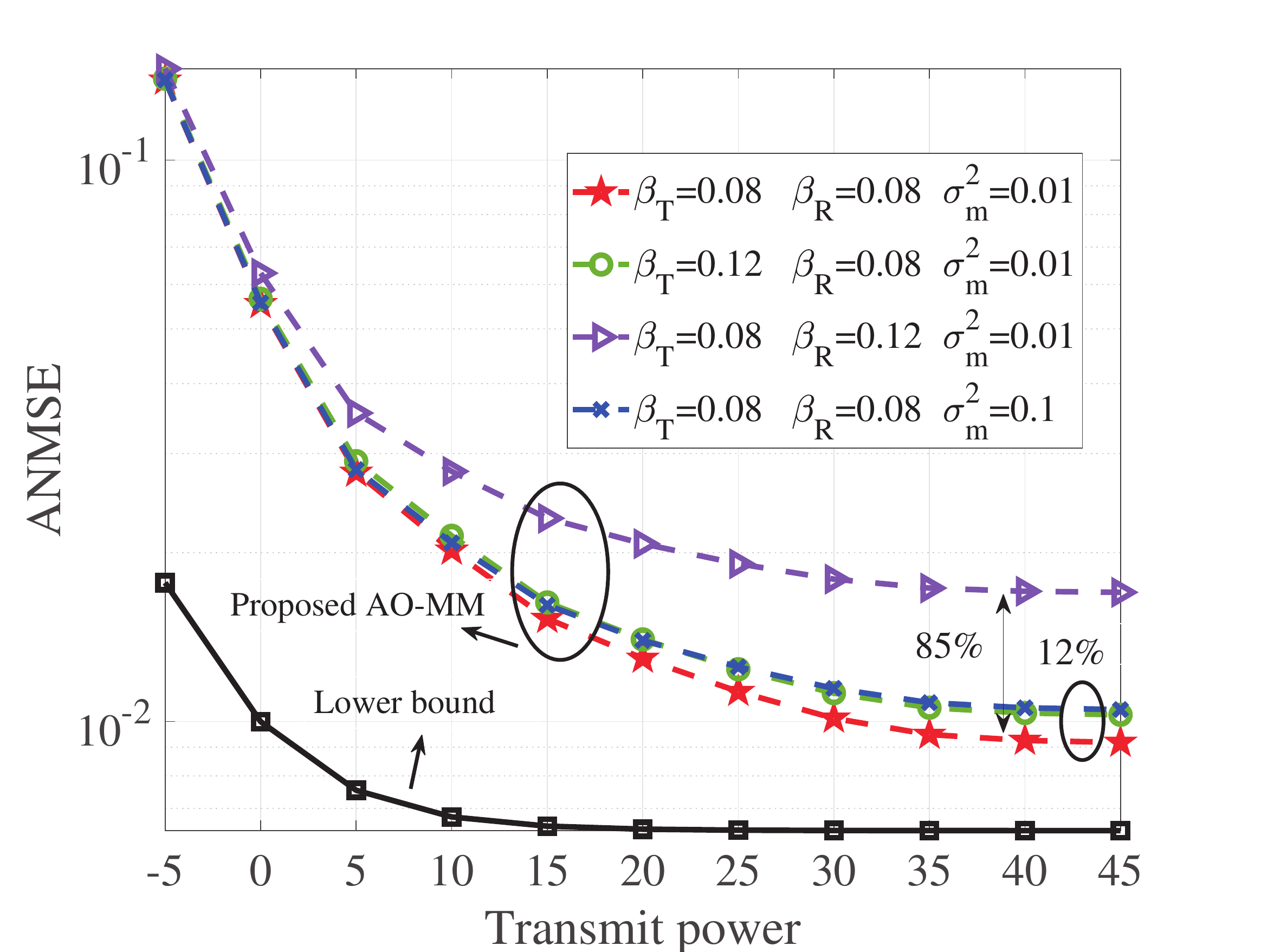}\label{lowerbound}} 
            \caption{ANMSE performance versus transmit power under different system parameter setups in the RIS-aided MISO system.}\label{Fig5}
           \vspace{-10pt}
        \end{figure}

        In order to validate the theoretical analysis in Section IV.B,
        Fig.~\ref{Fig5} shows the ANMSE performance of the proposed AO-MM algorithm versus the transmit power $P$ for the RIS-aided MISO system, where $ N_T\!=\!8$, $N_R\!=d\!=\!1$. Firstly, it can be seen that the performance of the proposed algorithm is lower bounded by  $f_{\rm lower}^{\rm opt}$ mentioned in Proposition~\ref{proposition1}, which is denoted as `Lower bound' in Fig.~\ref{Fig5}, for each SNR. 
        Note that the hardware distortions at the transmitter's and receiver's sides have different impacts on the system performance in the RIS-aided MISO case.
        {Specifically, the receiver distortion, i.e. $\beta_R$, causes more severe performance loss, about $70\%$, than transmitter distortion $\beta_T$, when the distortion level changes from 0.08 to 0.12.} This reminds us that more expensive hardware should be deployed on the user's side rather than the BS's side. 
        In addition, increasing the CSI errors of the compound channel, i.e. $\sigma_m^2$, also increases the MSE floor.
        { Therefore, combining the results of Fig.~\ref{differ_para} and Fig.~\ref{Fig5}, we conclude that the transceiver hardware distortions, CSI errors, and RIS phase noise all lead to the irreducible MSE floor in the RIS-aided MIMO/MISO communication system.}

    \begin{figure}[t] 
      \centering
      \subfigure[(a) Tx-Rx sizes]{
      \includegraphics[width=0.45\textwidth]{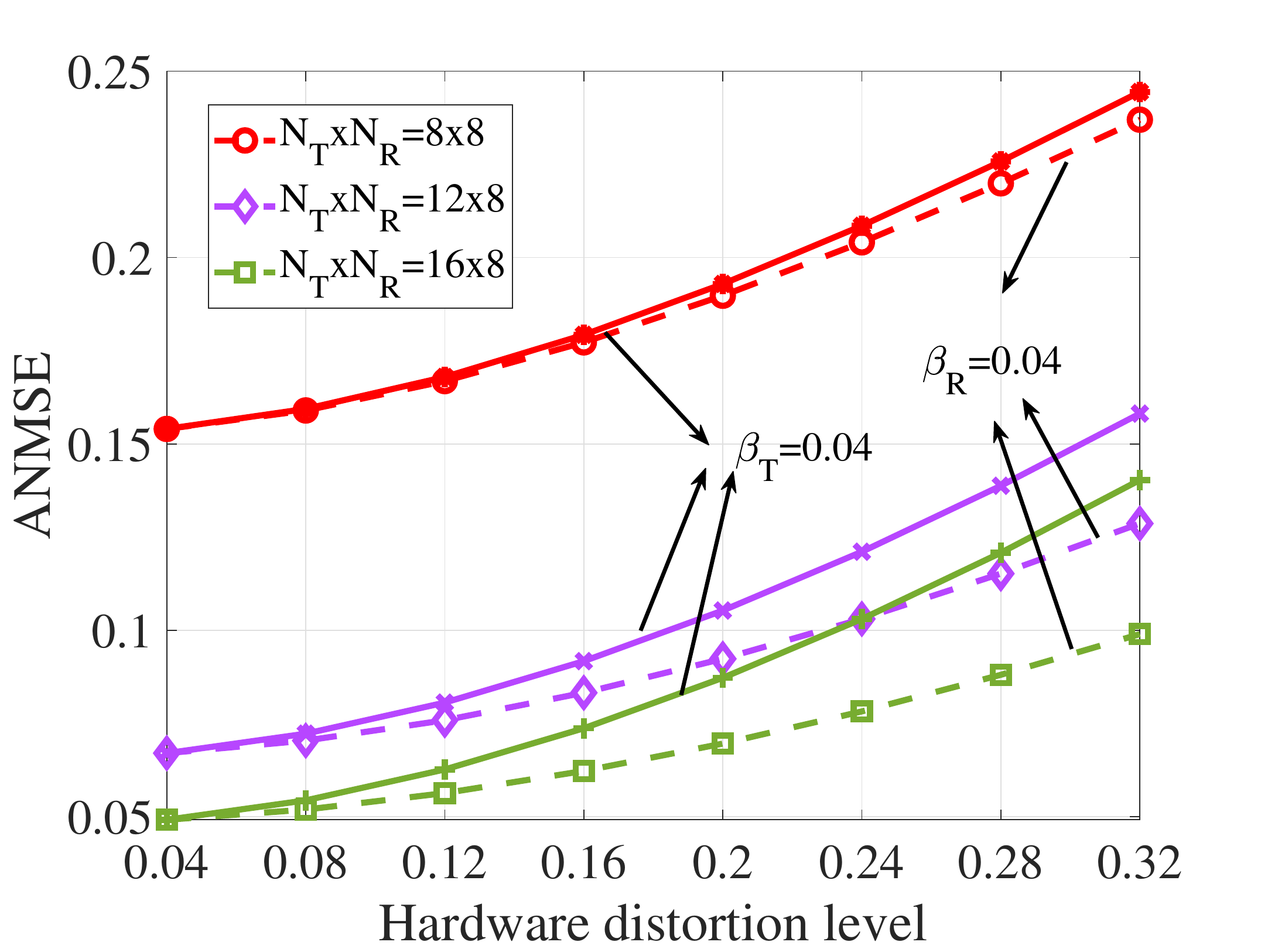}
      }
      \subfigure[(b) RIS elements]{
      \includegraphics[width=0.45\textwidth]{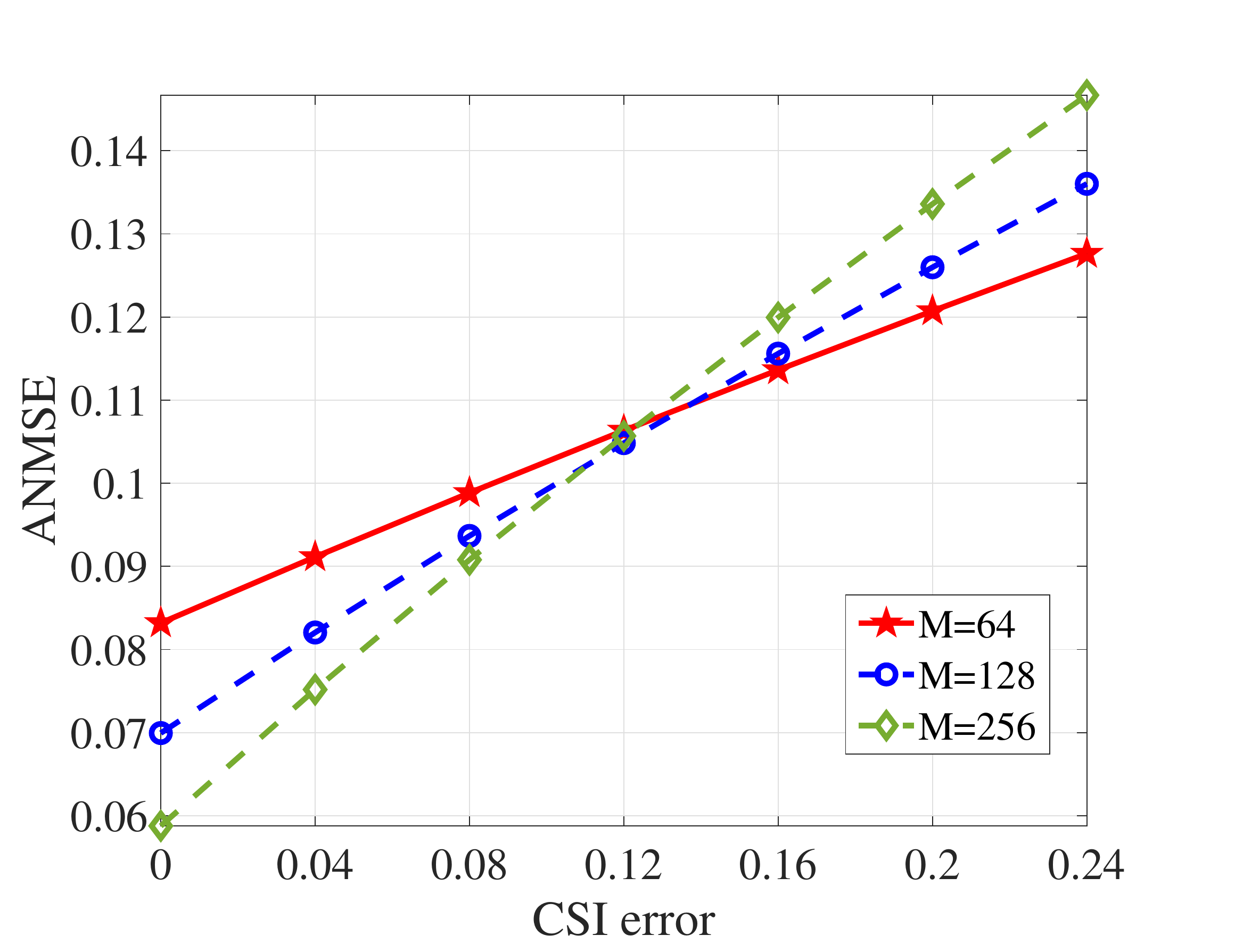}
      }
      \caption{ ANMSE performance versus the system imperfections under different number of (a) Tx-Rx sizes (b) RIS elements.}\label{differ_ant}
      \vspace{-10pt}
      \end{figure}
    { 
    In Fig.~\ref{differ_ant}, the ANMSE performance versus the channel estimation errors and transceiver hardware impairments under different numbers of transmitter's antennas and RIS elements has been investigated. 
    First, Fig.~\ref{differ_ant} (a) shows the system's performance versus the hardware distortion level $\beta_T$ or $\beta_R$ for different Tx-Rx configurations, represented by the dotted line and solid line, respectively. It is evident that as the size of the transmitter antenna increases, the performance gap brought by the transmitter's and receiver's distortion increases, which further demonstrates that the user's hardware has more impact in the system performance than that at BS in realistic communications.}
    Besides, it can be seen from Fig.~\ref{differ_ant} (b) that the ANMSE performance increases as the number of RIS elements $M$ increases for a small region $0\thicksim 0.12$. The reason is that  the RIS benefits from the diversity gain by increasing number of RIS elements. However, in the high CSI error regime, increasing the number of RIS elements may decrease the ANMSE performance, since the impact induced by imperfect RIS-related channels dominates the diversity gain. 
   As such, we readily conclude that cares must be taken w.r.t. the size of RIS in enhancing the system's performance.

\section{conclusion}
  In this work, we studied the robust transceiver and passive beamforming design for the RIS-aided MIMO system with the existence of hardware impairments at the transceiver and RIS and channel errors. 
  We aimed to minimize the total average MSE for all data streams subject to the total power constraint at BS and discrete phase constraints at RIS. To tackle this non-convex NP-hard optimization problem, an efficient AO-based algorithm  was proposed. Moreover, to handle the non-convex discrete phase constraints, we also proposed two methods, namely two-tier MM-based algorithm and modified RGA method, to obtain the sub-optimal solution of the RIS reflection matrix.
  Futhermore, we also studied the special cases of RIS-aided MIMO system, e.g., RIS-aided MIMO system without direct link and RIS-aided MISO system, to demonstrate the optimality of the proposed algorithm and the MSE performance floor in the high-SNR regime, respectively. 
   Simulation results demonstrated the superiority of our proposed algorithms compared to the nonrobust design and other benchmark schemes. It was also found that increasing the number of RIS elements was not always beneficial under severe CSI errors. These results provided useful insights for the implementation of practical RIS-aided MIMO systems.

  %

 
\appendices
\section{ \label{appendix_A}}
Recalling the lower bound of $f_{\rm MSE}({{\bf{w}},{\bm{\theta}}})$  in \eqref{f_lowerbou}, the transmit precoder design for $f_{\rm lower}({{\bf{w}},{\bm{\theta}}})$ is equivalently transformed into the the maximization of a generalized Rayleigh quotient as
{\begin{subequations}
  \begin{alignat}{2}
  (\text{P9}): \  \max_{{\bf{w}}} \quad  & \frac{{\bf w}^H  {\bf {\bar h}}{\bf {\bar h}}^H {\bf w}}   { {\bf w}^H \left( \frac{c_2}{P}(1\!\!+\!\!\beta_T^2){\bf I} \!\!+\!\!(1\!\!+\!\!\beta_R^2) {\bf {\bar h}}{\bf {\bar h}}^H \!\!+\!\! {\bf Q}  \right) {\bf w} } \nonumber \\
   \mbox{s.t.} \ 
  & (1\!+\!\beta_T^2)||{\bf{w}}||^2\leq P, \tag{42}
 \end{alignat}
\end{subequations}}
\noindent where ${\bf{Q}}$ has been defined in Proposition~\ref{proposition1}.
 Hence, the optimal transmit precoder ${\bf w}^{\star}$ can be derived as  
  \begin{equation}
  {\bf w}^{\star}\!\!=\!\!\sqrt{\frac{P}{1\!\!+\!\!\beta_T^2}} \frac{( \frac{c_2}{P}(1\!\!+\!\!\beta_T^2){\bf I} \!+\!(1\!+\!\beta_R^2) {\bf {\bar h}}{\bf {\bar h}}^H \!+\! {\bf Q}  )^{-1} {\bf {\bar h}}}{\left\lVert ( \frac{c_2}{P}(1+\beta_T^2){\bf I} \!+\!(1\!+\!\beta_R^2) {\bf {\bar h}}{\bf {\bar h}}^H \!+\! {\bf Q}  )^{-1} {\bf {\bar h}} \right\rVert}.       
\end{equation}

By substituting the optimal ${\bf w}^{\star}$ and relaxing the unit-modulus discrete phase constraints of \eqref{P1C2}, the optimization problem for the RIS reflection beamforming is further given by 
\begin{subequations}
  \begin{alignat}{2}
    (\text{P10}): 
    \max_{{\bm{\theta}}} \  & \frac  { {\bm{\tilde \theta}}^H {\bf{\tilde H}}_{\rm cat}^H (   {\bf{Q}}\!+\! (1\!+\!\beta_T^2) \frac{\sigma^2}{P} {\bf{I}}_{N_T} )^{-1} {\bf{\tilde H}}_{\rm cat}{\bm{\tilde \theta}} } {1 \!\!+\!\! (1\!\!+\!\!\beta_R^2){\bm{\tilde \theta}}^H {\bf{\tilde H}}_{\rm cat}^H (   {\bf{Q}} \!\!+\!\! (1\!\!+\!\! \beta_T^2) \frac{\sigma^2}{P} {\bf{I}}_{N_T} )^{-1} {\bf{\tilde H}}_{\rm cat}{\bm{\tilde \theta}} }, \nonumber \\
    \mbox{s.t.} \ &  ||{\bm{\tilde \theta}}||_2^2\leq M+1, \tag{44}
  \end{alignat}
\end{subequations}
  where ${\bf{\tilde H}}_{\rm cat}$ is defined in Proposition~\ref{proposition1}.
  Note that the objective function in Problem (P10) is monotonically increasing with the term of {{${{ {\bm{\tilde \theta}}^H {\bf{\tilde H}}_{\rm cat}^H (   {\bf{Q}}\!+\! (1\!+\!\beta_T^2) \frac{\sigma^2}{P} {\bf{I}}_{N_T} )^{-1} {\bf{\tilde H}}_{\rm cat}{\bm{\tilde \theta}} }}$}}.
  Therefore, the problem (P10) for the RIS phase shift design can be equivalently transformed to the maximization of {${ {\bm{\tilde \theta}}^H {\bf{\tilde H}}_{\rm cat}^H (   {\bf{Q}}\!+\! (1\!+\!\beta_T^2) \frac{\sigma^2}{P} {\bf{I}}_{N_T} )^{-1} {\bf{\tilde H}}_{\rm cat}{\bm{\tilde \theta}} }$}.             
  It is readily inferred that maximum value can be obtained as {$(M\!+\!1)\lambda_{\rm max}({\bf {\tilde H}}_{\rm cat}^H ( {\bf{Q}}\!+\! (1\!+\!\beta_T^2) \frac{\sigma^2}{P} {\bf{I}}_{N_T} )^{-1} {\bf {\tilde H}}_{\rm cat}) $} when {${\bm{\theta}}\!=\!\sqrt{(M+1)}{\bf u}_{\rm max}$} by leveraging the Cauchy-Schwarz inequality, where ${\bf u}_{\rm max}$ denotes the corresponding eigenvector w.r.t the largest eigenvalue $\lambda_{\rm max}$ from the EVD of {${\bf {\tilde H}}_{\rm cat}^H ( {\bf{Q}}\!+\! (1\!+\!\beta_T^2) \frac{\sigma^2}{P} {\bf{I}}_{N_T} )^{-1} {\bf {\tilde H}}_{\rm cat}$}. Based on the above transformation and relaxation, the optimal value of $f_{\rm lower}({{\bf{w}},{\bm{\theta}}})$ is lower bounded by
    \begin{align}
    & f_{\rm lower}^{\rm opt} = \nonumber \\
    & \!\! 1 \!\!-\!\!  \frac{(M+1)\lambda_{\rm max}({\bf {\tilde H}}_{\rm cat}^H ( {\bf{Q}}\!+\! (1\!+\!\beta_T^2) \frac{\sigma^2}{P} {\bf{I}}_{N_T} )^{-1} {\bf {\tilde H}}_{\rm cat})}{1 \!\!+\!\! (1 \!\!+\!\! \beta_R^2)(M \!\!+\!\! 1)\lambda_{\rm max}({\bf {\tilde H}}_{\rm cat}^H ( {\bf{Q}}\!\!+\!\! (1\!\!+\!\!\beta_T^2) \frac{\sigma^2}{P} {\bf{I}}_{N_T})^{-1} {\bf {\tilde H}}_{\rm cat})}.
  \end{align}
  Moreover, in the high-SNR regime, i.e. $\frac{P}{\sigma^2} \to \infty$, the optimal value of $f_{\rm lower}^{\rm opt}$ is reduced to the MSE floor, expressed as 
  \begin{equation}
    f_{\rm floor}=1-\frac{(M+1)\lambda_{\rm max}({\bf {\tilde H}}_{\rm cat}^H  {\bf{Q}} ^{-1} {\bf {\tilde H}}_{\rm cat})}{1 \!+\!(1\!+\!\beta_R^2)(M\!+\!1)\lambda_{\rm max}({\bf {\tilde H}}_{\rm cat}^H  {\bf{Q}}^{-1} {\bf {\tilde H}}_{\rm cat})}.
  \end{equation}              
  Thus, we complete the proof.

\section{ \label{appendix_B}}
   To begin with, we generally characterize the relationship between surrogate function $g({\bf x};{\bf x}_{t})$ and original function $f({\bf x})$ in the MM technique. For any iteration, $g({\bf x};{\bf x}_{t})$ is the majoring function of $f({\bf x})$ at ${\bf x}_{t}$, which satisfies 
   1) $g({\bf x};{\bf x}_{t}) \leq f({\bf x}), \forall {\bf x} \in {\rm dom} f$  
   2)   $g({\bf x}_{t};{\bf x}_{t}) = f({\bf x}_{t})$,
   3) $\nabla_{{\bf x}} g({\bf x};{\bf x}_{t})|_{{\bf x}={\bf x}_{t}} = \nabla_{{\bf x}} f({\bf x})$, 
   4) $g({\bf x};{\bf x}_{t}) $ is continuous in both ${\bf x}$ and ${\bf x}_{t}$. 
   The first two conditions guarantee $g({\bf x};{\bf x}_{t})$ is a tight global lower bound of $ f({\bf x})$, while the last two 
   conditions guarantee convergence to a stationary solution.
   Recalling the surrogate function defined in \eqref{MMsub}, \eqref{thetaMMf} and \eqref{thetasubsub}, it is easily verified that they all satisfy these four conditions. 
   
   Specifically, in the optimization of the transmit precoder with RIS reflection matrix fixed at the $t$-th iteration, we have 
     \begin{subequations}
       \begin{align}
       &  {g_{\rm sub1}^{\rm Low}({\bf{W}};{\bf W}_{t} | {\bm{\theta}}_{t})} \leq  g_{\rm MSE}({\bf{W}},{\bm{\theta}}_{t}), 
       \label{28a} \\
       &{g_{\rm sub1}^{\rm Low}({\bf{W}}_{t};{\bf W}_{t}| {\bm{\theta}}_{t})} =  g_{\rm MSE}({\bf{W}}_{t},{\bm{\theta}}_{t}). \label{28b}  
       \end{align}
     \end{subequations}
   The update rule for the transmit precoder as proposed in Section III.A can be rewritten as 
   \begin{equation}\label{update_W_2}
     {\bf{W}}_{t+1}={\rm arg} \max_{{\bf{W}}} {g_{\rm sub1}^{\rm Low}({\bf{W}};{\bf W}_{t})} 
   \end{equation}
   where the closed form has been derived in \eqref{update_W}. Hence, the following relationship between the objective values holds: 
   \begin{align} \label{first}
        g_{\rm MSE}({\bf{W}}_{t},{\bm{\theta}}_{t}) 
        & \mathop   = \limits^{({d_1})} 
       {g_{\rm sub1}^{\rm Low}({\bf{W}}_{t};{\bf W}_{t}| {\bm{\theta}}_{t})} \nonumber \\
       & \mathop  \leq \limits^{({d_2})} 
       {g_{\rm sub1}^{\rm Low}({\bf{W}}_{t+1};{\bf W}_{t}, {\bm{\theta}}_{t})}
       \nonumber \\
       & \mathop  \leq \limits^{({d_3})}
       g_{\rm MSE}({\bf{W}}_{t+1},{\bm{\theta}}_{t}) 
   \end{align}
   where $(d_1)$ holds because of \eqref{28b}, $(d_2)$ holds since \eqref{update_W_2} is optimally solved and $(d_3)$ is due to \eqref{28a}. 

   In the optimization of the passive RIS using the two-tier MM-based algorithm with transmit precoder fixed, we have
   \begin{subequations}
     \begin{align}
     &{g_{\rm sub2}^{\rm Low} ({\bm{\theta}};{\bm{\theta}}_{t}|{\bf W}_{t}})\leq  g_{\rm MSE}({\bf{W}}_{t},{\bm{\theta}}), \label{31a} \\

     &{g_{\rm sub2}^{\rm Low}({\bm{\theta}}_{t};{\bm{\theta}}_{t}|{\bf W}_{t})} =  g_{\rm MSE}({\bf{W}}_{t},{\bm{\theta}}_{t}), \label{31b} \\
     
     &{{\tilde g}_{\rm sub2}^{\rm Low}({\bm{\theta}}; {\bm{\theta}}_{t,r}|{\bf W}_{t})} \leq  {g_{\rm sub2}^{\rm Low}({\bm{\theta}}; {\bm{\theta}}_{t} | {\bf W}_{t})}, \label{31c} \\
     &{{\tilde g}_{\rm sub2}^{\rm Low}({\bm{\theta}}_{t,r};{\bm{\theta}}_{t,r} | {\bf W}_{t})} =  {g_{\rm sub2}^{\rm Low}({\bm{\theta}}_{t,r}; {\bm{\theta}}_{t} | {\bf W}_{t})}, \label{31d} 
     \end{align}
   \end{subequations}
   where $r$ denotes the $r$-th iteration for updating ${\bm{\theta}}$ using the MM technique. Denoting the maximum number of iteration as $R$,  
   we have ${\bm{\theta}}_{t,0}\!=\!{\bm{\theta}}_{t}$ and  ${\bm{\theta}}_{t,R-1}\!=\!{\bm{\theta}}_{t+1}$.
   The update rule for the passive RIS in the two-tier MM-based algorithm is re-expressed as 
   \begin{equation} \label{update_theta_MM_2}
    \setlength\abovedisplayskip{3pt}
    \setlength\belowdisplayskip{3pt}
     {\bm{\theta}}_{t,r+1}={\rm arg} \max_{{\bm{\theta}}} ~{{\tilde g}_{\rm sub2}^{\rm Low}({\bm{\theta}}; {\bm{\theta}}_{t,r} | {\bf W}_{t})},
   \end{equation}
   which is essentially defined as \eqref{update_theta_MM}. 
   Then, the following relationship holds:
   \begin{subequations}
     \begin{align} \label{second}
      & g_{\rm MSE}({\bf{W}}_{t+1},{\bm{\theta}}_{t}) \nonumber \\
      & \mathop  = \limits^{({e_1})}
       {g_{\rm sub2}^{\rm Low}({\bm{\theta}}_{t}; {\bm{\theta}}_{t}|{\bf W}_{t+1})} 
       \mathop  = \limits^{({e_2})}
       {{\tilde g}_{\rm sub2}^{\rm Low}({\bm{\theta}}_{t,0}; {\bm{\theta}}_{t,0} | {\bf W}_{t+1} )} \nonumber  \\
       & \mathop  \leq \limits^{({e_3})}
       {{\tilde g}_{\rm sub2}^{\rm Low}({\bm{\theta}}_{t+1};{\bm{\theta}}_{t,R-2} | {\bf W}_{t+1} )} 
       \mathop  \leq \limits^{({e_4})}
       {g_{\rm sub2}^{\rm Low}({\bm{\theta}}_{t+1}; {\bm{\theta}}_{t} | {\bf W}_{t+1} )} \nonumber \\
       &\mathop  \leq \limits^{({e_5})} 
       g_{\rm MSE}({\bf{W}}_{t+1},{\bm{\theta}}_{t+1}),  \tag{52}
     \end{align}
    \end{subequations}
   where 
   $(e_1)$ and  $(e_2)$ hold due to the properties of \eqref{31b} and \eqref{31d}, respectively. $(e_3)$ holds because of the update rule in \eqref{update_theta_MM_2}, $(e_4)$ is due to the equation in \eqref{31c}, and $(e_5)$ follows \eqref{31a}.
   
   Similarly, in the optimization of the passive RIS using the modified RGA algorithm, the relationship of \eqref{31a} and \eqref{31b} still hold. However, the update rule is conducted by the Riemannian gradient ascent on the CCM space, and can be rewritten as 
   \begin{equation} \label{update_theta_MM_3}
    {\bm{\theta}}_{t+1}={\rm arg} \max_{{\bm{\theta}}} ~{{ g}_{\rm sub2}^{\rm Low}({\bm{\theta}}; {\bm{\theta}}_{t} | {\bf W}_{t})},
  \end{equation}
   where the monotonicity is guaranteed by the backtracking search method. Thus, we have the following relationship: 
   \begin{subequations}
   \begin{align} \label{third} 
    g_{\rm MSE}({\bf{W}}_{t+1},{\bm{\theta}}_{t}) 
     & \mathop  = \limits^{({f_1})} 
    {g_{\rm sub2}^{\rm Low}({\bm{\theta}}_{t};{\bm{\theta}}_{t}|{\bf W}_{t+1} ) } \nonumber \\
    & \mathop  \leq \limits^{({f_2})} 
    {g_{\rm sub2}^{\rm Low}({\bm{\theta}}_{t+1};{\bm{\theta}}_{t}|{\bf W}_{t+1} ) } \nonumber \\
    & \mathop  \leq \limits^{({f_3})}\!
    g_{\rm MSE}({\bf{W}}_{t+1},{\bm{\theta}}_{t+1}), \tag{54} 
  \end{align}
\end{subequations}
  where $(f_1)$ holds due to  \eqref{31b}, $(f_2)$ holds since \eqref{update_theta_MM_3} is optimally solved and $(f_3)$ is because of \eqref{31a}.
   Therefore, with the fact of \eqref{first} and \eqref{second} \eqref{third}, the objective value $g_{\rm MSE}({\bf{W}},{\bm{\theta}})$  is monotonically non-decreasing in the proposed AO-MM or AO-RGA algorithm and finally converges to a finite value due to it is bounded.
   Thus, we complete the proof for Proposition~\ref{proposition2}.

\ifCLASSOPTIONcaptionsoff
  \newpage
\fi



%
%



\bibliographystyle{IEEEtran}
\bibliography{IRS_MIMO}

\begin{thebibliography}{10}
\providecommand{\url}[1]{#1}
\csname url@samestyle\endcsname
\providecommand{\newblock}{\relax}
\providecommand{\bibinfo}[2]{#2}
\providecommand{\BIBentrySTDinterwordspacing}{\spaceskip=0pt\relax}
\providecommand{\BIBentryALTinterwordstretchfactor}{4}
\providecommand{\BIBentryALTinterwordspacing}{\spaceskip=\fontdimen2\font plus
\BIBentryALTinterwordstretchfactor\fontdimen3\font minus
  \fontdimen4\font\relax}
\providecommand{\BIBforeignlanguage}[2]{{%
\expandafter\ifx\csname l@#1\endcsname\relax
\typeout{** WARNING: IEEEtran.bst: No hyphenation pattern has been}%
\typeout{** loaded for the language `#1'. Using the pattern for}%
\typeout{** the default language instead.}%
\else
\language=\csname l@#1\endcsname
\fi
#2}}
\providecommand{\BIBdecl}{\relax}
\BIBdecl

\bibitem{UAV2021}
S.~Li, B.~Duo, M.~D. Renzo, M.~Tao, and X.~Yuan, ``Robust secure {UAV}
  communications with the aid of reconfigurable intelligent surfaces,''
  \emph{IEEE Trans. Wireless Commun.}, vol.~20, no.~10, pp. 6402--6417, 2021.

\bibitem{mmWave2021}
H.~Du, J.~Zhang, J.~Cheng, and B.~Ai, ``Millimeter wave communications with
  reconfigurable intelligent surfaces: Performance analysis and optimization,''
  \emph{IEEE Trans. Commun.}, vol.~69, no.~4, pp. 2752--2768, 2021.

\bibitem{secure2021}
S.~Hong, C.~Pan, H.~Ren, K.~Wang, K.~K. Chai, and A.~Nallanathan, ``Robust
  transmission design for intelligent reflecting surface-aided secure
  communication systems with imperfect cascaded {CSI},'' \emph{IEEE Trans.
  Wireless Commun.}, vol.~20, no.~4, pp. 2487--2501, 2021.

\bibitem{wpt2021}
H.~Yang, X.~Yuan, J.~Fang, and Y.-C. Liang, ``Reconfigurable intelligent
  surface aided constant-envelope wireless power transfer,'' \emph{IEEE Trans.
  Signal Process.}, vol.~69, pp. 1347--1361, 2021.

\bibitem{wu2021intelligent}
Q.~Wu, S.~Zhang, B.~Zheng, C.~You, and R.~Zhang, ``Intelligent reflecting
  surface aided wireless communications: A tutorial,'' \emph{IEEE Trans.
  Commun.}, 2021.

\bibitem{wu2019intelligent}
Q.~Wu and R.~Zhang, ``Intelligent reflecting surface enhanced wireless network
  via joint active and passive beamforming,'' \emph{IEEE Trans. Wireless
  Commun.}, vol.~18, no.~11, pp. 5394--5409, 2019.

\bibitem{transpowermini13}
Z.~Yang, W.~Xu, C.~Huang, J.~Shi, and M.~Shikh-Bahaei, ``Beamforming design for
  multiuser transmission through reconfigurable intelligent surface,''
  \emph{IEEE Trans. Commun.}, vol.~69, no.~1, pp. 589--601, 2021.

\bibitem{Multicell_2020_Pan}
C.~Pan, H.~Ren, K.~Wang, W.~Xu, M.~Elkashlan, A.~Nallanathan, and L.~Hanzo,
  ``Multicell {MIMO} communications relying on intelligent reflecting
  surfaces,'' \emph{IEEE Trans. Wireless Commun.}, vol.~19, no.~8, pp.
  5218--5233, 2020.

\bibitem{rate4}
K.~Xu, J.~Zhang, X.~Yang, S.~Ma, and G.~Yang, ``On the sum-rate of
  {RIS}-assisted {MIMO} multiple-access channels over spatially correlated
  rician fading,'' \emph{IEEE Trans. Commun.}, vol.~69, no.~12, pp. 8228--8241,
  2021.

\bibitem{rateZhangJun}
J.~Zhang, J.~Liu, S.~Ma, C.-K. Wen, and S.~Jin, ``Large system achievable rate
  analysis of {RIS}-assisted {MIMO} wireless communication with statistical
  {CSIT},'' \emph{IEEE Trans. Wireless Commun.}, vol.~20, no.~9, pp.
  5572--5585, 2021.

\bibitem{MSE6zhaoxin}
X.~Zhao, K.~Xu, S.~Ma, S.~Gong, G.~Yang, and C.~Xing, ``Joint transceiver
  optimization for {IRS}-aided {MIMO} communications,'' \emph{IEEE Trans.
  Commun.}, vol.~70, no.~5, pp. 3467--3482, 2022.

\bibitem{MSE_gong}
S.~Gong, C.~Xing, X.~Zhao, S.~Ma, and J.~An, ``Unified {IRS}-aided {MIMO}
  transceiver designs via majorization theory,'' \emph{IEEE Trans. Signal
  Process.}, vol.~69, pp. 3016--3032, 2021.

\bibitem{MSE9kaizhe}
K.~Xu, S.~Gong, M.~Cui, G.~Zhang, and S.~Ma, ``Statistically robust transceiver
  design for multi-{RIS} assisted multi-user {MIMO} systems,'' \emph{IEEE
  Commun. Lett.}, vol.~26, no.~6, pp. 1428--1432, 2022.

\bibitem{enereff7}
L.~You, J.~Xiong, D.~W.~K. Ng, C.~Yuen, W.~Wang, and X.~Gao, ``Energy
  efficiency and spectral efficiency tradeoff in {RIS}-aided multiuser {MIMO}
  uplink transmission,'' \emph{IEEE Trans. Signal Process.}, vol.~69, pp.
  1407--1421, 2021.

\bibitem{enereff12}
C.~Huang, A.~Zappone, G.~C. Alexandropoulos, M.~Debbah, and C.~Yuen,
  ``Reconfigurable intelligent surfaces for energy efficiency in wireless
  communication,'' \emph{IEEE Trans. Wireless Commun.}, vol.~18, no.~8, pp.
  4157--4170, 2019.

\bibitem{gong2020beamforming}
S.~Gong, Z.~Yang, C.~Xing, J.~An, and L.~Hanzo, ``Beamforming optimization for
  intelligent reflecting surface-aided {SWIPT} {IoT} networks relying on
  discrete phase shifts,'' \emph{IEEE Internet Things J.}, vol.~8, no.~10, pp.
  8585--8602, 2020.

\bibitem{Hua2021}
M.~Hua, Q.~Wu, D.~W.~K. Ng, J.~Zhao, and L.~Yang, ``Intelligent reflecting
  surface-aided joint processing coordinated multipoint transmission,''
  \emph{IEEE Trans. Commun.}, vol.~69, no.~3, pp. 1650--1665, 2021.

\bibitem{OFDM1}
Z.~He, H.~Shen, W.~Xu, and C.~Zhao, ``Low-cost passive beamforming for
  {RIS}-aided wideband {OFDM} systems,'' \emph{IEEE Wireless Commun. Lett.},
  vol.~11, no.~2, pp. 318--322, 2022.

\bibitem{yang2020intelligent}
Y.~Yang, B.~Zheng, S.~Zhang, and R.~Zhang, ``Intelligent reflecting surface
  meets {OFDM}: Protocol design and rate maximization,'' \emph{IEEE Trans.
  Commun.}, vol.~68, no.~7, pp. 4522--4535, 2020.

\bibitem{OFDM3}
C.~Pradhan, A.~Li, L.~Song, J.~Li, B.~Vucetic, and Y.~Li, ``Reconfigurable
  intelligent surface {(RIS)}-enhanced two-way {OFDM} communications,''
  \emph{IEEE Trans. Veh. Technol.}, vol.~69, no.~12, pp. 16\,270--16\,275,
  2020.

\bibitem{RF_imperfections_book_2008}
T.~Schenk, \emph{RF Imperfections in High-rate Wireless Systems}.\hskip 1em
  plus 0.5em minus 0.4em\relax Springer, Dordrecht, 2008.

\bibitem{JunJuan_2019}
J.~Feng, S.~Ma, S.~Aïssa, and M.~Xia, ``Two-way massive {MIMO} relaying
  systems with non-ideal transceivers: Joint power and hardware scaling,''
  \emph{IEEE Trans. Commun.}, vol.~67, no.~12, pp. 8273--8289, 2019.

\bibitem{residual_transmit_RF_impairments_2010}
C.~Studer, M.~Wenk, and A.~Burg, ``\BIBforeignlanguage{en}{{MIMO} transmission
  with residual transmit-{RF} impairments},'' in
  \emph{\BIBforeignlanguage{en}{Proc. Int. ITG Workshop Smart Antennas (WSA)}},
  Bremen, Germany, Feb. 2010, pp. 189--196.

\bibitem{how_much_HI_affect_2020}
A.-A.~A. Boulogeorgos and A.~Alexiou, ``How much do hardware imperfections
  affect the performance of reconfigurable intelligent surface-assisted
  systems?'' \emph{IEEE Open J. Commun. Soc.}, vol.~1, pp. 1185--1195, 2020.

\bibitem{khel2021effects}
A.~M.~T. Khel and K.~A. Hamdi, ``Effects of hardware impairments on
  {IRS}-enabled {MISO} wireless communication systems,'' \emph{IEEE Commun.
  Lett.}, vol.~26, no.~2, pp. 259--263, 2021.

\bibitem{saeidi2021weighted}
M.~A. Saeidi, M.~J. Emadi, H.~Masoumi, M.~R. Mili, D.~W.~K. Ng, and
  I.~Krikidis, ``Weighted sum-rate maximization for multi-{IRS}-assisted
  full-duplex systems with hardware impairments,'' \emph{IEEE Trans. Cogn.
  Commun. Netw.}, vol.~7, no.~2, pp. 466--481, 2021.

\bibitem{zhou2021secure}
G.~Zhou, C.~Pan, H.~Ren, K.~Wang, and Z.~Peng, ``Secure wireless communication
  in {RIS}-aided {MISO} system with hardware impairments,'' \emph{IEEE Wireless
  Commun. Lett.}, vol.~10, no.~6, pp. 1309--1313, 2021.

\bibitem{hong_shen_beamforming_2021}
H.~Shen, W.~Xu, S.~Gong, C.~Zhao, and D.~W.~K. Ng, ``Beamforming optimization
  for {IRS}-aided communications with transceiver hardware impairments,''
  \emph{IEEE Trans. Commun.}, vol.~69, no.~2, pp. 1214--1227, 2021.

\bibitem{Spectral_and_Energy_Efficiency_2020}
S.~Zhou, W.~Xu, K.~Wang, M.~Di~Renzo, and M.-S. Alouini, ``Spectral and energy
  efficiency of {IRS}-assisted {MISO} communication with hardware
  impairments,'' \emph{IEEE Wireless Commun. Lett.}, vol.~9, no.~9, pp.
  1366--1369, 2020.

\bibitem{phaseerror_explain}
M.-A. Badiu and J.~P. Coon, ``Communication through a large reflecting surface
  with phase errors,'' \emph{IEEE Wireless Commun. Lett.}, vol.~9, no.~2, pp.
  184--188, 2020.

\bibitem{MSEphasenoise2022}
J.~Zhao, M.~Chen, C.~Pan, Z.~Li, G.~Zhou, and X.~Chen, ``{MSE}-based
  transceiver designs for {RIS}-aided communications with hardware
  impairments,'' \emph{IEEE Commun. Lett.}, pp. 1--1, 2022.

\bibitem{liu2020energy}
Y.~Liu, E.~Liu, and R.~Wang, ``Energy efficiency analysis of intelligent
  reflecting surface system with hardware impairments,'' in \emph{Proc. IEEE
  Global Commun. Conf. (GLOBECOM)}, Taipei, Taiwan, 2020, pp. 1--6.

\bibitem{phasenoise2}
Z.~Xing, R.~Wang, J.~Wu, and E.~Liu, ``Achievable rate analysis and phase shift
  optimization on intelligent reflecting surface with hardware impairments,''
  \emph{IEEE Trans. Wireless Commun.}, vol.~20, no.~9, pp. 5514--5530, 2021.

\bibitem{chu2022ris}
Z.~Chu, J.~Zhong, P.~Xiao, D.~Mi, W.~Hao, R.~Tafazolli, and A.~P. Feresidis,
  ``{RIS} assisted wireless powered {IoT} networks with phase shift error and
  transceiver hardware impairment,'' \emph{IEEE Trans. Commun.}, vol.~70,
  no.~7, pp. 4910--4924, 2022.

\bibitem{Dai_Jianxin_MU_MISO_2021}
J.~Dai, F.~Zhu, C.~Pan, H.~Ren, and K.~Wang, ``Statistical {CSI}-based
  transmission design for reconfigurable intelligent surface-aided massive
  {MIMO} systems with hardware impairments,'' \emph{IEEE Wireless Commun.
  Lett.}, pp. 1--1, 2021.

\bibitem{HIandCSIerror2022}
Z.~Peng, Z.~Chen, C.~Pan, G.~Zhou, and H.~Ren, ``Robust transmission design for
  {RIS}-aided communications with both transceiver hardware impairments and
  imperfect {CSI},'' \emph{IEEE Wireless Commun. Lett.}, vol.~11, no.~3, pp.
  528--532, 2022.

\bibitem{liu2020beamforming}
Y.~Liu, E.~Liu, and R.~Wang, ``Beamforming and performance evaluation for
  intelligent reflecting surface aided wireless system with hardware
  impairments,'' \emph{arXiv preprint arXiv:2006.00664}, 2020.

\bibitem{papazafeiropoulos2021intelligent}
A.~Papazafeiropoulos, C.~Pan, P.~Kourtessis, S.~Chatzinotas, and J.~M. Senior,
  ``Intelligent reflecting surface-assisted {MU-MISO} systems with imperfect
  hardware: Channel estimation and beamforming design,'' \emph{IEEE Trans.
  Wireless Commun.}, vol.~21, no.~3, pp. 2077--2092, 2021.

\bibitem{Xing_obj_2021}
C.~Xing, S.~Wang, S.~Chen, S.~Ma, H.~V. Poor, and L.~Hanzo, ``Matrix-monotonic
  optimization-{Part} {I}: Single-variable optimization,'' \emph{IEEE Trans.
  Signal Process.}, vol.~69, pp. 738--754, 2021.

\bibitem{energy_efficiency_HI_2020_Oluwatayo}
O.~Y. Kolawole, S.~Biswas, K.~Singh, and T.~Ratnarajah, ``Transceiver design
  for energy-efficiency maximization in mmwave {MIMO} {IoT} networks,''
  \emph{IEEE Trans. Green Commun. Netw.}, vol.~4, no.~1, pp. 109--123, 2020.

\bibitem{Capacity_Emil_Bjornson_2013}
E.~Bjornson, P.~Zetterberg, M.~Bengtsson, and B.~Ottersten, ``Capacity limits
  and multiplexing gains of {MIMO} channels with transceiver impairments,''
  \emph{IEEE Commun. Lett.}, vol.~17, no.~1, pp. 91--94, 2013.

\bibitem{WuXianda_millimeter}
X.~Wu, S.~Ma, and X.~Yang, ``Tensor-based low-complexity channel estimation for
  mmwave massive {MIMO-OTFS} systems,'' \emph{J. Commun. Netw.}, vol.~5, no.~3,
  pp. 324--334, 2020.

\bibitem{youIntelligentReflectingSurface2020}
C.~You, B.~Zheng, and R.~Zhang, ``Intelligent reflecting surface with discrete
  phase shifts: Channel estimation and passive beamforming,'' in \emph{{{ICC}}
  2020 - 2020 {{IEEE Int. Conf. Commun.}} ({{ICC}})}, pp. 1--6.

\bibitem{wang2014outage}
K.-Y. Wang, A.~M.-C. So, T.-H. Chang, W.-K. Ma, and C.-Y. Chi, ``Outage
  constrained robust transmit optimization for multiuser {MISO} downlinks:
  Tractable approximations by conic optimization,'' \emph{IEEE Trans. Signal
  Process.}, vol.~62, no.~21, pp. 5690--5705, 2014.

\bibitem{zhang2020robust}
J.~Zhang, Y.~Zhang, C.~Zhong, and Z.~Zhang, ``Robust design for intelligent
  reflecting surfaces assisted {MISO} systems,'' \emph{IEEE Commun. Lett.},
  vol.~24, no.~10, pp. 2353--2357, 2020.

\bibitem{MM}
Y.~Sun, P.~Babu, and D.~P. Palomar, ``Majorization-minimization algorithms in
  signal processing, communications, and machine learning,'' \emph{IEEE Trans.
  Signal Process.}, vol.~65, no.~3, pp. 794--816, 2017.

\bibitem{CCM2019}
K.~Alhujaili, V.~Monga, and M.~Rangaswamy, ``Transmit {MIMO} radar beampattern
  design via optimization on the complex circle manifold,'' \emph{IEEE Trans.
  Signal Process.}, vol.~67, no.~13, pp. 3561--3575, 2019.

\bibitem{cho2010mimo}
Y.~S. Cho, J.~Kim, W.~Y. Yang, and C.~G. Kang, \emph{MIMO-OFDM wireless
  communications with MATLAB}.\hskip 1em plus 0.5em minus 0.4em\relax John
  Wiley \& Sons, 2010.

\bibitem{zhangMIMOCapacityResidual2014}
X.~Zhang, M.~Matthaiou, E.~Björnson, M.~Coldrey, and M.~Debbah, ``On the
  {{MIMO}} capacity with residual transceiver hardware impairments,'' in
  \emph{2014 {{IEEE Int. Conf. Commun.}} ({{ICC}})}, pp. 5299--5305.

\bibitem{bjornsonMassiveMIMOSystems2014}
E.~Björnson, J.~Hoydis, M.~Kountouris, and M.~Debbah, ``Massive {MIMO} systems
  with non-ideal hardware: Energy efficiency, estimation, and capacity
  limits,'' \emph{IEEE Trans. Inf. Theory}, vol.~60, no.~11, pp. 7112--7139,
  2014.

\bibitem{holma2011lte}
H.~Holma and A.~Toskala, \emph{LTE for UMTS: Evolution to LTE-advanced}.\hskip
  1em plus 0.5em minus 0.4em\relax John Wiley \& Sons, 2011.

\end{thebibliography}

\begin{IEEEbiography}[{\includegraphics[width=1in,height=1.25in,clip,keepaspectratio]{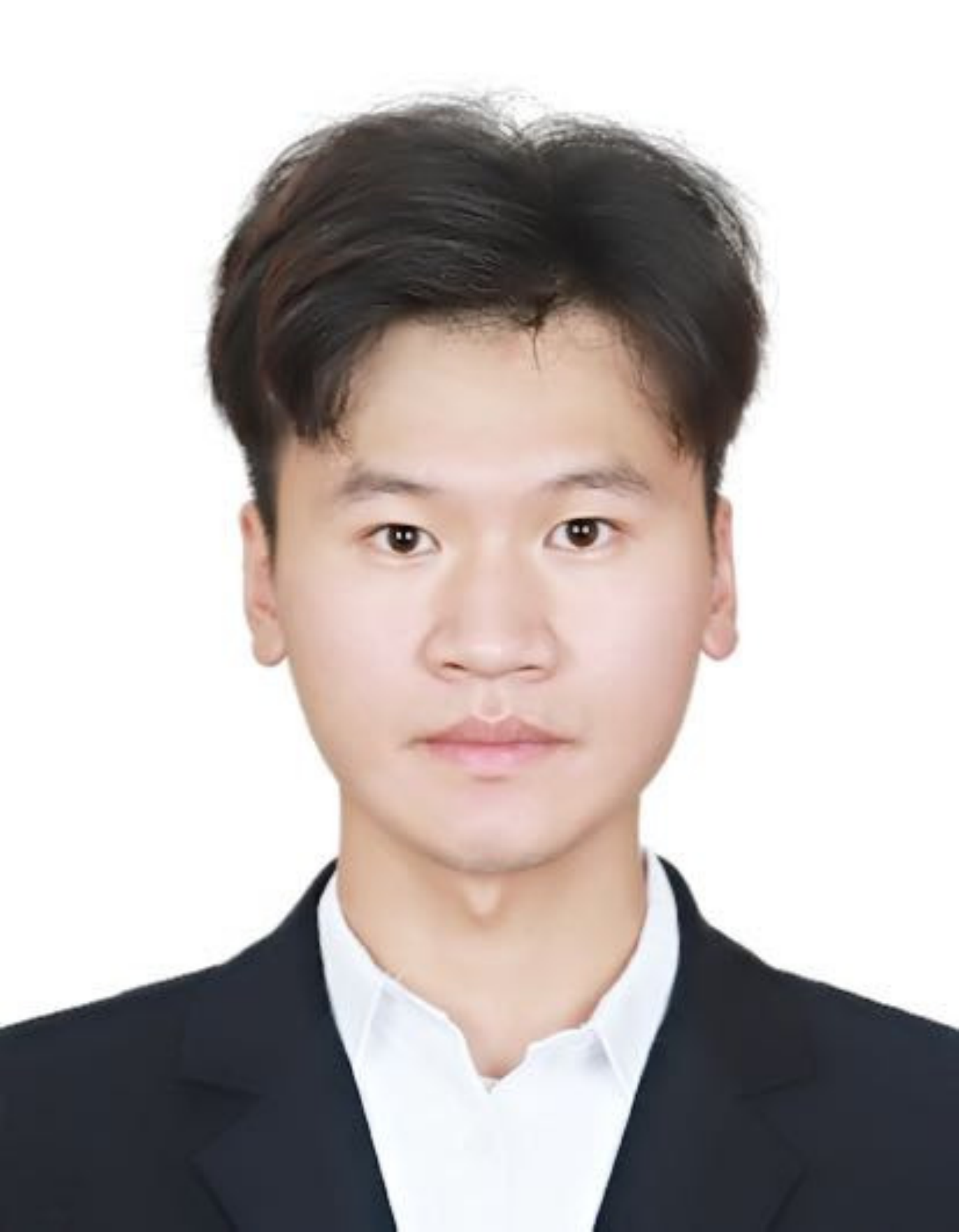}}]{Jintao Wang} received the B.S. degree in 2020 in communication engineering from Jilin University, Changchun, China. He is currently pursuing the Ph.D. degree with the State Key Laboratory of Internet of Things for Smart City and the Department of Electrical and Computer Engineering, University of Macau, Macau, China. His main research interests include RIS-aided communication, mmWave communication, transceiver design and convex optimization.
\end{IEEEbiography}

\begin{IEEEbiography}[{\includegraphics[width=1in,height=1.25in,clip,keepaspectratio]{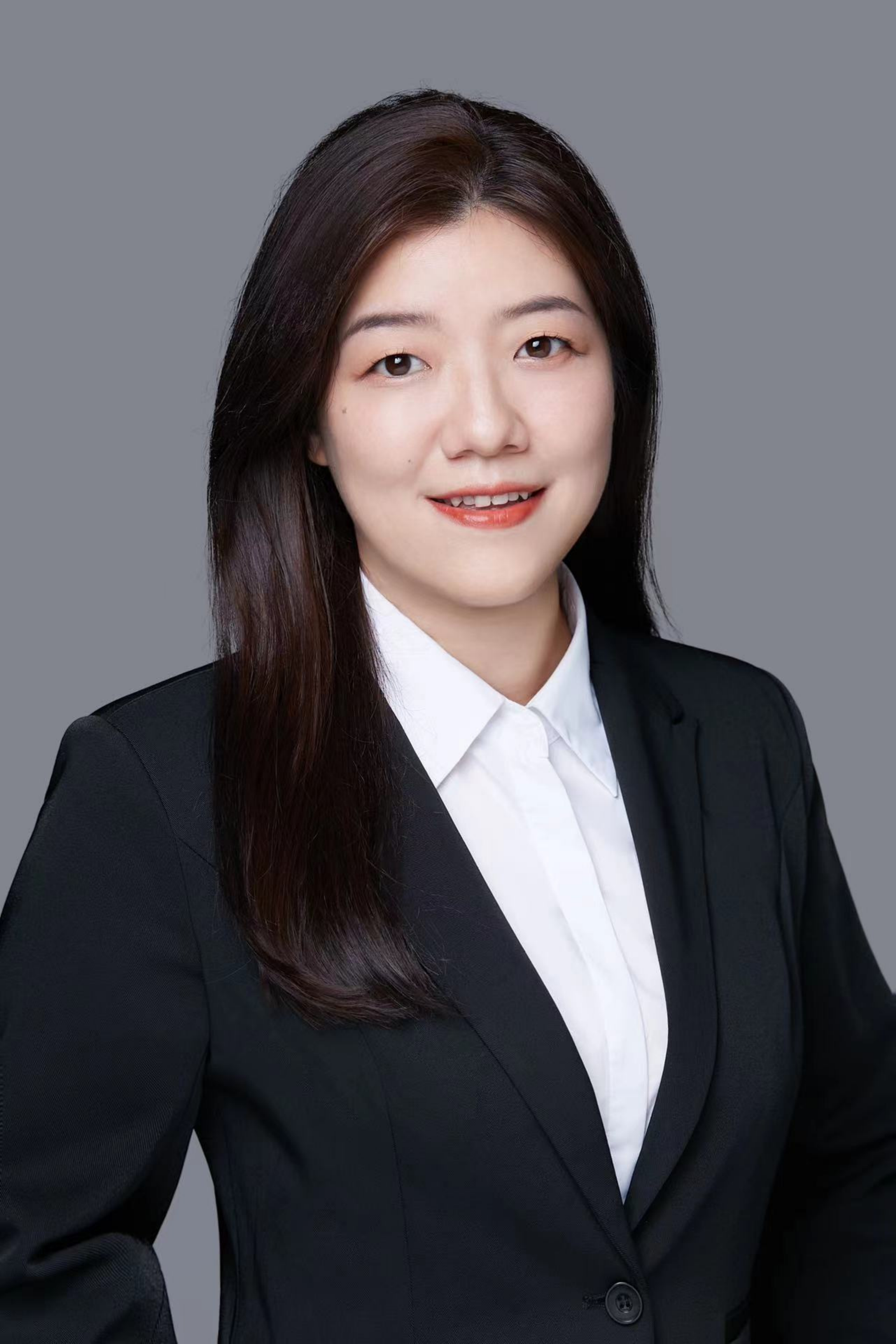}}]{Shiqi Gong} received the B.S. and Ph.D. degrees in electronic engineering from Beijing Institute of Technology, China, in 2014 and 2020, respectively. From January 2021 to August 2022, she was a postdoctoral fellow with the State Key Laboratory of Internet of Things for Smart City, University of Macau, China. She is currently an associate professor with the School of Cyberspace Science and Technology, Beijing Institute of Technology. Her research interests are in the areas of signal processing, mmWave and Terahertz communications and convex optimization. She was a recipient of the Best Ph.D. Thesis Award of Beijing Institute of Technology in 2020.
\end{IEEEbiography}

\begin{IEEEbiography}[{\includegraphics[width=1in,height=1.25in,clip,keepaspectratio]{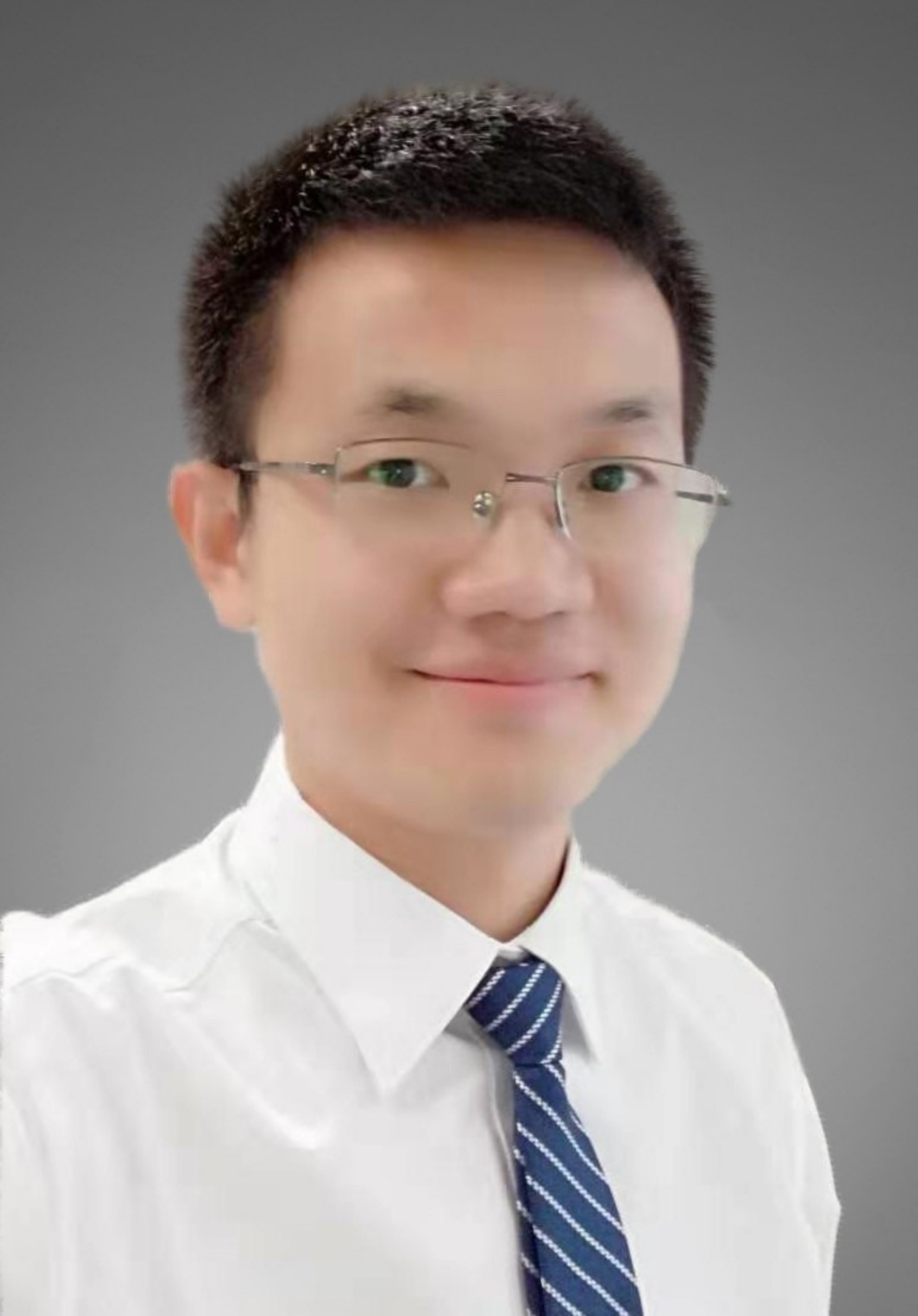}}] {Qingqing Wu} (S'13-M'16-SM'21) received the B.Eng. and the Ph.D. degrees in Electronic Engineering from South China University of Technology and Shanghai Jiao Tong University (SJTU) in 2012 and 2016, respectively. From 2016 to 2020, he was a Research Fellow in the Department of Electrical and Computer Engineering at National University of Singapore. He is currently an Associate Professor with Shanghai Jiao Tong University. His current research interest includes intelligent reflecting surface (IRS), unmanned aerial vehicle (UAV) communications, and MIMO transceiver design. He has coauthored more than 100 IEEE journal papers with 26 ESI highly cited papers and 8 ESI hot papers, which have received more than 18,000 Google citations. He was listed as the Clarivate ESI Highly Cited Researcher in 2022 and 2021, the Most Influential Scholar Award in AI-2000 by Aminer in 2021 and World's Top 2\% Scientist by Stanford University in 2020 and 2021.

  He was the recipient of the IEEE Communications Society Asia Pacific Best Young Researcher Award and Outstanding Paper Award in 2022, the IEEE Communications Society Young Author Best Paper Award in 2021, the Outstanding Ph.D. Thesis Award of China Institute of Communications in 2017, the Outstanding Ph.D. Thesis Funding in SJTU in 2016, the IEEE ICCC Best Paper Award in 2021, and IEEE WCSP Best Paper Award in 2015. He was the Exemplary Editor of IEEE Communications Letters in 2019 and the Exemplary Reviewer of several IEEE journals. He serves as an Associate Editor for IEEE Transactions on Communications, IEEE Communications Letters, IEEE Wireless Communications Letters, IEEE Open Journal of Communications Society (OJ-COMS), and IEEE Open Journal of Vehicular Technology (OJVT). He is the Lead Guest Editor for IEEE Journal on Selected Areas in Communications on ``UAV Communications in 5G and Beyond Networks", and the Guest Editor for IEEE OJVT on “6G Intelligent Communications" and IEEE OJ-COMS on “Reconfigurable Intelligent Surface-Based Communications for 6G Wireless Networks". He is the workshop co-chair for IEEE ICC 2019-2022 workshop on “Integrating UAVs into 5G and Beyond”, and the workshop co-chair for IEEE GLOBECOM 2020 and ICC 2021 workshop on “Reconfigurable Intelligent Surfaces for Wireless Communication for Beyond 5G”. He serves as the Workshops and Symposia Officer of Reconfigurable Intelligent Surfaces Emerging Technology Initiative and Research Blog Officer of Aerial Communications Emerging Technology Initiative. He is the IEEE Communications Society Young Professional Chair in Asia Pacific Region.
  
\end{IEEEbiography}

\begin{IEEEbiography}[{\includegraphics[width=1in,height=1.25in,clip,keepaspectratio]{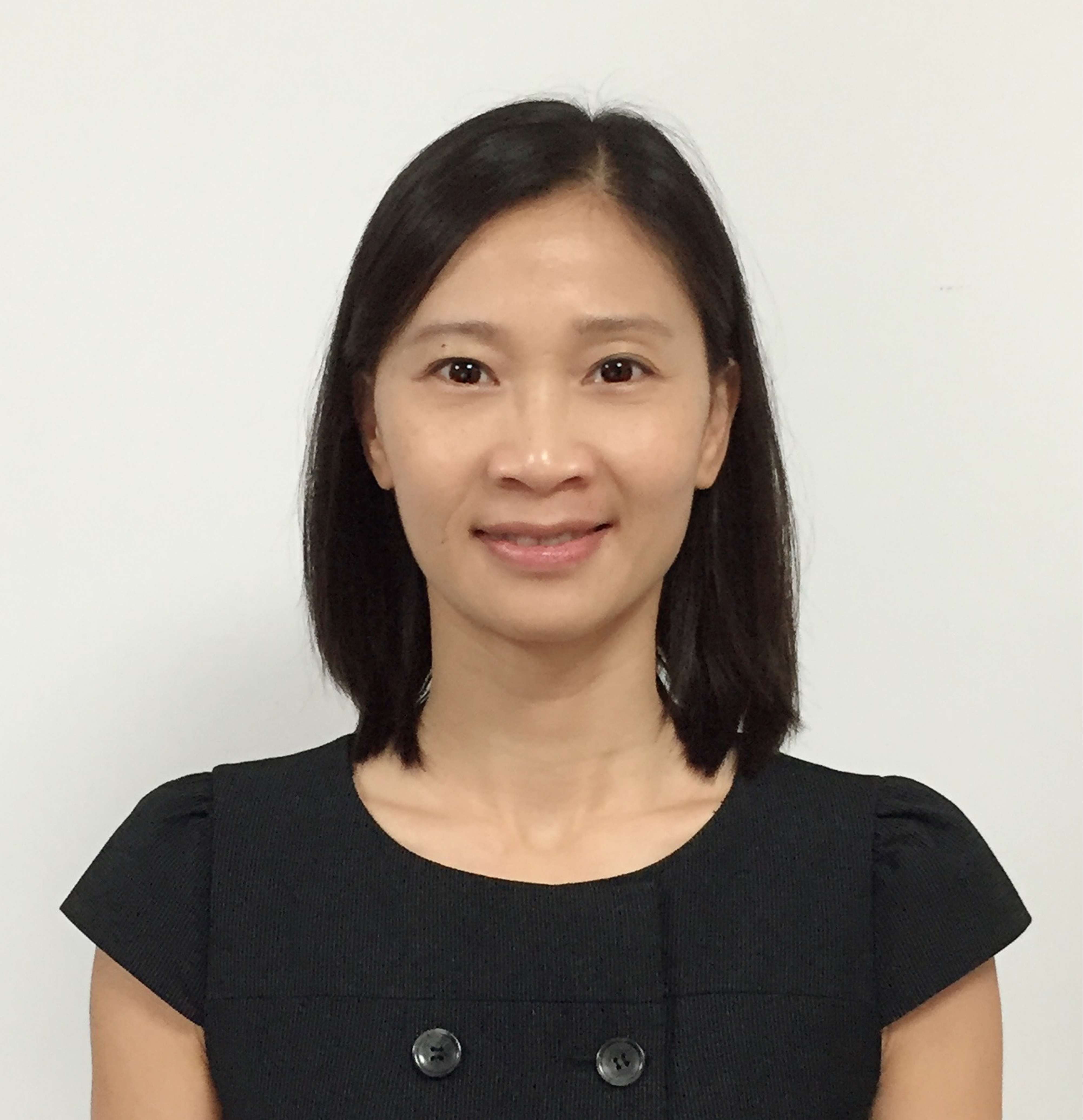}}]{Shaodan Ma}(Senior Member, IEEE) received the double Bachelor's degrees in science and economics and the M.Eng. degree in electronic engineering from Nankai University, Tianjin, China, in 1999 and 2002, respectively, and the Ph.D. degree in electrical and electronic engineering from The University of Hong Kong, Hong Kong, in 2006. From 2006 to 2011, she was a post-doctoral fellow at The University of Hong Kong. Since August 2011, she has been with the University of Macau, where she is currently a Professor. Her research interests include array signal processing, transceiver design, localization, mmWave communications and massive MIMO. She was a symposium co-chair for various conferences including IEEE ICC 2021, 2019 and 2016, IEEE/CIC ICCC 2019, IEEE GLOBECOM 2016, etc. Currently she serves as an Editor for IEEE Transactions on Wireless Communications, IEEE Transactions on Communications, IEEE Communications Letters, and Journal of Communications and Information Networks.
\end{IEEEbiography}

\end{document}